\documentclass[preprint]{elsarticle} 
\usepackage{amsmath}
\usepackage{wrapfig}
\usepackage{amsthm}
\usepackage{amssymb}
\usepackage{stmaryrd}
\usepackage{thmtools}
\usepackage{fullpage}
\usepackage{epsfig}
\usepackage{bbold}
\usepackage[T1]{fontenc}
\usepackage{multicol}
\usepackage{MnSymbol}
\usepackage[usenames,dvipsnames,table]{xcolor}
\usepackage{float}
\usepackage{hyperref} 
\definecolor{linkcolour}{rgb}{0,0.2,0.6} 
\hypersetup{colorlinks,breaklinks,urlcolor=linkcolour,linkcolor=linkcolour} 

\usepackage{graphicx}
\usepackage{caption}
\usepackage{sidecap}

\usepackage[textwidth=2cm,textsize=small]{todonotes}

\usepackage{pgf}
\usepackage{tikz}
\usetikzlibrary{arrows,decorations.pathmorphing,backgrounds,positioning,fit,calc,automata}
\usetikzlibrary{trees}
\usetikzlibrary{shapes}

\usepackage{multirow}

\newcommand{\pmap}{\rightharpoonup}
\newcommand{\cF}{\mathcal{F}}
\newcommand{\id}{\operatorname{id}}

\newcommand{\nodes}{\operatorname{nodes}}
\newcommand{\tleaves}{\operatorname{leaves}}

\newcommand{\trees}{\operatorname{Trees}}

\newcommand{\xred}{\operatorname{red}_{\cX}}

\newcommand{\collapse}{\downarrow}

\newcommand{\cG}{\mathcal{G}}

\newcommand{\textend}{\operatorname{extend}}
\newcommand{\tprune}{\operatorname{prune}}
\newcommand{\tshrink}{\operatorname{shrink}}
\newcommand{\treduce}{\operatorname{reduce}}

\newcommand{\Oo}{\mathcal{O}}

\newcommand{\NLogSpace}{\textsc{NLogSpace}}

\renewcommand\thmcontinues[1]{Continued}

\newcommand{\expr}{\operatorname{Expr}}
\newcommand{\cX}{\mathcal{X}}
\newcommand{\cY}{\mathcal{Y}}
\newcommand{\hX}{\hat{\cX}}
\newcommand{\comp}{\circ}
\newcommand{\var}{\operatorname{Var}}
\newcommand{\Var}{\var}
\newcommand{\subs}{\operatorname{Subs}}
\newcommand{\val}{\operatorname{Val}}
\newcommand{\gsem}[1]{|#1|}

\newcommand{\cR}{\mathcal{R}}

\newcommand{\AL}{\Sigma}

\newcommand{\REG}{\operatorname{REG}_\Sigma}

\newcommand{\dom}{\operatorname{dom}}

\newcommand{\sem}[1]{{\lsem{}{#1}\rsem}}

\newcommand{\cS}{\mathcal{S}}

\newcommand{\SR}{\mathbb{S}}
\newcommand{\sr}{S}
\newcommand{\add}{\oplus}
\newcommand{\mult}{\odot}

\newcommand{\bigadd}{\bigoplus}

\newcommand{\zero}{\mathbb{0}}
\newcommand{\one}{\mathbb{1}}

\newcommand{\nat}{\mathbb{N}}
\newcommand{\natinf}{\mathbb{N}_\infty}
\newcommand{\natninf}{\mathbb{N}_{-\infty}}

\newcommand{\trop}{\natinf(\min,+)}
\newcommand{\artic}{\natninf(\max,+)}
\newcommand{\natsemiring}{\nat(+, \cdot)}

\newcommand{\cL}{\mathcal{L}}
\newcommand{\cA}{\mathcal{A}}
\newcommand{\cB}{\mathcal{B}}

\newcommand{\trans}[2][]{\raisebox{-1pt}[10pt][0pt]{$\overset{#2}{\underset{^{#1}}{\raisebox{0pt}[3pt][0pt]{$\relbar\mspace{-8mu}\rightarrow$}}}$}}

\newcommand{\Run}{\operatorname{Run}}

\newcommand{\asem}[1]{{\sem{#1}}}

\tikzset{
defaultstyle/.style={>=stealth,semithick, auto,font=\small,
initial text= {},
initial distance= {3.5mm},
accepting distance= {3.5mm}},
accepting/.style=accepting by arrow,
nstate/.style={circle, semithick,inner sep=1pt, minimum size=4mm}}

\newtheorem{lemma}{Lemma}

\newtheorem{theorem}{Theorem}

\newtheorem{corollary}{Corollary}
\newtheorem{proposition}{Proposition}
\declaretheorem[style=definition]{example}
\renewcommand\thmcontinues[1]{Continued}

\begin{document}

\begin{frontmatter}
\title{Copyless Cost-Register Automata: \\
	Structure, Expressiveness, and Closure Properties\tnoteref{t1}}
\tnotetext[t1]{A preliminary version of this paper appeared in~\cite{MazowieckiR16} at the 33rd Symposium on Theoretical Aspects of Computer Science, STACS.}
 
\author[1]{Filip Mazowiecki\corref{cor1}}
\ead{filip.mazowiecki@labri.fr}

\author[2]{Cristian Riveros\corref{cor1}}
\ead{cristian.riveros@uc.cl}

\cortext[cor1]{Corresponding author}

\address[1]{University of Oxford, UK}
\address[2]{Pontificia Universidad Cat\'olica de Chile, Chile}

\begin{abstract}
Cost register automata (CRA) and its subclass, copyless CRA, were recently proposed by Alur et al. as a new model for computing functions over strings.
We study some structural properties, expressiveness, and closure properties of copyless CRA. 
We show that copyless CRA are strictly less expressive than weighted automata and are not closed under reverse operation.
To find a better class we impose restrictions on copyless CRA, which ends successfully with a new robust computational model that is closed under reverse and other extensions.
\end{abstract}


\end{frontmatter}

\section{Introduction}
\label{sec:introduction}
Weighted automata (WA) are an expressible extension of finite state automata for computing functions over strings.
They have been extensively studied since Sch\"{u}tzenberger~\cite{Schuetzenberger56}, and its decidability problems~\cite{Krob92,AlmagorBK11}, extensions~\cite{DrosteG07}, logic characterization~\cite{DrosteG07,kreutzer2013quantitative}, and applications~\cite{Mohri97,culik1993image} have been deeply investigated.

Recently, Alur et al.~\cite{alur2013regular,alur2013decision} introduced the computational model of cost register automata (CRA), an alternative model for computing functions over strings. The idea of this model is to enhance deterministic finite automata with registers that can be updated by combining registers contents using operations over a fixed semiring. 
In contrast to automata models with counters~\cite{HopcroftU79}, CRA blindly updates its registers on each transition by using values computed on the  previous state.
In~\cite{alur2013regular}, it was shown that CRA are strictly more expressive than WA.
Interestingly, it was also shown that a natural subfragment of CRA is equally expressive to WA, which gives a new representation for understanding this class of functions.

New representations for WA allows to study natural subclasses of functions that could not be proposed from the classical perspective.
This is the case for the class of copyless CRA that where proposed in~\cite{alur2013regular,AlurC11}. 
The idea of the so-called copyless restriction is to use each register at most once in every transition. Intuitively, this automaton model is what we call ``register-deterministic'' in the sense that it cannot copy the content of each register, similar to a deterministic finite automaton that cannot make a copy of its current state.
The copyless restriction was successfully used in the context of streaming tree transducers~\cite{alur2012streaming} for capturing MSO-transductions over trees and it was proposed as a natural restriction over CRA. 
Furthermore, copyless CRA are an excellent candidate for having good decidability properties.
It was stated in~\cite{alur2013regular} that ``the existing proofs of the undecidability of equivalence rely on the unrestricted non-deterministic nature of weighted automata''
and, thus, it is believed that copyless CRA might have good decidability properties.
Despite that this is a natural and interesting model for computing functions, research on this line has not been pursued further and not much is known about copyless~CRA.

In this paper we study the structure, expressiveness, and closure properties of copyless CRA.
We start by developing a toolkit of structural properties for analysing copyless CRA.
Towards this goal, we introduce a \emph{normal form} on the registers of copyless CRA.
We show that every copyless CRA can be put in this normal form which considerably simplifies the analysis of this model. 
With this restriction we provide further results that explain the flow and grow of registers content during a run.
Specifically, we prove that from its normal form one can identify a subset of registers that cannot be reset and are constantly growing during a run. 
These registers are called \emph{stable registers}, in the sense of having a stable assignment for transitions (see Section~\ref{sec:structure_copyless}).
We show that stable registers lead the behaviour of copyless CRA and that they are crucial to analyse the growing rate of loops. 

Then we turn our attention on studying the expressivity of copyless CRA. 
As a proof of concept, we use the structural properties developed in this paper to compare the expressiveness of copyless CRA with the class of functions defined by WA. 
We show that copyless CRA are strictly less expressive than WA.
It is important to stress that it was previously believed that copyless CRA are strictly less expressive than WA, but this is the first paper which proves this statement formally. 

In the last sections, we focus on the robustness of copyless CRA in terms of its closure properties.
The robustness of a computational model is usually measured in terms of how stable is the model when new operations or extensions are allowed. 
Deterministic finite automata are a good example for the previous statement: they are closed under several operations like union or intersection and, further, they can be enhanced by non-determinism and regular look-ahead without changing the class of recognized languages.
These properties are probably one of the reasons behind its fruitful connection with MSO logic or finite monoids~\cite{Buchi60,Rabin1959}.
Unfortunately, this measure of robustness put copyless CRA in an undesirable position: our expressiveness result shows that copyless CRA are not closed under reverse and, furthermore, under any extensions regarding directions of its reading head. 
This implies that the behaviour of copyless CRA is asymmetric with respect to the input, which buried our expectations of a robust class for computing functions.

The lack of good closure properties for copyless CRA fuels our interest in its subclasses. 
We consider a natural fragment of copyless CRA, called \emph{bounded alternation copyless} CRA (BAC).
This class was previously introduced in~\cite{kickasspaper} and  characterized in terms of the so-called \emph{Maximal Partition logic}.
In contrast to copyless CRA, BAC are robust under several and natural extensions previously considered in~\cite{alur2013regular,alur2012streaming}. 
Specifically, we show that BAC are closed under unambiguous non-determinism, regular look-ahead and under reverse.
Furthermore, all the structural toolkit introduced for copyless CRA also extend for this class.
These results emphasize that BAC is a promising computational model in the world of quantitative functions, showing that there exists a rich theory of functions below the class of WA.

\emph{New material in this article}. Preliminary versions of some of the results in this
work appeared in~\cite{MazowieckiR16}.  However, this article contains substantial new material. We include the full proofs of the main results and characterization presented in~\cite{MazowieckiR16}. In particular, the proofs in Section~\ref{sec:structure_copyless} are of fundamental importance to understand the structure of copyless CRA (e.g. Propositions~\ref{proposition:normal_form} to~\ref{prop:long_exp}) and the proof of Theorem~\ref{theorem:counterexample} shows how to use all the structural results to show its limits of expressibility. 
Furthermore, we believe that the proof of Theorem~\ref{theo:unambiguous-theorem} is interesting on its own: the determinization of unambiguous non-determistic BAC is a non-standard automata construction. The construction has some resemblance to the well-known Safra's construction~\cite{Safra88}, in the sense of keeping a tree structure of runs that is updated in a non-trivial way.
In terms of new results, in Section~\ref{sec:natural_semiring} we show that copyless CRA is strictly less expressive than WA over the semiring of natural numbers. Interestingly, the proof over the natural semiring is simpler than over the max-plus semiring and works over the one-letter alphabet.

\emph{Organization}. In Section~\ref{sec:preliminaries} we introduce copyless CRA and some basic definitions. In Section~\ref{sec:structure_copyless} we introduce the normal form and analyze the content of registers during the runs of copyless CRA. In Section~\ref{sec:natural_semiring} we work over the natural semiring and we show that copyless CRA are strictly contained in the class of weighted automata. Then we turn our attention to the max-plus semiring and we show in Section~\ref{sec:nonexpressibility} that the class of copyless CRA is not closed under reverse, which implies that copyless CRA are strictly contained in weighted automata also over the max-plus semiring. 
In Section~\ref{sec:bounded_alternation} we define BAC and show some closure properties of this class. We conclude in Section~\ref{sec:conclusions} with possible directions for future research.

\section{Preliminaries}
\label{sec:preliminaries}
In this section, we recall the definitions of cost register automata and the copyless restriction. We start with the definitions of expressions and substitutions over a semiring that are standard in the area.

%

\noindent\textbf{Semirings and functions.} A semiring is a structure $\SR = (\sr, \add, \mult, \zero, \one)$ where $(\sr, \add, \zero)$ is a commutative monoid, $(\sr-\{\zero\}, \mult, \one)$ is a monoid, multiplication distributes over addition, and $\zero \mult s = s \mult \zero = \zero$ for each $s \in \sr$. 
If the multiplication is commutative, we say that $\SR$ is commutative. In this paper, we always assume that $\SR$ is commutative and we usually denote the set of elements $S$ by the name of the semiring $\SR$.
For examples of semirings we will consider the \emph{semiring of natural numbers} $\natsemiring = (\nat, +, \cdot, 0, 1)$, the \emph{min-plus semiring} $\trop = (\natinf, \min, +, \infty, 0)$ and the \emph{max-plus semiring} $\artic = (\natninf, \max, +, -\infty, 0)$ which are standard semirings in the field of weighted automata~\cite{DrosteHWA09}. 

In this paper, we study the specification of functions from strings to values, namely, from $\Sigma^*$ to $\SR$.
We say that a function $f: \Sigma^* \rightarrow \SR$ is definable by a computational system $\cA$ (e.g. weighted automaton, or CRA) if $f(w) = \asem{\cA}(w)$ for any $w \in \Sigma^*$, where $\asem{\cA}$ is the semantics of $\cA$ over strings.  
For any string $w$, we denote by $w^r$ the reverse string. 
We say that a class of functions $F$ is \emph{closed under reverse}~\cite{alur2013regular} if for every $f \in F$ there exists a function $f^r \in F$ such that $f^r(w^r) = f(w)$ for all $w \in \Sigma^*$.

\noindent\textbf{Variables, expressions, and substitutions.}  Fix a semiring $\SR = (\sr, \add, \mult, \zero, \one)$ and a finite set of variables $\cX$ disjoint from $\sr$.
We denote by $\expr(\cX)$ the set of all syntactical \emph{expressions} that can be defined with $\cX$, constants in $\sr$, and the binary operations $\add, \mult$. 
For any expression $e \in \expr(\cX)$ we denote by $\var(e)$ the set of variables in $e$. 
We call an expression $e \in \expr(\cX)$ a \emph{ground expression} if $\var(e) = \emptyset$. For any ground expression we define $\asem{e} \in \SR$ to be the evaluation of $e$ with respect to $\SR$.

A \emph{substitution} over $\cX$ is defined as a mapping $\sigma: \cX \rightarrow \expr(\cX)$.
We denote the set of all substitutions over $\cX$ by $\subs(\cX)$.
A \emph{ground substitution} $\sigma$ is a substitution where the expression $\sigma(x)$ is ground for every $x \in \cX$.
Any substitution $\sigma$ can be extended to a mapping $\hat{\sigma}: \expr(\cX) \rightarrow \expr(\cX)$ such that, for every $e \in \expr(\cX)$, $\hat{\sigma}(e)$ is the expression $e[\sigma]$ of substituting each $x \in \var(e)$ by the expression $\sigma(x)$.
For example, if $\sigma(x) = 2x$ and $\sigma(y) = 3y$, and $e = x+y$, then $\hat{\sigma}(e) = 2x + 3y$.
Using the extension $\hat{\sigma}$, we can define the composition substitution $\sigma_1 \comp \sigma_2$ by $\sigma_1 \comp \sigma_2(x) = \hat{\sigma}_1(\sigma_2(x))$ for each $x \in \cX$.

A valuation is defined as a substitution of the form $\nu: \cX \rightarrow \SR$. 
We denote the set of all valuations over $\cX$ by $\val(\cX)$.
Clearly, any valuation $\nu$ composed with a substitution $\sigma$ defines a ground substitution, since $\var(\nu \circ \sigma (x))=\emptyset$ for all $x \in \cX$. We say that two expressions $e_1$ and $e_2$ are equal (denoted by $e_1 = e_2$) if they are equal up to evaluation equivalence, that is, $\asem{\nu \circ e_1} = \asem{\nu \circ e_2}$ for every valuation $\nu \in \val(\cX)$.
Similarly, we say that two substitutions $\sigma_1$ and $\sigma_2$ are equal (denoted by $\sigma_1 = \sigma_2$) if $\sigma_1(x) = \sigma_2(x)$ for every $x \in \cX$. 

\noindent\textbf{Cost register automata.} 
A cost register automaton (CRA) over a semiring $\SR$ \cite{alur2013regular} is a tuple $\cA = (Q, \AL, \cX, \delta, q_0, \nu_0, \mu)$  
where $Q$ is a finite set of states, $\AL$ is a finite alphabet, $\cX$ is a finite set of variables (we also call them registers), $\delta: Q \times \AL \rightarrow Q \times \subs(\cX)$ is the transition function, $q_0$ is the initial state, $\nu_0: \cX \rightarrow \SR$ is the initial valuation, and $\mu: Q \rightarrow \expr(\cX)$ is the final output function.  
A configuration of $\cA$ is a tuple $(q, \nu)$ where $q\in Q$ and $\nu \in \val(\cX)$ represents the current values in the variables of $\cA$.
Given a string $w = a_1 \ldots a_n \in \AL^*$, the run of $\cA$ over $w$ is a sequence of configurations:
$
(q_0, \nu_0) \:\trans{a_1} \: (q_1, \nu_1) \: \trans{a_2} \: \ldots \: \trans{a_n} \: (q_n, \nu_n) 
$
such that, for every $1 \leq i \leq n$,  $\delta(q_{i-1}, a_i) = (q_i, \sigma_i)$ and $\nu_i(x) = \asem{\nu_{i-1} \comp \sigma_i(x)}$ for each $x \in \cX$.
The output of $\cA$ over $w$, denoted by $\asem{\cA}(w)$, is $\asem{\nu_n \circ \mu(q_n)}$.
Note that we assume that every CRA defines a (total) function from words to values, opposed to the definition in \cite{AlurDDRY11} where CRA can define partial functions. This difference is not significant and all our results can be extended from total functions to partial functions. 

The run of $\cA$ over $w$ can be equally defined in terms of ground expressions rather than values. A ground configuration of $\cA$ is a tuple $(q, \varsigma)$ where $q \in Q$ and $\varsigma \in \subs(\cX)$ is a ground substitution. Given a string $w = a_1 \ldots a_n \in \AL^*$, the ground run of $\cA$ over $w$ is a sequence of ground configurations:
$
(q_0, \varsigma_0) \:\trans{a_1} \ldots \: \trans{a_n} \: (q_n, \varsigma_n) 
$
such that for $1 \leq i \leq n$,  $\delta(q_{i-1}, a_i) = (q_i, \sigma_i)$, $\varsigma_0 = \nu_0$ and $\varsigma_i(x) = \varsigma_{i-1} \circ \sigma_i(x)$ for each $x \in \cX$.
We denote the output ground expression of $\cA$ over a string $w$ by $|\cA|(w) = \varsigma_{n} \circ \mu(q_n)$.
Notice that, in contrast to ordinary runs, ground runs keep ground expressions as partial values of the run.
It is easy to see that  $\asem{\cA}(w) =  \asem{|\cA|(w)}$.

We define the transitive closure of the transition function $\delta^* : Q \times \AL^* \to Q \times \subs(\cX)$, by induction over the word-length. 
Formally, $\delta^*(q, \epsilon) = (q, \id)$ where $\epsilon$ is the empty word and $\id(x) = x$ for all $x \in \cX$; and $\delta^*(q_1, w \cdot a) = (q_3, \sigma \circ \sigma')$, whenever $\delta^*(q_1,w) = (q_2,\sigma)$ and $\delta(q_2, a) = (q_3,\sigma')$. 
For a CRA $\cA$ we define the set $\subs(\cA) = \{ \sigma \mid \text{$\delta^*(p, w) = (q, \sigma)$ for $p,q \in Q$ and $w \in \Sigma^*$}\}$. 

\noindent\textbf{Copyless restriction and copyless CRA.} We say that an expression $e \in \expr(\cX)$ is \emph{copyless} if $e$ uses every variable from $\cX$ at most once.
For example, $x \cdot (y + z)$ is copyless but $x \cdot y + x \cdot z$ is not copyless (i.e. $x$ is mentioned twice).
Notice that the copyless restriction is a syntactic constraint over expressions. 
Furthermore, we say that a substitution $\sigma$ is \emph{copyless} if for every $x \in \cX$ the expression $\sigma(x)$ is copyless and $\var(\sigma(x)) \cap \var(\sigma(y)) = \emptyset$ for $x, y \in \cX$ and $x \neq y$. 

A CRA $\cA$ is called \emph{copyless}~\cite{alur2013regular} if for every transition $\delta(q_1, a) = (q_2, \sigma)$  the substitution $\sigma$ is copyless; and for every state $q \in Q$ the output expression $\mu(q)$ is copyless.
In other words, every time $\cA$ updates registers or outputs a value, each register can be used only once.
It is straightforward that if $\cA$ is copyless then all substitutions $\sigma \in \subs(\cA)$ are also copyless.
In the following, we give some  examples of copyless CRA. 

\begin{figure}[t]
	\begin{center}
\hspace{-1cm}
		\begin{tikzpicture}[-\string>,\string>=stealth',shorten \string>=1pt,auto,node distance=1.8cm,semithick,initial text={}]
		\tikzstyle{every state}=[fill=white,draw=black,text=black]

		\node[state, draw=white] 		(p0) at (-4,0) {};
		\node[state, right of=p0, node distance=1.5cm] 		(p) {};
		\node[state,  below of=p, node distance=1.5cm, draw=white] (pf) {};
		
		\draw (p0) edge[pos=0.01] node {$x,y:=0 \;\;\;\;$} (p);
		\draw (p) edge[pos=0.8] node {$\; \max\{x,y\}$} (pf);
		
		\path (p)   edge	[loop above]	node {
			\renewcommand{\arraystretch}{0.7}
			$\;\;\;\;\;\;\;\;\;\; \;\;\;\;\; \;\;\;\;\; \;\;\;\;\;  a \;\;   \begin{array}{|rcl}
			\;\; x & := & \max\{x,y\} \\
			\;\; y & := & 0
			\end{array}$
		}  (p)
		(p)   edge	[loop right]	node {
			\renewcommand{\arraystretch}{0.7}
			$\; b \;\;
			\begin{array}{|rcl}
			\;\; x & := & x \\
			\;\; y & := & y + 1
			\end{array}$}  (p);
	
		
		\node[state, draw=white] 		(p02) at (3,0) {};
		\node[state, right of=p02, node distance=1.5cm] 		(p2) {};
		\node[state,  below of=p2, node distance=1.7cm, draw=white] (pf2) {};
		
		\draw (p02) edge[pos=0.01] node {$x,y,z:=0 \;\;\;\;\;\;\;\;\;\;$} (p2);
		\draw (p2) edge[pos=0.8] node {$\; x+\max\{y,z\}$} (pf2);
		
		\path (p2)   edge	[loop above]	node {
			\renewcommand{\arraystretch}{0.7}
			\;\;\;\;\;\;\;\;\;\;\;\;\;\;\;\;\;\;\;\;\;\;\;\;\;\;\;\;\;  $\# \;\;   \begin{array}{|rcl}
			\;\; x & := & x + \max\{y,z\} \\
			\;\; y,z & := & 0 
			\end{array}$
		}  (p2)
		(p2)   edge	[loop right]	node {
			$
			\renewcommand{\arraystretch}{1.2}
			\begin{array}{ll}
			\; a \! &
			\renewcommand{\arraystretch}{1}
			\begin{array}{|rcl}
			\;\; y & := & y+1 
			\end{array} \\
			\; b \! &
			\renewcommand{\arraystretch}{1}
			\begin{array}{|rcl}
			\;\; z & := & z+1 
			\end{array} 
			\end{array}
			$
		}  (p2);
		\node at ($(p) + (-1.5, -0.7)$) {$\cA_1$};
		\node at ($(p2) + (-1.5, -0.7)$) {$\cA_2$};
			\end{tikzpicture}
		\vspace{-0.8cm}
\caption{Examples of copyless cost-register automata}\label{fig:CRA}
	\end{center}
	\vspace{-0.75cm}
\end{figure}
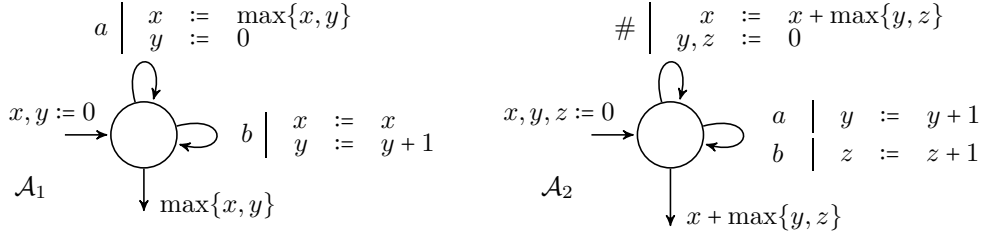

\begin{example} \label{ex:max-b-substrings}
	Let $\SR$ be the max-plus semiring $\artic$ and $\Sigma = \{a,b\}$. 
	Consider the function $f_1$ that for a given string $w \in \Sigma^*$ computes the longest substring of $b$'s.
	This can be easily defined by the CRA $\cA_1$ in Figure~\ref{fig:CRA}. 
	The CRA $\cA_1$ stores in the $y$-register the length of the last suffix of $b$'s and in the $x$-register the length of the longest substring of $b$'s seen so far.
	One can easily check that $\cA_1$ is a copyless CRA. 
	Indeed, each substitution is copyless and the final output expression $\max\{x, y\}$ is copyless as well. 
\end{example}

\begin{example}
	\label{ex:no-poly-ambiguous}
	Let $\SR$ be the max-plus semiring $\artic$ and $\Sigma = \{a,b, \# \}$.
	Consider the function $f_2$ such that, for any $w \in \Sigma^*$ of the form $w_0 \# w_1 \# \ldots \# w_n$ with $w_i \in \{a,b\}^*$, it computes the maximum number of $a$'s or $b$'s for each substring $w_i$ (i.e. $\max\{|w_i|_a, |w_i|_b\}$) and then it sums these values over all substrings $w_i$, that is, $f_2(w) = \sum_{i=0}^n \max\{|w_i|_a, |w_i|_b\}$. One can check that the copyless CRA $\cA_2$ in Figure~\ref{fig:CRA} computes $f_2$.
	The copyless CRA $\cA_2$ follows similar ideas to $\cA_1$: the registers $y$ and $z$ count the number of $a$'s and $b$'s, respectively, in the longest suffix without $\#$ and the register $x$ stores the partial output without considering the last suffix of $a$'s and $b$'s.
	Note that in Figure~\ref{fig:CRA} we omit an assignment if a register is not updated (i.e. it keeps its previous value). 
	For example, for the $a$-transition we omit the assignments $x:=x$ and $z:=z$ for the sake of presentation. 
	One should keep in mind these assignments because of the copyless restriction.
	
\end{example}



\noindent\textbf{Weighted automata and linear CRA.} We will use the class of functions defined by weighted automata only to compare its expressibility with copyless CRA (see Section~\ref{sec:natural_semiring} and~\ref{sec:nonexpressibility}). For the sake of completeness, we give the definition and semantics of weighted automata here (see~\cite{Sakarovitch09,DrosteHWA09} for some concrete examples of weighted automata).
Fix a finite alphabet $\AL$ and a commutative semiring $\SR$.
A \emph{weighted automaton} (WA) over $\AL$ and $\SR$ is a tuple $\cA = (Q, \AL, E, I, F)$ where $Q$ is a finite set of states, $E: Q \times \AL \times Q \rightarrow \SR$ is a weighted transition relation, and $I, F: Q \rightarrow \SR$ are the initial and the final function, respectively~\cite{Sakarovitch09,DrosteHWA09}.
Usually, if $E(p, a, q) = s$, we denote this transition graphically by $p \:\trans{a/s}\: q$.
A \emph{run} $\rho$ of $\cA$ over a word $w = w_1 \ldots w_n$ is a sequence of transitions:
$
\rho = q_0 \:\trans{w_1/s_1}\: q_1  \:\trans{w_2/s_2}\: \cdots\:\trans{w_n/s_n}\: q_n.
$
The \emph{weight} of a run $\rho$ of $\cA$ over $w$ is defined by
$
|\rho| = I(q_0) \mult (\bigodot_{i=1}^{n} s_i) \mult F(q_n).
$
We define $\Run_\cA(w)$ as the set of all runs of $\cA$ over $w$.
Finally, the weight of $\cA$ over a word $w$ is defined by
$
\asem{\cA}(w) = \bigoplus_{\rho \in \Run_\cA(w)} |\rho|
$
where the sum is equal to $\zero$ if $\Run_\cA(w)$ is empty.

In~\cite{alur2013regular} it was shown that the class of functions defined by WA over some semiring $\SR$ is equally expressive to the class of functions defined by linear CRA over $\SR$. Formally, we call a CRA \emph{linear} if all its substitutions and also its final output function are made by linear expressions over $\SR$, namely, expressions of the form $\bigadd_{i=1}^n s_i \mult x_i$ for some values $s_i \in \SR$. For example, the CRA $\cA_1$ in Example~\ref{ex:max-b-substrings} is linear and can be defined by a WA, but the CRA $\cA_2$ in Example~\ref{ex:no-poly-ambiguous} is not linear, given that the expression $x+\max\{y,z\}$ is not linear (although it is not hard to redefine $\cA_2$ by an equivalent weighted automaton or linear CRA).

\noindent\textbf{Trim assumption.}
For technical reasons, in this paper we assume that our finite automata and cost register automata are always \emph{trim}, namely, all their states are reachable from some initial states (i.e., they
are accessible) and they can reach some final state (i.e., they are co-accessible).
It is worth noticing that verifying if a state is accessible or co-accessible is reduced to a reachability test in the transition graph~\cite{papadimitriou1993computational} and this can be done in \NLogSpace.
Thus, we can assume without lost of generality that all our automata are trimmed. 


\noindent\textbf{Copyless expressions as sum of monomials.}
By distributivity it is easy to show that every copyless expression can be rewritten as a sum of monomials (despite that the new expression does not have to preserve the copyless property). 
This representation of copyless expressions as sum of monomials is an easy technical result and the reader can safely skip its proof in a first read.

\begin{lemma} \label{lemma:copyless-expressions}
	For any copyless expression $e$, there exist an expression $e'$ of the form
	$\bigoplus_{i = 1}^{k} \, \left( c_i \odot \bigodot_{x \in X_i} x \right)$
	where $k \geq 0$, $X_i \subseteq \var(e)$, $c_i \in \SR$, and all sets $X_1, \ldots, X_k$ are different, such that $e \equiv e'$.
\end{lemma}

\begin{proof}
	The lemma is shown by induction on the size of $e$. For the base case, when $e$ is equal to a constant or a variable, the lemma trivially holds by taking $e' = e$. For the inductive case, suppose that $e = e_1 \otimes e_2$ where $\otimes$ is either $\add$ or $\mult$. 
	By the inductive hypothesis we know that there exist expressions $e_1'$ and $e_2'$ equivalent to $e_1$ and $e_2$, respectively, such that for $j \in \{1,2\}$:
	\[
	e_j' \; \equiv \; \bigoplus_{i = 1}^{k_j} \, \left( c_i^j \odot \bigodot_{x \in X_i^j} x \right)
	\]
	where $X_1^j, \ldots, X_{k_j}^j$ is a sequence of different sets over $\cX$ and $c_1^j, \ldots, c_{k_j}^j$ is a sequence of values over $\SR$ for $k_j \ge 0$.
	Given that $e$ is a copyless expression we know that $\var(e_1) \cap \var(e_2) = \emptyset$. 
	Then since $X_i^j \subseteq \var(e_j)$ and $\var(e_1) \cap \var(e_2) = \emptyset$, we get:
	\begin{align}
	X_{i_1}^1 \cap X_{i_2}^2 = \emptyset & \;\;\;\;\;\;\; \text{for every } i_1 \leq k_1 \text{ and } i_2 \leq k_2. \label{fact:non-intersecting-variables} 
	\end{align}
	Now we consider two cases when $\otimes$ is either $\add$ or $\mult$. If $e = e_1 \add e_2$, then by considering $e' = e_1' \add e_2'$ the lemma is proved.
	Otherwise, $e = e_1 \mult e_2$ and we get:
	\begin{align*}
	e_1' \mult e_2' & \;\; = \;\; \bigoplus_{i = 1}^{k_1} \, \big( c_i^1 \odot \bigodot_{x \in X_i^1} x \big) \ \mult \ \bigoplus_{i = 1}^{k_2} \, \big( c_i^2 \odot \bigodot_{x \in X_i^2} x \big) \\
	& \;\; = \;\; \bigoplus_{i_1 = 1}^{k_1} \, \bigoplus_{i_2 = 1}^{k_2} \ \big( c_{i_1}^1 \mult c_{i_2}^1 \odot \bigodot_{x \in X_{i_1}^1 \cup  X_{i_2}^2} x \big) \;\; = \;\; e' \\
	\end{align*}
	Recall that $X_{i_1}^1 \cap X_{i_2}^2 = \emptyset$ by~(\ref{fact:non-intersecting-variables}).
	It remains to show that all sets of the form $X_{i_1}^1 \cup X_{i_2}^2$ are different.
	Indeed, all sets $X_{i_1}^1$ are pairwise different and the same holds for $X_{i_2}^2$. 
	This means that all sets $X_{i_1}^1 \cup X_{i_2}^2$ are pairwise different as well.
\end{proof}

\noindent\textbf{Removing zeros from copyless CRA.}
\label{subsection:zeros}
We say that an expression $e \in \expr(\cX)$ is \emph{reduced} if $e = \zero$ or the $\zero$-constant is not mentioned inside $e$. 
It is straightforward to show that for any expression $e$ there exists an equivalent expression $e^*$ that is reduced. 
Indeed, one can construct inductively an equivalent expression by using the following reductions: $e \add \zero = e$ and $e \mult \zero = \zero$. 
Then by reducing each subexpression recursively, the resulting expression is either $\zero$ or do not use the $\zero$-constant at all. 
Further, note that if $e$ is copyless, then its reduced expression $e^*$ is copyless as well. 

Let $f : \AL^* \to \SR$ be a function definable by a copyless CRA. We say that $f$ is a \emph{non-zero function} if $f(w) \neq \zero$ for all $w \in \Sigma^*$.
The following result shows that, without lost of generality, we can assume that all constants in a copyless CRA are different from $\zero$ whenever the function defined by $\cA$ is a non-zero function.
Like for Lemma~\ref{lemma:copyless-expressions}, the proof of Lemma~\ref{proposition:non-zero} is an easy but technical construction, and the reader can skip its proof in a first read.

\begin{lemma}
	\label{proposition:non-zero}
	Let $\cA$ be a copyless CRA such that $\asem{\cA}$ is a non-zero function. Then there exists a copyless CRA $\cA'$ such that its initialization, substitutions and final output functions are reduced and different from~$\zero$.
\end{lemma}

\begin{proof}
 We show how to avoid keeping a $\zero$ value in registers that where forced to be $\zero$ (i.e. forced by a substitution $\sigma$ such that $\sigma(x) =\zero$ for some register $x$). The idea is keep the information which registers are equal to $\zero$ in the states. 
 Let $\cA  =  (Q, \AL, \cX, \delta, q_0, \nu_0, \mu)$ be a copyless CRA.
 We define a new copyless CRA $\cA' =  (Q', \AL, \cX, \delta, q_0', \nu_0', \mu')$, where $Q' = Q \times 2^\cX$ is the new set of states, $q_0' = (q_0, S_0)$ where $S_0$ is the set of registers such that $x \in S$ iff $\nu_0(x) = \zero$, and  $\nu_0'$ is the same as $\nu_0$ for registers in $\cX \setminus S_0$. Further, for registers $x \in S_0$ we define $\nu_0'(r) = \one$ (or any other constant in $\SR \setminus \{\zero\}$).
 
 For the definition of $\delta'$ and $\mu'$ we introduce some notation. Let $S \subseteq \cX$ and $e$ be an expression over $\cX$. We define the expression $e[S]$ as the reduced expression $d^*$, where $d$ is the result of taking $e$ and replacing all registers $x \in S \cap \var(e)$ by $\zero$.
 We are ready to define $\delta'$ and $\mu'$. 
 For every $(q, S) \in Q'$ and $a \in \Sigma$, we define $\delta'((q,S), a) = ((q',S'), \sigma')$, where $\delta(q, a) = (q',\sigma)$, $S'$ is the set of all registers $x$ such that $\sigma(x)[S]$ is equal to $\zero$, and $\sigma'$ is defined for every $x \in \cX$ as follows:
 \[
 \sigma'(x) \; = \; \left\{ 
 \begin{array}{ll}
 \sigma(x)[S] \;\;& \text{if } x \notin S' \\
 \one & \text{if } x \in S'
 \end{array}
 \right.
 \]
 Finally, we define the output function $\mu'((q, S)) = \mu(q)[S]$, for every $(q, S) \in Q'$.
 
 It is straightforward to show that $\cA'$ and $\cA$ define the same function and that $\nu_0'$, $\delta'$ and $\mu'$ do not use the $\zero$-constant.
 Note that there cannot exist a reachable state $(q, S) \in Q'$ such that $\mu'((q, S)) = \zero$. Otherwise, the output of the function defined by $\cA$ will be $\zero$ for some input in $\Sigma^*$, which contradicts the fact that $\asem{\cA}$ is a non-zero function.	
\end{proof}

We say that such an automaton $\cA$ is non-zero. Similarly, a substitution $\sigma$ is non-zero if $\sigma(x)$ is non-zero for every register $x$.
Furthermore, we say that a semiring is a \emph{non-zero semiring} if every expression $e$ that does not use $\zero$ cannot evaluate to $\zero$. For example in the $\artic$ semiring, where $\zero = -\infty$ if an expression uses only natural numbers then with $+$ and $\max$ operation one can never get $-\infty$. With this property every non-zero CRA over a non-zero semiring do not have $\zero$ in the registers in any configuration.
We will need this property of the $\artic$ semiring later in Section~\ref{sec:structure_copyless} and~\ref{sec:nonexpressibility}.


\section{Structural properties of copyless CRA}
\label{sec:structure_copyless}
In this section we analyze the structure of copyless CRA and develop some machinery that will be useful in Section~\ref{sec:nonexpressibility}. 
These results will help to understand the internal structure of copyless CRA.

\subsection{Normal form} Let $\cA = (Q, \AL, \cX, \delta, q_0, \nu_0, \mu)$ be a copyless CRA and let $\preceq$ be a predefined linear order over $\cX$.
We say that a substitution $\sigma : \cX \to \expr(\cX)$ is in normal form with respect to $\preceq$ if $x \preceq y$ for all $x \in X$ and all $y \in \var(\sigma(x))$.
In other words, all variables mentioned in $\sigma(x)$ are greater or equal to $x$ with respect to~$\preceq$.
Furthermore, we write that $\cA$ is in \emph{normal form} with respect to $\preceq$ if every $\sigma \in \subs(\cA)$ is in normal form with respect to $\preceq$.
For example, the copyless CRAs in Example~\ref{ex:max-b-substrings} and~\ref{ex:no-poly-ambiguous} are in normal form with respect to the orders $x \preceq y$ and $x \preceq y \preceq z$, respectively.
For the sake of simplification, we assume that every set of registers $\cX$ is given with a linear order $\preceq_{\cX}$ and instead of writing that $\cA$ or $\sigma$ are in normal form with respect to $\preceq_{\cX}$ we write in short that $\cA$ or are in normal form.

\begin{example}\label{example:normal_form}
	
Consider the set of registers $\cX = \{x,y\}$ with the order $x \preceq y$.
The copyless CRA $\cB_1$ in Figure~\ref{fig:CRA2} is not in normal form, because of the $b$-transitions.
On the other hand, the copyless CRA $\cB_2$ is in normal form. For both automata we omit the initial states, initial valuations and final output functions because they are not relevant for the discussion.
\end{example}

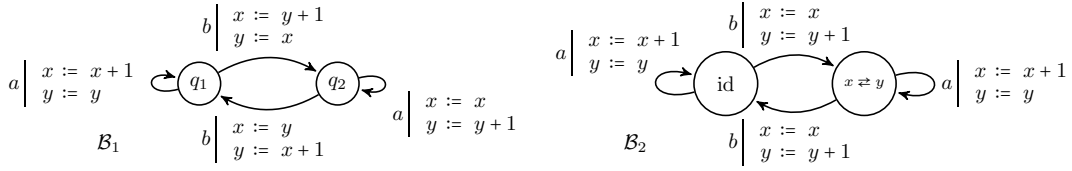
\begin{figure}[t]
	\begin{center}
		\begin{tikzpicture}[-\string>,\string>=stealth',shorten \string>=1pt,auto,node distance=2.3cm,semithick,initial text={},scale=0.4,every node/.style={scale=0.8}]
		\tikzstyle{every state}=[fill=white,draw=black,text=black,style={font=}]
		
		\node[state,minimum size=2ex] 		(p) at (-9,-0.4) {$q_1$};
		\node[state,minimum size=2ex] 		[right of=p](r) {$q_2$};
		\node at ($(p)+(-3,-2)$) {$\cB_1$};
		
		\path (p)   edge	[loop left]		node (s1) {$a \; 
			\renewcommand{\arraystretch}{0.9} 
			\begin{array}{|r c l}
			\; x & \!\!\!\! := \!\!\!\! & x+1 \\
			\; y & \!\!\!\! := \!\!\!\! & y
			\end{array}$}  (p)
		(p)   edge	[bend left]	node {$b \;
			\renewcommand{\arraystretch}{0.9}
			\begin{array}{|rcl}
			\; x & \!\!\!\! := \!\!\!\! & y+1 \\
			\; y & \!\!\!\! := \!\!\!\! & x
			\end{array}$}  (r)
		(r)   edge	[bend left]	node {$b \;  
			\renewcommand{\arraystretch}{0.9}
			\begin{array}{|rcl}
			\; x & \!\!\!\! := \!\!\!\! & y \\
			\; y & \!\!\!\! := \!\!\!\! & x+1
			\end{array}$}  (p)
		(r)   edge	[loop right]	node [align=center,below right]{$a \;   
			\renewcommand{\arraystretch}{0.9}
			\begin{array}{|rcl}
			\; x & \!\!\!\! := \!\!\!\! & x \\
			\; y & \!\!\!\! := \!\!\!\! & y+1
			\end{array}$}  (r);
		
		
		\node[state,minimum size=7ex] 		(p) at (8.5,-0.4) {$\operatorname{id}$};
		\node[state,minimum size=7ex] 		[right of=p](r) {\scriptsize{$x \rightleftarrows y$}};
		
		\node at ($(p)+(-3,-2)$) {$\cB_2$};
		
		\path (p)   edge	[loop left]		node [above] (s1) {$a \; 
			\renewcommand{\arraystretch}{0.9} \begin{array}{|rclc}
			\; x & \!\!\!\! := \!\!\!\! & x+1 &\textcolor{white}{aaaa} \\
			\; y & \!\!\!\! := \!\!\!\! & y
			\end{array}$}  (p)
		(p)   edge	[bend left]	node {$b \;
			\renewcommand{\arraystretch}{0.9}
			\begin{array}{|rcl}
			\; x & \!\!\!\! := \!\!\!\! & x \\
			\; y & \!\!\!\! := \!\!\!\! & y+1
			\end{array}$}  (r)
		(r)   edge	[bend left]	node {$b \;  
			\renewcommand{\arraystretch}{0.9}
			\begin{array}{|rcl}
			\; x & \!\!\!\! := \!\!\!\! & x \\
			\; y & \!\!\!\! := \!\!\!\! & y+1
			\end{array}$}  (p)
		(r)   edge	[loop right]	node [align=center]{$a \;
			\renewcommand{\arraystretch}{0.9}
			\begin{array}{|rcl}
			\; x & \!\!\!\! := \!\!\!\! & x+1 \\
			\; y & \!\!\!\! := \!\!\!\! & y
			\end{array}$}  (r);
		\end{tikzpicture}
		\caption{Examples of copyless cost-register automata regarding normal form}\label{fig:CRA2}
	\end{center}
\end{figure}

In the previous example, the automaton $\cB_1$ uses the registers $x$ and $y$ to count the number of $a$'s and $b$'s.
However, depending on the current state both registers have either the number of $a$'s or the number of $b$'s. 
It is clear that one would like to avoid this type of behavior in a theoretical analysis. 
Intuitively, one register should always contain the number of $a$'s and the other register the number of $b$'s.
One can clearly transform $\cB_1$ to an automaton in normal form by exchanging the use of $x$ and $y$ in the transitions and encoding this permutation between registers in the states.
This is precisely the automaton $\cB_2$ in Figure~\ref{fig:CRA2}. 
We generalize this idea for all copyless CRA in the next proposition.


\begin{proposition}
\label{proposition:normal_form}
	For every copyless CRA $\cA$ there exists a copyless CRA in normal form $\cA'$ with the same set of registers such that they output the same ground expressions for all words and thus recognize the same function. The number of states in $\cA'$ can be bounded exponentially in the size of the automaton $\cA$.
\end{proposition}

\begin{proof}
Let $\cA \; = \; (Q, \AL, \cX, \delta, q_0, \nu_0, \mu)$ be a copyless CRA. 
 We define a copyless CRA $\cA'$ in normal form such that $\cA'$ computes the same function as $\cA$.
 The idea of $\cA'$ is to store a state $q \in Q$ and a permutation $\rho$ of the set $\cX$ such that, if $(q, \nu)$ is the current configuration of $\cA$ over an input $w$, then $((q,\rho), \nu')$ is the configuration of $\cA'$ over $w$ and $\nu(x) = \nu'(\rho(x))$. 
 In other words, the content of $\nu$ is still defined by $\nu'$ but the value $\nu(x)$ of $x$ is now in the register $\rho(x)$ for every $x \in \cX$.
 The permutation of registers' content will allow us to keep the normal form in $\cA'$. 
 Formally, let $\cA' \; = \; (Q', \AL, \cX, \delta', q_0', \nu_0', \mu')$, where:
 \[
 Q' \; = \; Q \times \{\rho \mid \rho \text{ is a permutation of the set } \cX\}
 \]
 is the set of states, $q_0' = (q_0, \id)$ is the initial state where $\id$ is the identity permutation, and $\nu_0' = \nu_0$ is the initial function.
 For the sake of presentation, let us show how the run of $\cA'$ will correspond to the run of $\cA$ before defining the transition function $\delta'$ and the output function~$\mu'$. 
 For an expression $e \in \expr(\cX)$ and a permutation $\rho$ over $\cX$, we define $\rho(e)$ to be the expression $e$ where the registers are replaced according to~$\rho$.
 Let $(q, \nu)$ and $(q', \nu')$ be the configuration of $\cA$ and $\cA'$, respectively, after reading $w \in \Sigma^*$. We will show:
\begin{equation}
\label{property}
q' = (q, \rho) \text{ for some permutation } \rho \text{ and } \nu(x) = \nu'(\rho(x)) \text{ for every } x \in \cX.
\end{equation}
Note that, for the word $\epsilon$, this correspondence holds since $(q_0, \nu_0)$ and $((q_0, \id), \nu_0')$ are the initial configuration of $\cA$ and $\cA'$, respectively, and $\nu_0'(\id(x)) = \nu_0'(x) = \nu_0(x)$. 
We define $\mu'((q,\rho)) = \rho(\mu(q))$ and show that if~(\ref{property}) always holds, then $\cA$ and $\cA'$ computes the same function. 
Indeed, if the transition function $\delta'$ preserves~(\ref{property}), then:
\[
\begin{array}{rcll}
\nu'(\mu'((q,\rho))) & = & \nu'(\rho(\mu(q))) \;\;\; & \text{(by definition of $\mu'$)} \\
& = & \nu(\mu(q)) & \text{(by Property (\ref{property})}
\end{array}
\]
This proves that the outputs of $\cA$ and $\cA'$ are the same (provided that $\delta'$ preserves (1)).
It remains to define $\delta'$ such that $\cA'$ is a copyless CRA in normal form and its definition satisfies~(\ref{property}).

\newcommand{\sS}{S_{\sigma, \rho}}
\newcommand{\stau}{\tau_{\sigma, \rho}}

We need some additional definitions. 
For a copyless substitution $\sigma$ and a permutation $\rho$ both over $\cX$ we define the set $\sS = \{x \in \cX \mid \var(\rho(\sigma(x))) \neq \emptyset\}$, that is, the set of all variables $x$ where $\sigma(x)$ is not ground.
Further, define the set $\sS' = \{y \mid \exists x \in \cX. \,y = \min\{\var(\rho(\sigma(x)))\}\}$. For each $x \in \sS$, the set $\var(\rho(\sigma(x)))$ is non-empty and thus has a least element. This motivates the function $\stau^0: \sS \rightarrow \sS'$ that associates to each $x \in \sS$ its least element, formally:
\begin{equation}
\label{property2}
\stau^0(x) \; = \; \min\left\{\var\big(\rho(\sigma(x)) \big) \right\}
\end{equation}
One can easily check that $\stau^0$ is a bijective function from $\sS$ to $\sS'$.
It is surjective by the definition of $\sS$ and $\sS'$, and injective by the copyless restriction over $\sigma$.
To see the last claim, recall that any copyless substitution satisfies $\var(\sigma(x)) \cap \var(\sigma(y)) = \emptyset$ and, in particular, $\var(\rho(\sigma(x))) \cap \var(\rho(\sigma(y))) = \emptyset$ for any permutation $\rho$. 
Finally, given that $\stau^0$ is a bijective function from $\sS \subseteq \cX$ to $\sS' \subseteq \cX$, we can extend $\stau^0$ to a bijection $\stau: \cX \rightarrow \cX$ such that $\stau(x) = \stau^0(x)$ for every $x \in \sS$.
Of course, there might be many extensions of $\stau^0$ to $\stau$ but we can choose any extension (the decision is not important for the construction).

We have now all the ingredients to define the transition function $\delta'$. For every $p \in Q$, $a \in \Sigma$, and permutation $\rho$ over $\cX$, if $\delta(p, a) = (q, \sigma)$ then we define $\delta'((p, \rho), a) = ((q, \stau), \sigma')$ such that
\[
\sigma'(x) \; = \; \rho(\sigma(\stau^{-1}(x)))
\]
for every $x \in \cX$.
The substitution  $\rho \circ \sigma \circ \stau^{-1}$ is a copyless substitution for every copyless $\sigma$ because $\rho$ and $\stau^{-1}$ are just permuting the variables. 
Therefore, we can conclude that $\sigma'$ is copyless as well.

Our next step is to show that $\sigma'$ is in normal form.
Recall that $\sigma'$ is in normal form if for every $x \in \cX$ it holds that $x \preceq y$ for every $y \in \var(\sigma'(x))$.
We prove this by case analysis by considering whether $x \in \sS'$ or not.
First suppose that $x \in \sS'$. 
Since $x$ is in the codomain of $\stau^0$ and $\stau$ is an extension of $\stau^0$, we have that $x = \min\{\var(\rho(\sigma(\stau^{-1}(x))))\}$ by (\ref{property2}). 
Then if we replace $\rho(\sigma(\stau^{-1}(x)))$ by $\sigma'$, we get that $x = \min\{\var(\sigma'(x))\}$.
In particular, we have that $x \preceq y$ for every $y \in \var(\sigma'(x))$.
Now suppose that $x \notin \sS'$.
This means that $x$ is not in the codomain of $\stau^0$ which implies that $\stau^{-1}(x) \notin \sS$.
In other words, $\sigma'(x) = \rho(\sigma(\stau^{-1}(x)))$ is a ground expression and $\var(\sigma'(x)) = \emptyset$.
Since for both cases it holds that $x \preceq y$ for every $x \in \cX$ and $y \in \var(\sigma'(x))$, then we conclude that $\sigma'$ is in normal form. 

For the last part of the proof, we show by induction that $\delta'$ satisfies the correspondence~(\ref{property}) between $\cA$ and $\cA'$.
Let $(p, \nu_n)$ and $((p, \rho), \nu_n')$ be the configuration of $\cA$ and $\cA'$, respectively, after reading $w \in \Sigma^*$. Assume by induction that $\nu_n(x) = \nu_n'(\rho(x))$ for every $x \in \cX$ and let:
\[
(p, \nu_n) \:\trans{a} \: (q, \nu_{n+1}) \;\; \text{ and } \; \; ((p, \rho), \nu_n') \:\trans{a} \: ((q, \stau), \nu_{n+1}')
\]
be the transitions for $\cA$ and $\cA'$, respectively,  after reading a new letter $a \in \Sigma$. 
We prove the correspondence~(\ref{property}) between $\nu_{n+1}$ and $\nu_{n+1}'$ as follows:
\[
\renewcommand{\arraystretch}{1.3}
\begin{array}{rcll}
\nu_{n+1}'(\stau(x)) & \; = \;   & \nu_{n}'(\sigma'(\stau(x))) & \text{(by definition of $\nu_{n+1}'$)} \\
&  \; = \; & \nu_{n}'(\rho(\sigma(\stau^{-1}(\stau(x))))) \;\;\; & \text{(by definition of $\sigma'$)} \\
&  \; = \;  & \nu_{n}'(\rho(\sigma(x))) & \text{(by composing $\stau$ and $\stau^{-1}$)} \\
& \; = \; & \nu_n(\sigma(x)) & \text{(by the induction assumption)} \\
&  \; = \;  & \nu_{n+1}(x) & \text{(by definition of $\nu_{n+1}$)} 
\end{array}
\]
This proves that the transition function $\delta'$ keeps the correspondence (\ref{property}) between $\cA$ and $\cA'$. Since it also holds for the initial configuration then by induction it holds for all reachable configurations, which proves that the outputs of $\cA$ and $\cA'$ are the same. 

So far, we have shown that $\cA$ and $\cA'$ define the same function and every substitution on $\cA'$-transitions are in normal form. To conclude the proof we need to show that if every substitution on $\cA$-transitions is in normal form, then every $\sigma \in \subs(\cA)$ is in normal form. We prove this by induction over the length of the word.
Assume that it holds for all words $w$ of length at most $n$ and let $\delta^*(q,w) = (q', \sigma)$. Suppose we want to extend $w$ with $a \in \Sigma$ and let $\delta(q',a) = (q'', \sigma')$. By definition, we know that $\delta^*(q,w\cdot a) = (q'',\sigma \circ \sigma')$. Set $y \in \var(\sigma \circ \sigma'(x))$. By definition there exists a register $z$ such that $z \in \var(\sigma'(x))$ and $y \in \var(\sigma(z))$. Since $\delta$ is in normal form, we conclude that $x \preceq z \preceq y$. In other words, $x \preceq y$ for every $y \in \var(\sigma \circ \sigma'(x))$ and, thus, $\sigma \circ \sigma'$ is in normal form. 	
\end{proof}


\subsection{Stable registers and reset substitutions}
\label{subsection:stable_reset}
Let $\cA$ be a copyless CRA in normal form. 
During a run of $\cA$ the content of its registers flows from higher to lower registers with respect to the total order $\preceq$. 
This does not necessarily mean that the content of all registers eventually reaches the $\preceq$-minimum register. 
For example, if all substitutions in $\cA$ are of the form $\sigma(x) = x \add k$ for some $k \in \SR$, then each register will store just its own content during the whole run. 
Intuitively, in this example each register is ``stable'' with respect to the content flow of $\cA$, since each register never passes its value to lower registers.  
This idea motivates the notion of \emph{stable registers}, which are essential to understand the behavior and output of copyless CRA.
Let $\sigma \in \subs(\cA)$ be a copyless substitution in normal form.
We say that a register $x$ is $\sigma$-\emph{stable} (or stable on $\sigma$) if $x \in \var(\sigma(x))$.
More general, we write that a register $x$ is \emph{stable} in $\cA$ if  $x$ is $\sigma$-stable for all substitutions~$\sigma \in \subs(\cA)$.

Stable registers play a crucial role in the behavior of copyless CRA. 
Indeed, we will show that they are the only registers whose value always depends on the whole word, namely, we can always ``reset'' the value of non-stable registers to a constant. For instance, the automaton $\cA_1$ in Example~\ref{ex:max-b-substrings} resets the register $y$ to 0 each time the symbol $a$ is read. On the other side, the register $x$ is stable and it cannot be reset to a constant. In fact its value only grows or remains the same during the run of $\cA_1$.

We formalize this idea of reseting the content of registers as follows: a substitution $\sigma \in \subs(\cA)$ is a \emph{reset} substitution if $\var(\sigma(x)) = \emptyset$ for all non $\sigma$-stable registers $x$. 
We say that a substitution $\sigma \in \subs(\cA)$ is an \emph{$\cA$-reset substitution} if $\var(\sigma(x)) = \emptyset$ for all non-stable registers $x$~in~$\cA$.

We start with the following lemma that shows that the composition preserves stability between registers.
\begin{lemma}\label{lemma:stability}
Let $\sigma, \sigma'$ be two copyless substitution in normal form. For any register $x \in \cX$ it holds that $x$ is stable on $\sigma$ and $\sigma'$ if, and only if, $x$ is $(\sigma \circ \sigma')$-stable. 
\end{lemma}
\begin{proof}
Suppose that $x$ is stable on $\sigma$ and $\sigma'$. This means that $x \in \Var(\sigma(x))$ and $x \in \Var(\sigma'(x))$ and, thus, $x \in \Var(\sigma\circ \sigma'(x))$ by composition, implying that $x$ is $(\sigma\circ \sigma')$-stable. 
For the other direction, suppose that $x \in \Var(\sigma \circ \sigma'(x))$.
Then we know that there exists $y$ such that $x \in \Var(\sigma(y))$ and $y \in \Var(\sigma'(x))$.
Since $\sigma$ and $\sigma'$ are in normal form, then $x \preceq y \preceq x$ which implies that $x = y$ and, thus, $x$ is stable on $\sigma$ and $\sigma'$.
\end{proof}
%
By the previous lemma, it is enough to check that a register $x$ is $\sigma$-stable for all transitions $\delta(p,a) = (q,\sigma)$ to know whether $x$ is stable in all substitutions from $\subs(\cA)$.

For $Q' \subseteq Q$ we say that $Q'$ is a \emph{bottom strongly connected component} (BSCC) of $\cA$ if (1) for every pair $q_1, q_2 \in Q'$ there exists $w \in \AL^*$ such that $\delta^*(q_1, w) = (q_2, \sigma)$ for some substitution $\sigma$ and (2) for every $q_1, q_2 \in Q'$ and $w \in \AL^*$ if $\delta^*(q_1, w) = (q_2, \sigma)$ then $q_2 \in Q'$.
Intuitively a BSCC $Q'$ of $\cA$ is a set of mutually reachable states such that there is no word that leaves $Q'$. 
We say that $\cA$ is \emph{strongly connected} if the whole set $Q$ is a BSCC of $\cA$.

\begin{proposition}
	\label{prop:reset_register}
	Let $\cA$ be a copyless and strongly connected CRA in normal form.
	Then for all $q,q' \in Q$ there exists $w^{q,q'} \in \Sigma^*$ and a substitution $\sigma$ such that $\delta^*(q,w^{q,q'}) = (q', \sigma)$ and $\sigma$ is an $\cA$-reset substitution.
	Furthermore, there exists $w^{q,q'}$ containing all letters in $\Sigma$.
\end{proposition}
For instance in Example~\ref{ex:max-b-substrings} it suffices to take the word $w = ba$ or any word that contains $a$, given that the substitution defined by the $a$-transition is an $\cA$-reset substitution. For simplicity if $q = q'$, we write $w^{q}$ and $\sigma^{q}$ instead of $w^{q,q}$ and $\sigma^{q,q}$. Note that the additional requirement that $w^{q,q'}$ contains all letters in $\Sigma$ is rather technical and its usefulness will be clear later in Section~\ref{sec:nonexpressibility}.

\begin{proof}[Proof (of Proposition~\ref{prop:reset_register})]
	Let $x_1, \ldots, x_n$ be all registers in $\cX$ in the increasing order with respect to $\preceq$. 
	We construct the word $w^{q,q'}$ by (inverse) induction starting from $x_n$ and ending in $x_1$.
	Specifically, for every $i \leq n$ we define a word $w_i^{q, q'}$ and a substitution $\sigma^{q,q'}_i$ such that the proposition holds for all non-stable registers $x$ such that $x_i \preceq x$.
	Clearly, the proposition  will be shown by defining $w^{q, q'} = w^{q,q'}_1$.
	
	We start with the base case $i = n$ and consider whether $x$ is stable or not.
	If $x_n$ is stable, then take a word $u$ and a substitution $\sigma$ such that $\delta^*(q,u) = (q', \sigma)$ and $u$ contains each letter in $\Sigma$.
	Given that $\cA$ is strongly connected we know that $u$ and $\sigma$ always exists. 
	Then by defining $w^{q,q'}_n = u$ and $\sigma^{q,q'}_n = \sigma$, the proposition holds for the stable register $x_n$.
	Now, suppose that $x_n$ is non-stable which means that there exist a pair $p, p' \in Q$ and a word $u$ such that $\delta^*(p, u) = (p', \sigma)$ and $x_n$ is non-stable in $\sigma$.
	Given that $\cA$ is in normal form and $x_n$ is the maximum register with respect to $\preceq$, this implies that $\var(\sigma(x_n)) = \emptyset$.
	Pick two words $v_1, v_2 \in \Sigma^*$ such that $\delta^*(q, v_1) = (p, \sigma_1)$ and $\delta^*(p', v_2) = (q', \sigma_2)$ for some substitutions $\sigma_1$ and $\sigma_2$, and $v_1$ contains each letter in~$\Sigma$.
	Again, we know that $v_1$ and $v_2$ always exists since $\cA$ is strongly connected.
	Then define  $w^{q,q'}_n = v_1 \cdot u \cdot v_2$ and  $\sigma^{q,q'}_n = \sigma_1 \cdot \sigma \cdot \sigma_2$.
	By construction, we know that $\delta^*(q, w^{q,q'}_n) = (q', \sigma^{q,q'}_n)$ and $w^{q,q'}_n$ contains all letters in $\Sigma$. 
	To prove that $\var(\sigma^{q,q'}_n(x_n)) = \emptyset$, notice that $x_n$ is non-stable on $\sigma$.
	By Lemma~\ref{lemma:stability} this implies that $x_n$ is non-stable on $\sigma_1 \circ \sigma \circ \sigma_2$.
	Given that $\cA$ is in normal form and $x_n$ is the maximal register, we get that $\var(\sigma^{q,q'}_n(x_n)) = \emptyset$.
	
	For the inductive step, assume that $w^{q,q'}_{i+1}$ and $\sigma^{q,q'}_{i+1}$ exist. We show how to construct $w^{q,q'}_i$ and $\sigma^{q,q'}_i$ that satisfy the proposition for registers greater or equal than $x_i$. 
	Again, we consider two cases depending on whether $x_i$ is stable or not. 
	If $x_i$ is stable, then by defining $w^{q,q'}_i = w^{q,q'}_{i+1}$ and $\sigma^{q,q'}_i = \sigma^{q,q'}_{i+1}$ the inductive step trivially holds.
	Suppose that $x_i$ is non-stable and, therefore, there exist a pair  $p, p' \in Q$ and a word $u$ such that $\delta^*(p, u) = (p', \sigma)$ and $x_i$ is non-stable on $\sigma$.
	Let $v_1, v_2 \in \Sigma^*$ such that $\delta^*(q', v_1) = (p, \sigma_1)$ and $\delta^*(p', v_1) = (q', \sigma_2)$ for some substitutions $\sigma_1$ and $\sigma_2$.
	Recall that $v_1$ and $v_2$ exist because $\cA$ is strongly connected. 
	Now, define:
	\[
	\renewcommand{\arraystretch}{1.4}
	\begin{array}{rcl}
	w^{q,q'}_i & = & w^{q,q}_{i+1} \cdot v_1 \cdot u \cdot v_2 \\
	\sigma^{q,q'}_i & = & \sigma^{q,q}_{i+1} \circ \sigma_1 \circ \sigma \circ \sigma_2
	\end{array}
	\]
	It is clear by construction that  $\delta^*(q, w^{q,q'}_i) = (q', \sigma^{q,q'}_i)$ and $w^{q,q'}_i$ contains all letter in $\Sigma$
	(because already $w^{q,q}_{i+1}$ contains all letters in $\Sigma$). 
	To conclude the proof, we show that $\var(\sigma^{q,q'}_i(x)) = \emptyset$ for every non-stable register $x \succeq x_i$. 
	Let $x$ be any non-stable register $x \succeq x_i$ (possibly $x_i$) and let $\sigma^* = \sigma_1 \circ \sigma \circ \sigma_2$. 
	First, note that all $y \in \Var(\sigma^*(x))$ are non-stable.
	Otherwise, if $y \in \Var(\sigma^*(x))$ is stable, then $y \in \Var(\sigma^*(y))$, which is a contradiction because $\Var(\sigma^*(x)) \cap \Var(\sigma^*(y)) = \emptyset$ by the definition of being copyless.  
	Therefore, we have that every register in $\Var(\sigma^*(x))$ is non-stable.
	Note also that $x_i \notin \Var(\sigma^*(x))$.
	This is obvious when $x \neq x_i$ because $\cA$ is in normal form. Otherwise, $x = x_i$ and we know that $x_i \notin \Var(\sigma^*(x_i))$ because $x_i$ is non-stable.
	Then we have that all registers in $\Var(\sigma^*(x))$ are non-stable and strictly greater than $x_i$. 
	By the induction assumption $\Var(\sigma^{q,q}_{i+1}(y)) = \emptyset$ for all $y \in \Var(\sigma^*(x))$. By composing $\sigma^{q,q}_{i+1}$ and $\sigma^*$, we conclude that $\Var(\sigma^{q,q}_{i}(x)) = \emptyset$.
\end{proof}


\subsection{Growing rate of stable registers in a cycle}
The behavior of cycles in a computation model is always important; most of the decidability results can be derived from a good understanding of its cyclic behavior. 
Here, we study how the content of stables registers behaves through cycles.
We say that a word $w \in \Sigma^*$ is a cycle over a state $q \in Q$ in $\cA$ if $\delta^*(q, w) = (q, \sigma)$ for some substitution~$\sigma$. 
Of course, the iteration of a cycle $w$ (i.e. $w^n$ for any $n \geq 1$) is also a cycle over $q$ and it satisfies $\delta^*(q, w^n) = (q, \sigma^n)$ for any $n$, where $\sigma^n$ is $\sigma$ composed $n$-times with itself. 
We will need the next result to show that by iterating cycles one can always ``reset'' the content of non $\sigma$-stable registers.  
\begin{lemma} \label{lemma:reseting-loop}
Let $\cA$ be in normal form and $\sigma \in \subs(\cA)$ a copyless substitution. There exists $N \leq |\cX|$ such that $\sigma^N$ is a reset substitution. 
\end{lemma}

\begin{proof}
Suppose $x$ is non-stable over $\sigma$, i.e. $x \not \in \var(\sigma(x))$, and assume that $\var(\sigma(x)) \neq \emptyset$. Then consider the register $y_1 =  \min(\var(\sigma(x)))$. Since $\cA$ is copyless and in normal form then $x \prec y_1$. But since $y_1 \in \var(\sigma(x))$ and $\cA$ is copyless then $y_1 \not \in \var(\sigma(y_1))$ (i.e. $y_1$ cannot appear twice in $\sigma$). 
Recall that $\var(\sigma(y_1)) \subseteq \var(\sigma^2(x))$ given that $y_1 \in \var(\sigma(x))$.
Combining both facts we get that $y_1 \not \in \var(\sigma^2(x))$ and because $\cA$ is in normal form $y_1 \prec \min(\var(\sigma^2(x)))$. 

So far, we have proved that $x \prec y_1 =  \min(\var(\sigma(x))) \prec \min(\var(\sigma^2(x)))$ and $x \not \in \var(\sigma^2(x))$. Also, we know that $\sigma^2$ is a copyless substitution in normal form. Then, we can repeat the same argument above for $x$ and $\sigma^2$, forming an increasing sequence of registers:
\[x \prec y_1 = \min(\var(\sigma(x))) \prec y_2 = \min(\var(\sigma^2(x))) \prec y_3 = \min(\var(\sigma^3(x))) \prec  \dots\]
The number of registers is finite so this sequence cannot be infinite. Thus there exists an $N$ such that $\var(\sigma^N(x)) = \emptyset$.
It is straightforward that we can take $N \leq |\cX|$.
\end{proof}
By Lemma~\ref{lemma:reseting-loop} we know that a non-stable register $x$ over $\sigma$ becomes a constant when $\sigma$ is iterated at least $|\cX|$-times.
In the next proposition, we study the growing behavior of stable registers when a reset substitution is iterated. This will be useful to understand the behavior of copyless CRA inside their~cycles.
For this result we need additional assumptions that automata do not use $\zero$ in their expressions and that the semiring $\SR$ is non-zero (see Section~\ref{subsection:zeros} for details).
\begin{proposition}
	\label{prop:long_exp}
	Let $\SR$ be a non-zero semiring and let $\cA$ be a non-zero automaton
	in normal form. Let $\sigma \in \subs(\cA)$ be a reset substitution and $x$ a $\sigma$-stable register. Then there exist $c,d \in \SR$ with $c \neq \zero$ such that for every $i \geq 0$ we have:
	$$
	\sigma^{i+1}(x) \;\; = \;\; (c^{i} \mult \sigma(x)) \add \big( d \mult \bigadd_{j=0}^{i-1} c^j \big)
	$$
\end{proposition}

Before proving Proposition~\ref{prop:long_exp}, we must show a straightforward fact about copyless expressions.
\begin{lemma}
	\label{lemma:easy_exp}
	Suppose the semiring $\SR$ is a non-zero semiring and $\cA$ is a non-zero automaton.
	Let $\sigma, \tau \in \subs(\cA)$ where $\sigma$ is a reset substitution and let $x$ be a $\sigma$-stable and $\tau$-stable register.
	Then the expression $\sigma \circ \tau(x)$ is equivalent to an expression of the form $(c \mult \sigma(x)) \add d$, where $c, d \in \SR$ and $c \neq \zero$.
\end{lemma}

\begin{proof}
	We prove this by induction on the length of expression $\tau(x)$. For the base step take $\tau(x) = x$. Then $\sigma \circ \tau(x) = \sigma(x)$ is equivalent to $(\one \mult \sigma(x)) \add \zero$. Suppose we can write such an expression $(c \mult x) \add d$ for $\tau$ of length at most $n$. By the inductive assumption when $\tau$ is of length $n+1$ then $\sigma \circ \tau(x)$ can be written in the form $((c \mult \sigma(x)) \add d) \add \sigma(e)$ or $((c \mult \sigma(x)) \add d) \mult \sigma(e)$ for some expression $e$. 
	Since $\tau(x)$ is copyless then $\var(e) \subseteq \cX -\{x\}$.
	Therefore, every $y \in \var(e)$ is non-stable in $\cA$ and thus $\sigma(e) = d'$ for some constant $d' \neq \zero$ (by the non-zero assumptions on $\cA$ and $\SR$). 
	Thus, $\sigma \circ \tau(x)$ can be rewritten as $(c \mult \sigma(x)) \add (d \add d')$ or $(((d' \mult c)\mult \sigma(x)) \add (d' \mult d))$, respectively.
\end{proof}

\begin{proof}[Proof (of Proposition~\ref{prop:long_exp})]
	Since $\sigma \in \subs(\cA)$ is a copyless and reset substitution, then for every $y \in \var(\sigma(x))$ either $y = x$ or $\var(\sigma(y)) = \emptyset$. 
	The expression $e = \sigma^2(x)$ can be seen as $\sigma(x)$ where every register $y \neq x$ is replaced with $\asem{\sigma(y)}$. This is a copyless expression with only one variable $x$ and it does not use $\zero$ since $\cA$ is non-zero. By Lemma~\ref{lemma:easy_exp} the expression $e = \sigma \circ \sigma(x)$ (setting $\tau = \sigma$) can be rewritten in the form $e^* = (c \mult \sigma(x)) \add d$ for some $c,d \in \sr$ and $c \neq \zero$.
	
	We prove the proposition by induction using the constants $c,d$ from $e^*$. For $i=0$ the claim is trivially true and for $i = 1$ the expression $e^*$ is of the desired form.
		For the inductive step, we have $\sigma^{i+2}(x) = \sigma^i \circ \sigma^2(x) = \sigma^{i} \circ e^*$. 
	Then:
	\[
	\renewcommand{\arraystretch}{1.3}
	\begin{array}{rcll}
	\sigma^{i+2} (x) & \; = \;\; & \sigma^i \circ ((c \mult \sigma(x)) \add d) & \text{(because $e^* = (c \mult \sigma(x)) \add d$)}\\
	& \; = \;\; & (c \mult \sigma^{i+1}(x)) \add d \\
	 & \; = \;\; &  \big(c \mult \big( (c^{i} \mult \sigma(x)) \add \big( d \mult \bigadd_{j=0}^{i-1} c^j \big) \big) \big)  \add d & \text{(by induction)} \\
	& \; = \;\; &  \big((c^{i+1} \mult \sigma(x)) \add \big( d \mult \bigadd_{j=1}^{i} c^j \big) \big) \add d \\
	& \; = \;\; &  (c^{i+1} \mult \sigma(x)) \add \big( d \mult \bigadd_{j=0}^{i} c^j \big)  
	\end{array}
	\]
 \end{proof}


Proposition~\ref{prop:long_exp} shows how $\sigma$-stable registers grow with respect of the number of times that a cycle is iterated.
In particular, when $\SR = \artic$ then $\sigma$-stable registers grows linearly because $\mult$ operation is $+$ in this semiring. 

The next result is a refinement of Proposition~\ref{prop:long_exp} for stable registers but in terms of the $\artic$-semiring. This result will be crucial in the proof of Theorem~\ref{theorem:counterexample}.
Recall that, by Proposition~\ref{prop:reset_register}, for any $q \in Q$ there exists a word $w^q$ such that $\delta^*(q,w^q) = (q, \sigma^q)$ and $\sigma^q$ is an $\cA$-reset substitution.

\begin{lemma}
\label{lemma:loops}
Let $\cA$ be a copyless, strongly connected and non-zero CRA in normal form over the $\artic$ semiring. Furthermore, let $q \in Q$ and $v \neq \epsilon$ be a cycle in $q$, namely, $\delta^*(q,v) = (q, \tau)$ for some reset substitution~$\tau$.
For every $j \in \nat$ let $\sigma_j$ be the substitution such that $\delta^*(q,w^q \cdot v^{j+1} \cdot w^q) = (q, \sigma_j)$. Then for every $x \in \cX$, $\lambda \in \subs(\cA)$ and $j$ big enough the following holds:
\begin{align*}
\sigma_j \circ \lambda (x) \;\; = \;\; 
\begin{cases}
 \Oo(1) & \text{if $x$ is non-stable}\\
 \max\{ \ j \cdot c + \sigma^q(x) + \Oo(1), \ j \cdot d + \Oo(1) \ \} & \text{otherwise,}
 \end{cases}
\end{align*}
where $c, d \in \nat$ are constants that depend only on $x$.
\end{lemma}

The additional substitution $\lambda$ will be necessary in Section~\ref{sec:nonexpressibility}. Intuitively, we will compose many substitutions and $\lambda$ will represent a substitution from previous compositions. 
\begin{proof}
Recall that we work with the $\artic$ semiring. Let $q \in Q$ and let $v \neq \epsilon$ be a cycle such that $\delta^*(q,v) = (q, \tau)$ for some reset substitution $\tau$.
By Proposition~\ref{prop:long_exp} we have that for every $j \geq 1$ it holds that:
\setlength{\jot}{4pt}
\begin{align}
\tau^{j+1}(x) & \;\; = \;\;  (c_x^{j} \mult \tau(x)) \add \big( d_x \mult \bigadd_{k=0}^{j-1} c_x^k \ \big) && \text{(by Prop.~\ref{prop:long_exp})}   \nonumber \\
& \;\; = \;\;  \max\left\{ \ j \cdot c_x + \tau(x), d_x + \max_{k=0}^{j-1}\{ k \cdot c_x \} \  \right\}  && \text{(by definition of $\artic$)} \nonumber \\
& \;\; = \;\;  \max\left\{ \ j \cdot c_x + \tau(x), \ d_x + (j-1) \cdot c_x \  \right\}  && \text{(by definition of $\max$)} \nonumber \\
& \;\; = \;\;  \max\left\{ \ j \cdot c_x + \tau(x) + \Oo(1), \ j \cdot c_x + \Oo(1) \ \right\}   &&  \text{(by using the $\Oo$-notation)} \label{eq:artic-version}
\end{align}
Recall that $\cA$ is a non-zero CRA and in this semiring $\zero = - \infty$, thus $c_x \geq 0$. If $d_x = - \infty$ then the last equality is not immediate. It holds because:
\begin{align*}
\max\left\{ \ j \cdot c_x + \tau(x), \ d_x + (j-1) \cdot c_x  \right\} = j \cdot c_x + \tau(x) = \max\left\{ j \cdot c_x + \tau(x), j \cdot c_x \ \right\} = \\
\max\left\{ \ j \cdot c_x + \tau(x) + \Oo(1), \ j \cdot c_x + \Oo(1) \ \right\}
\end{align*}

Fix now a substitution $\lambda$.
Recall that $w^{q}$ and $\sigma^q$ are the word and the substitution from Proposition~\ref{prop:reset_register} for~$q$. 
Then for any $j\geq 1$, consider the substitution $\sigma_j$ such that $\delta^*(q,w^q \cdot v^{j+1} \cdot w^q) = (q, \sigma_j)$.
We prove that $\sigma_j \circ \lambda(x)$ satisfies the lemma for $j$ big enough.
For a non-stable register $x$ in $\cA$ the result is straightforward. 
Indeed, $\cA$ is in normal form, and thus $\var(\lambda(x))$ contains just non-stable registers. 
This implies that:
\[
\renewcommand{\arraystretch}{1.3}
\begin{array}{rcll}
\sigma_j \circ \lambda(x) & \;\; = \;\; & \sigma^q \circ \tau^j \circ \sigma^q \circ \lambda (x) \;\;\;  & \text{(by definition)} \\
& \;\; = \;\; & \sigma^q \circ \lambda (x) & \text{($\var(\lambda(x))$ contains only non-stable registers)} \\ 
& \;\; = \;\; & \Oo(1) & \text{($\sigma^q$ and $\lambda$ do not depend on $j$)}
\end{array}
\]
Suppose now that $x$ is a stable register in $\cA$. We need to show that:
\begin{align}
\sigma_j \circ \lambda (x) & \;\; = \;\;  \max\{ \ j \cdot c  + \sigma^q(x)+ \Oo(1), \ j \cdot d + \Oo(1) \ \} \label{eq:stable-cases}
\end{align}
for $j$ big enough and some constants $c, d \in \nat$ such that $c,d$ do not depend on $\lambda$. 

The expression $\sigma_j \circ \lambda (x)$ is equivalent to $\sigma^q \circ \tau^{j+1} \circ \sigma^{q} \circ \lambda(x)$. In other words, it is equivalent to the expression $\sigma^{q} \circ \lambda(x)$ where all registers $y \in \var(\sigma^{q} \circ \lambda(x))$ are substituted with $\sigma^{q} \circ \tau^{j+1}(y)$. We start the proof by analyzing these expressions.

Let $y \in \var(\sigma^{q} \circ \lambda (x))$. If $y$ is not a $\tau$-stable register, then $\var(\tau(y)) = \emptyset$ and thus $\sigma^q \circ \tau^{j+1}(y) = \Oo(1)$. 
Otherwise, by (\ref{eq:artic-version}) we know that $\tau^{j+1}(y) =  \max\left\{ j \cdot c_y + \tau(x) + \Oo(1), \ j \cdot c_y + \Oo(1) \right\}$ for some $c_y \in \nat$ and, thus,  by applying $\sigma^q$ on $\tau^{j+1}(y)$ we get: 
\begin{align}
\sigma^{q} \circ \tau^{j+1}(y) & \;\; = \;\;  \max\left\{ \ j \cdot c_y + \sigma^{q} \circ \tau(x) + \Oo(1), \ j \cdot c_y + \Oo(1) \ \right\} \label{eq:artic-version2}
\end{align}
If $y$ is a $\tau$-stable register, but non-stable on $\cA$ (i.e. non-stable in general) then $\sigma^{q} \circ \tau(y)$ is equal to a constant and does not depend on $j$. Thus we can estimate $\sigma^{q} \circ \tau(y)$ by $\Oo(1)$ and then (\ref{eq:artic-version2}) becomes:
\[
\renewcommand{\arraystretch}{1.3}
\begin{array}{rcl}
\sigma^{q} \circ \tau^{j+1}(y) & \;\; = \;\; & \max\left\{ \ j \cdot c_y + \Oo(1), \ j \cdot c_y + \Oo(1) \ \right\} \\
 & \;\; = \;\; & j \cdot c_y + \Oo(1)
\end{array}
\]
Otherwise, $y$ is a stable register and, moreover, $\sigma^q$ is a reset substitution. Then by Lemma~\ref{lemma:easy_exp} we know that $\sigma^q \circ \tau(y)$ can be represented as:
\setlength{\jot}{4pt}
\begin{align}
\sigma^{q} \circ \tau(y) & \;\; = \;\; \max\{ \ c' + \sigma^{q}(y), \  d' \ \}  \nonumber \\
& \;\; = \;\;  \max\{ \ \sigma^{q}(y) + \Oo(1), \ \Oo(1) \ \} \label{eq:stable-simple-version}
\end{align}
for some constants $c',d' \in \nat$ not depending on $j$ and where $c' \neq -\infty$. Thus, by combining~(\ref{eq:artic-version2}) and~(\ref{eq:stable-simple-version}) we~get:
\setlength{\jot}{4pt}
\begin{align}
\sigma^{q} \circ \tau^{j+1}(y) & \;\; = \;\;  \max\left\{ \ j \cdot c_y + \sigma^{q} \circ \tau(x) + \Oo(1), \ j \cdot c_y + \Oo(1) \ \right\} && \text{(by (\ref{eq:artic-version2}))} \nonumber \\
& \;\; = \;\;  \max\left\{ \ j \cdot c_y + \max\{ \ \sigma^{q}(y) + \Oo(1), \ \Oo(1) \ \} + \Oo(1), \ j \cdot c_y + \Oo(1) \ \right\} && \text{(by (\ref{eq:stable-simple-version}))} \nonumber \\
& \;\; = \;\;  \max\left\{ \ j \cdot c_y +  \sigma^{q}(y) + \Oo(1), \ j \cdot c_y + \Oo(1), \ j \cdot c_y + \Oo(1) \ \right\} && (\text{by distribution}) \nonumber \\
& \;\; = \;\;  \max\left\{ \ j \cdot c_y +  \sigma^{q}(y) + \Oo(1), \ j \cdot c_y + \Oo(1) \ \right\} \label{equation:stable_register}
\end{align}
Observe that a variable $y \in \var(\sigma^{q} \circ \lambda (x))$ is stable if, and only if, $y=x$. We summarize the three possible cases in the next equation:
\begin{align}
\sigma^{q} \circ \tau^{j+1}(y)  \;\; = \;\;
\begin{cases}
\;\; \Oo(1) & \text{if $y$ is not $\tau$-stable} \\
\;\; j \cdot c_y + \Oo(1) & \text{if $y$ is non-stable} \\
\;\; \max\left\{ j \cdot c_y +  \sigma^{q}(y) + \Oo(1), \ j \cdot c_y + \Oo(1) \right\}  & \text{if $x=y$}
\end{cases}
\label{eq:sigma-tau-equation} 
\end{align}

Now we can prove (\ref{eq:stable-cases}). Recall that the expression $\sigma_j \circ \lambda (x)$ is the expression $\sigma^{q} \circ \lambda (x)$ where all registers $y \in \var(\sigma^{q} \circ \lambda(x))$ are substituted with $\sigma^q \circ \tau^{j+1}(y)$.
First notice that by Lemma \ref{lemma:easy_exp} the expression $\sigma^{q} \circ \lambda(x)$ is equivalent to $\max\{ c'' + \sigma^q(x), d''\} = \max\{\sigma^q(x) + \Oo(1), \Oo(1)\}$.  We can do these estimations because $\lambda$ does not depend on $j$.
By combining this equation with the definition of $\sigma_j$ we have:
\begin{align*}
\sigma_j \circ \lambda (x) & \;\; = \;\; \sigma^q \circ \tau^{j+1} \circ \sigma^q \circ \lambda (x) \\
& \;\; = \;\; \sigma^q \circ \tau^{j+1} \circ \max\{ \ \sigma^q(x) + \Oo(1), \ \Oo(1) \ \} \\
& \;\; = \;\; \max\{ \ \sigma_j(x) + \Oo(1), \ \Oo(1) \ \}
\end{align*}
Thus to finish the proof it suffices to show that:
\begin{align}
\sigma_j(x) & \;\; = \;\; \max\{ \ j \cdot c  + \sigma^q(x)+ \Oo(1), \ j \cdot d + \Oo(1) \ \} \label{equation:lemma_induction2}
\end{align}
for some constants $c, d \in \nat$. Indeed, if (\ref{equation:lemma_induction2}) holds, then $\sigma_j(x) + \Oo(1) = \max\{j \cdot c + \sigma^q(x)+ \Oo(1) , j \cdot d + \Oo(1)\}$. Similarly the value $\Oo(1)$ can be estimated by the expression  $j \cdot d + \Oo(1)$ so the last $\max$ operation is not needed. Notice that this does not change the constants $c$ and $d$, which proves that they do not depend on $\lambda$.

To conclude the proof, we show that (\ref{equation:lemma_induction2}) holds for a stable register $x$. 
Recall that $\sigma_j = \sigma^q \circ \tau^{j+1} \circ \sigma^q$. 
We will show that (\ref{equation:lemma_induction2}) holds even if we change $\sigma_j$ to $\sigma^q \circ \tau^{j+1} \circ \sigma'$, where $\sigma'$ is any copyless substitution in normal form and $x$ is $\sigma'$-stable. 
The proof is by induction on the size of $\sigma'(x)$. For the base step we must consider $\sigma'(x) = x$ (because $x$ is stable). Then
\begin{align*}
\sigma_j(x) & \;\; = \;\; \sigma^{q} \circ \tau^{j+1} \circ \sigma'(x) \\
& \;\; = \;\; \sigma^{q} \circ \tau^{j+1} (x) \\
& \;\; = \;\; \max\{ \ j \cdot c_x + \sigma^q(x) + \Oo(1), \ j \cdot c_x + \Oo(1) \ \} && \text{(by (\ref{equation:stable_register}))}.
\end{align*}
For the inductive step, assume that~(\ref{equation:lemma_induction2}) holds for expressions $\sigma'$ of length $n$. We want to show that~(\ref{equation:lemma_induction2}) holds for an expression $\sigma''(x)$ of the form $\sigma''(x) = \sigma'(x) \circledast f$ where $\circledast$ is either $+$ or $\max$, $x$ is $\sigma'$-stable and $f$ is an expression where all registers are non-stable (recall that $\sigma''$ is copyless).
By unraveling $\sigma^q \circ \tau^{j+1} \circ \sigma''(x)$, we get that:
\begin{align*}
\sigma^q \circ \tau^{j+1} \circ \sigma''(x) & \;\; = \;\;  \sigma^q \circ \tau^{j+1} \circ (\sigma'(x) \circledast f) \\
 & \;\; = \;\;  (\sigma^q \circ \tau^{j+1} \circ \sigma'(x)) \circledast (\sigma^q \circ \tau^{j+1} \circ f) \\
 & \;\; = \;\; \max\{ \ j \cdot c  + \sigma^q(x)+ \Oo(1), \ j \cdot d + \Oo(1) \ \} \circledast (\sigma^q \circ \tau^{j+1} \circ f) && (\text{by induction}) 
\end{align*}
In the final expression it is easy to check that $\sigma^q \circ \tau^{j+1} \circ f$ is equal to a constant or to $j \cdot c_f + \Oo(1)$ for some $c_f \neq -\infty$.
Indeed,  by (\ref{eq:sigma-tau-equation}) we know that $\sigma^{q} \circ \tau^{j+1}(y) =  j \cdot c_y + \Oo(1)$ for non-stable registers. Since $f$ is an expression over non-stable registers, one can show by induction over the size of $f$ that  $\sigma^q \circ \tau^{j+1} \circ f$ is of the form $j \cdot c_f + \Oo(1)$ for some constant $c_f \neq -\infty$ (but possibly $c_f = 0$) and $j$ big enough.

We assume that $\sigma^q \circ \tau^{j+1} \circ f = j \cdot c_f + \Oo(1)$. Then we must consider two cases whether $\circledast$ is $+$ or $\max$ operation. For the former we have:
\begin{align*}
\sigma^q \circ \tau^{j+1} \circ \sigma''(x) & \;\; = \;\;  \max\{ \ j \cdot c + \sigma^q(x) + \Oo(1), \ d \cdot j + \Oo(1) \ \} + j \cdot c_f + \Oo(1) \\
& \;\; = \;\; \max\{ \ j \cdot (c + c_f) + \sigma^q(x) + \Oo(1),  \ j \cdot (d + c_f)  + \Oo(1) \ \}, 
\end{align*}
and for the latter we have:
\begin{align*}
\sigma^q \circ \tau^{j+1} \circ \sigma''(x) & \;\; = \;\;  \max\{ \ \max\{ \ j \cdot c + \sigma^q(x) + \Oo(1), \ j \cdot d + \Oo(1) \ \}, \ j \cdot c_f + \Oo(1) \ \} \\
& \;\; = \;\; \max\{ \ j \cdot c + \sigma^q(x) + \Oo(1), \ j \cdot d + \Oo(1), \ j \cdot c_f + \Oo(1) \ \}. \\
& \;\; = \;\; \max\{ \ j \cdot c + \sigma^q(x) + \Oo(1), \ j \cdot \max\{d, c_f\}  + \Oo(1) \ \}. 
\end{align*}
The number of induction steps depends on the size of the expression $\sigma'(x)$ which for $\sigma'(x) = \sigma^q(x)$ does not depend on $j$. This concludes the proof. 
\end{proof}

\section{Inexpressibility of copyless CRA over the natural semiring}
\label{sec:natural_semiring}
In this section we show that there exists a function definable by WA (or linear CRA, see Section~\ref{sec:preliminaries}), which is not definable by copyless CRA over the semiring $\natsemiring$.
For this, we use the structural results of copyless CRA introduced in the previous section.  Consider the class of functions over the one-letter alphabet $\{a\}$ over the semiring $\natsemiring$. Since all words are of the form $a^n$ for $n \in \nat$ we identify the domain of the functions with $\nat$. 
For WA this class of functions is equivalent to linear recurrence systems (c.f. the matrix characterization of weighted automata~\cite{DrosteHWA09}). In particular, one can define the Fibonacci sequence with a linear and non-copyless CRA as follows.
Let $\cF$ be the linear CRA with one state and two registers $x,y$ presented in Figure~\ref{fig:fibonnaci-CRA}. The only transition updates the registers by $x := y$, $y := x+y$ with the initial values $x = y = 1$ and $x$ defined as the output. Then one can easily see that $\asem{\cF}(a^n) = F_n$, where $F_n$ is the $n$-th element in the Fibonacci sequence.

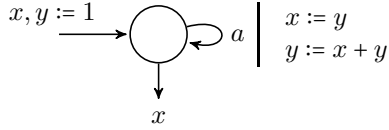
\begin{figure}
	\begin{center}
		\begin{tikzpicture}[-\string>,\string>=stealth',shorten \string>=1pt,auto,node distance=1.8cm,semithick,initial text={},every node/.style={scale=1}]
		\tikzstyle{every state}=[fill=white,draw=black,text=black]
		
		\node[state, draw=white,minimum size=2ex] 		(p0) {};
		\node[state, right of=p0, node distance=1.5cm,minimum size=5ex] 		(p) {};
		
		\draw (p0) edge node {$x,y:=1$ \;\;\;\; \;\;\;\;\;} (p);
		
		\path
		(p)   edge	[loop right]	node {$a \;\;   \begin{array}{|l}
			\;\; x:=y \\
			\;\; y:=x+y
			\end{array}$}  (p);
		
		\node[below of=p, node distance=1.1cm]	(out)  {$x$};
		\draw (p) edge (out);	
		\end{tikzpicture}
	\end{center}
	\caption{A CRA recognizing the Fibonacci sequence $F_n$.} \label{fig:fibonnaci-CRA}
\end{figure}

\begin{theorem}
	\label{theorem:counterexample-nat}
	The Fibonnaci sequence $F_n$ is not recognizable by any copyless CRA.
\end{theorem}
\begin{proof}
We will use some simple facts about Fibonacci numbers that are well known or can be easily extracted from the definitions, for a broader discussion see e.g.~\cite{knuth}. Let $\varphi = \frac{1+ \sqrt{5}}{2}$ be the golden ratio. Then $F_n = \frac{\varphi^n + (1 - \varphi)^n}{\sqrt{5}}$, and the following two properties hold:
\begin{align}
\lim_{n \to +\infty}\frac{F_{n+k}}{F_n} = \varphi^k \text{ for every fixed $k$.}\label{fib1} \\
\varphi^n \not \in \nat \text{ for any $n > 0$.}\label{fib2}
\end{align}

Combining these two properties with results in Section~\ref{sec:structure_copyless} we show that $\asem{\cF}$ cannot be defined by a copyless CRA. For a contradiction assume that there exists a copyless CRA $\cA$ such that $\asem{\cA} = \asem{\cF}$. Since $\cA$ is deterministic and the alphabet has one letter in every state there is always only one transition. Let $q_0$ be the initial state of $\cA$; $\nu_0$ the initial valuation; and $\mu$ its final output function. Then there exists $s,t \in \nat$ and a state $q$ such that $\cA$ loops in state $q$ when reading $a^s$; and $\cA$ moves from state $q_0$ to state $q$ when reading $a^t$. Let $\sigma$ be the substitution defined by reading $a^s$ from state $q$; and $\rho$ the substitution defined by reading $a^t$ from $q_0$. Then $\asem{\cA}(a^{t + i\cdot s}) = \nu_0 \circ \rho \circ \sigma^i \circ \mu(q)$. We will analyze the values of $\sigma^i(x)$ for $x \in \var(\mu(q))$.

By Proposition~\ref{proposition:normal_form} we can assume $\cA$ is in normal form. Then by Lemma~\ref{lemma:reseting-loop} there exists $N$ such that $\sigma^N$ is a reset substitution. Then by Proposition~\ref{prop:long_exp} applied to $\sigma^N$ and a $\sigma^N$-stable register $x$ we have that:
$$
\sigma^{N(i+1)}(x) \;\; = \;\; c^{i} \cdot \sigma^N(x) + d \cdot \sum_{j=0}^{i-1} c^j
$$
for some $c,d \in \nat$.
Now $\sigma^N$, $\nu_0$ and $\rho$ do not depend on $i$ so we can write that for every $\nu_0 \circ \rho \circ \sigma^N(x) = A_x$, where $A_x$ is a constant not depending on $i$. Hence we define the sequence $x(i)$ of values of register $x$
\begin{align}
\label{eq1}
x(i) \;\; = \;\; \nu_0 \circ \rho \circ \sigma^{N(i+1)}(x) \;\; = \;\; c^{i} \cdot A_x + d \cdot \sum_{j=0}^{i-1} c^j.
\end{align}

If $c = 0$ then $g_i = 0$ for all $i$. Otherwise, if $c = 1$ we have $\lim_{i \to + \infty} \frac{x(i+1)}{x(i)} = 1$. In the remaining cases:
$$
x(i) \;\; = \;\; c^{i} \cdot A_x + d \cdot \frac{c^i - 1}{c - 1}
$$
and by taking the limit of $x(i)$ when $i$ tends to infinity we have:
$$
\lim_{i \to + \infty} \frac{x(i+1)}{x(i)} \;\; = \;\; \lim_{i \to + \infty} \frac{c^{i+1} \cdot A_x + d \cdot \frac{c^{i+1} - 1}{c - 1}}{c^{i} \cdot A_x + d \cdot \frac{c^i - 1}{c - 1}} \;\; = \;\; \lim_{i \to + \infty} \frac{c \cdot A_x + c \cdot d \cdot \frac{1 - \frac{1}{c^{i+1}}}{c - 1}}{A_x + d \cdot \frac{1 - \frac{1}{c^i}}{c-1}} = c.
$$

The goal now is to arrive to a contradiction using properties~\eqref{fib1} and~\eqref{fib2}. Stipulating that $n = t + (i+1)\cdot Ns$ (the number of $a$'s to reach $q$ and make $i+1$ many $\sigma^N$ loops)  and $k = Ns$ (the number of $a$'s to make a single $\sigma^N$ loop) we analyze the value of $\lim_{i \to +\infty} \frac{\asem{\cA}(a^{n+k})}{\asem{\cA}(a^{n})}$. By definition $\asem{\cA}(a^{n+ k}) = \nu_0 \circ \rho \circ \sigma^{N(i+1)} \circ \mu(q)$. We can assume $\mu(q) = \sum_l c_l \cdot \prod_j x_{j,l}$ for some constants $c_l \in \nat$ (possibly in this form the registers $x_{j,l}$ can repeat and the expression is not copyless).

Now, $\nu_0 \circ \rho \circ \sigma^{N(i+1)} \circ \mu(q) = \sum_l c_l \cdot \prod_j \nu_0 \circ \rho \circ \sigma^{N(i+1)}(x_{j,l})$. For every $x \in \var(\mu(q))$ if $x$ is stable then its value is $\nu_0 \circ \rho \circ \sigma^{N(i+1)}(x) = x(i)$ as in~\eqref{eq1}. Otherwise, it is a constant not depending on $i$. We substitute all such registers with constants and also all registers such that $x(i) = 0$ are replaced with $0$. Now for all registers $x$ their sequences $x(i)$ have the property that $\lim_{i \to + \infty} \frac{x(i+1)}{x(i)}$ is a natural number. Consider two sequences $x(i)$ and $y(i)$ with such a property. The following are immediate:
\begin{align*}
\lim_{i \to + \infty} \frac{x(i+1)\cdot y(i+1)}{x(i) \cdot y(i)} \;\; = \;\; \lim_{i \to + \infty} \frac{x(i+1)}{x(i)} \cdot \lim_{i \to + \infty} \frac{y(i+1)}{y(i)} \\
\lim_{i \to + \infty} \frac{x(i+1) + y(i+1)}{x(i) + y(i)} \;\; = \;\; \max\left(\lim_{i \to + \infty} \frac{x(i+1)}{x(i)}, \lim_{i \to + \infty} \frac{y(i+1)}{y(i)} \right).
\end{align*}
It follows that $\lim_{i \to +\infty} \frac{\asem{\cA}(a^{n+k})}{\asem{\cA}(a^{n})}$ is a natural number, which is a contradiction with~\eqref{fib1} and~\eqref{fib2}.
\end{proof}

In~\cite{kickasspaper} it is shown that copyless CRA are contained in WA. Given that the Fibonnaci sequence $F_n$ can be defined by a linear CRA and, thus, by a weighted automata, then Theorem~\ref{theorem:counterexample-nat} delimit the expressiveness of copyless CRA over the natural semiring: they are strictly less expressive than WA.

\begin{corollary} 
	\label{corollary:weighted}
	The class of functions defined by copyless CRA over the natural semiring is strictly contained in the class of functions defined by WA over the natural semiring.
\end{corollary}

\section{Inexpressibility of copyless CRA over the max-plus semiring}
\label{sec:nonexpressibility}

Similar as in the previous section, we use here the techniques introduced in Section~\ref{sec:structure_copyless} to show a function definable by a weighted automaton over the max-plus semiring, that is not definable by any copyless CRA over the same semiring.
Apart of delimiting the expressiveness of copyless CRA over the max-plus semiring, this result will show that copyless CRA are not closed under reverse given that the ``reverse'' of this function is definable by copyless CRA.
Before we proceed, notice that the line of reasoning used in the previous section, where functions over one-letter was used, is not possible for automata over the $\artic$ semiring. Indeed, it is known by~\cite{lombardy} that over one-letter alphabets weighted automata over $\artic$ semiring can be turned into unambiguous automata. Since by~\cite{alur2013regular,kickasspaper} we know that unambiguous weighted automata are contained in copyless CRA over any semiring it follows that over one-letter alphabets copyless CRA and weighted automata are equally expressive over the semiring $\artic$.
Therefore, to show inexpressibility of copyless CRA over the semiring $\artic$ a more involved proof is needed over a non-unary alphabet.

\begin{figure}
	\begin{center}
		\begin{tikzpicture}[-\string>,\string>=stealth',shorten \string>=1pt,auto,node distance=1.8cm,semithick,initial text={},every node/.style={scale=1}]
		\tikzstyle{every state}=[fill=white,draw=black,text=black]
		
		\node[state, draw=white,minimum size=2ex] 		(p0) {};
		\node[state, right of=p0, node distance=1.5cm,minimum size=5ex] 		(p) {};
		
		\draw (p0) edge node {$x,y:=0$ \;\;\;\; \;\;\;\;\;} (p);
		
		\path (p)   edge	[loop above]	node {
			$a \;\;   \begin{array}{|l}
			\;\; x:=x+1 \\
			\;\; y:=y
			\end{array}$
		}  (p)
		(p)   edge	[loop right]	node {$\# \;\;   \begin{array}{|l}
			\;\; x:=\max\{x,y\} \\
			\;\; y:=0
			\end{array}$}  (p)
		(p)   edge	[loop below]	node {
			$b \;\;   \begin{array}{|l}
			\;\; x:=x \\
			\;\; y:=y+1
			\end{array}$}  (p);
		
		\node[inner sep=0mm]	(out) at ($(p) + (2,-1.2)$) {$\max\{x,y\}$};
		\draw (p) edge (out.north west);

		\node[state, draw=white,minimum size=2ex, right of=p,node distance=6cm] 		(p0) {};
		\node[state, right of=p0, node distance=1.5cm,minimum size=5ex] 		(p) {};
		
		\draw (p0) edge node {$x,y,z:=0$ \;\;\;\; \;\;\;\;\;} (p);
		
		\path (p)   edge	[loop above]	node {
			$a \;\;   \begin{array}{|l}
			\;\; x:=x+1 \\
			\;\; y:=y \\
			\;\; z:=z
			\end{array}$
		}  (p)
		(p)   edge	[loop right]	node {$\# \;\;   \begin{array}{|l}
			\;\; x:=x \\
			\;\; y:=x \\
			\;\; z:=\max\{z,y\}
			\end{array}$}  (p)
		(p)   edge	[loop below]	node {
			$b \;\;   \begin{array}{|l}
			\;\; x:=x \\
			\;\; y:=y+1 \\
			\;\; z:=z
			\end{array}$}  (p);
		
		\node[inner sep=0mm]	(out) at ($(p) + (2,-1.2)$) {$\max\{z,y\}$};
		\draw (p) edge (out.north west);
		
		\end{tikzpicture}
	\end{center}
	\caption{On the left a copyless cost register automaton~$\cB$ recognizing~$f_\cB$. On the right a cost-register automaton recognizing~$f_\cB^R$. Notice that the latter automaton is not copyless because it uses~$x$ twice when reading~$\#$. On the other hand both automata are linear CRA.} \label{fig:automaton_reverse}
\end{figure}
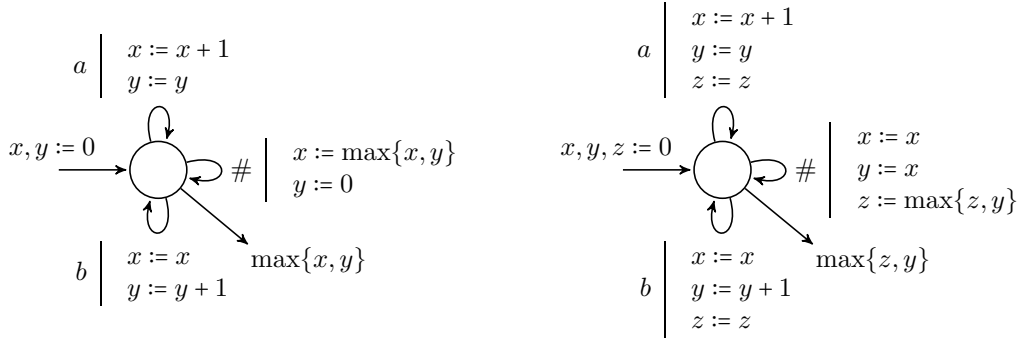

Consider the function $f_\cB$ given by the copyless CRA $\cB$ over $\Sigma = \{a,b,\#\}$ and $\artic$ in Figure~\ref{fig:automaton_reverse}.
To understand $f_\cB$, let us define the output of $\cB$ formally. 
For any $w \in \Sigma^*$, let $k$ be the number of $\#$-symbols in $w$. 
Furthermore, for $0 < i < k$ let $n_i^a$ and $n_i^b$ be the number of $a$'s and $b$'s, respectively, between the $i$-th and ($i+1$)-th occurrence of $\#$ in $w$. 
Additionally, let $n_0^a, n_0^b, n_k^a, n_k^b$ be the numbers of $a$'s and $b$'s before the first and after the last $\#$ in $w$. By the definition of $\cB$ in Figure \ref{fig:automaton_reverse}, one can easily check that $f_\cB$ is defined by $f_\cB(\epsilon) = 0$, for the empty word $\epsilon$, and
\begin{equation} \label{eq:fb-definition}
f_\cB(w) \;\;\; = \;\; \max_{ j \in \{-1, 0, \dots, k\}} \left\{ \ n_j^b + \sum_{i=j+1}^{k} n_i^a \ \right\}
\end{equation}
for $w \neq \epsilon$, where $n_{-1}^b = 0$. From the above definition, one can also give a formal definition of $f_\cB^R$, the reverse function of $f_\cB$, which is given by changing the interval of the index $i$ from $[j+1,k]$ to $[0, j-1]$ in  (\ref{eq:fb-definition}). Formally, one can easily check that $f_\cB^R$ is defined by $f_\cB^R(\epsilon) = 0$, and
\begin{equation} \label{eq:fbR-definition}
f_\cB^R(w) \;\;\; = \;\; \max_{j \in \{0,\dots, k, k+1\}} \left\{ \ \sum_{i=0}^{j-1}n_i^a + n_j^b \ \right\},
\end{equation}
for $w \neq \epsilon$, where $n_{k+1}^b = 0$. The following theorem is the main result of this section.
%

\begin{theorem}
\label{theorem:counterexample}
 The function $f_\cB^R$ is not recognizable by any copyless CRA.
\end{theorem}

\begin{proof}
Suppose there exists a copyless CRA $\cA \; = \; (Q, \AL, \cX, \delta, q_0, \nu_0, \mu)$ which computes the function $f_\cB^R$. 
Since the function $f_\cB^R$ is non-zero (i.e. $f_\cB^R(w) \neq -\infty$ for every $w \in \Sigma^*$) then by Lemma~\ref{proposition:non-zero} we can assume that the transitions, initial and final functions in $\cA$ do not use $-\infty$. Moreover, since the $\artic$ semiring is also non-zero then additionally we can assume that the registers never store $-\infty$.

Most of the time we will assume that $\cA$ is a strongly connected automaton (see Section~\ref{subsection:stable_reset}).
If $\cA$ is not strongly connected, we change our analysis by adding a word $w_0$ from the initial state to a BSCC and then construct a counterexample word from there. 
Formally, let $\delta^*(q_0, w_0) = (q_0', \sigma_0)$ where $q_0'$ is a state inside a BSCC $Q' \subseteq Q$ of $\cA$. We can always redefine $\cA$ as follows: $q_0'$ is the new initial state, and the new initialization function $\nu_0'$ is defined as $\nu_0'(x) = \nu_0 \circ \sigma_0(x)$. 
It is straightforward to check that for every word $w$ the new automaton constructed from $\cA$ will return the output value of $\cA$ over $w_0 \cdot w$. Thus, for the rest of the proof we assume that $\cA$ is strongly connected, and we will include $w_0$ and $\sigma_0$ in the final analysis.
Finally, given that the automaton is trimmed then we can assume that the BSCC contains final states.

We start the proof by analyzing the behavior of $\cA$ on words containing cycles of just one letter.
Since $Q$ is strongly connected, then there exists a state $q_a$, a word $v_a = a^{n_a}$ with $n_a > 0$ and a substitution $\tau_a$ such that $\delta^{*}(q_a, v_a) = (q_a, \tau_a)$. 
We can assume by Lemma~\ref{lemma:reseting-loop} that $\tau_a$ is a reset substitution, that is, $\var(\tau_a(x)) = \emptyset$ for all non-stable registers $x$ in $\tau_a$. 
In addition, define a sequence of words:
\[
w_a(j) \;\; = \;\; w^{q_a} \cdot v_a^{j+1} \cdot w^{q_a}
\]
such that $\delta^*(q_a, w_a(j)) = (q_a, \sigma_a^j)$ for some substitution $\sigma_a^j$ (i.e. $\sigma_a^j$ depends on $j$). 
Recall that $w^{q_a}$ is the reset word defined in Proposition~\ref{prop:reset_register} for the state $q_a$ and $\sigma^{q_a}$ is an $\cA$-reset substitution such that $\delta^*(q_a,w^{q_a}) = (q_a, \sigma^{q_a})$.
By Lemma~\ref{lemma:loops} we know that there exist constants $c_a^x$ and $d_a^x$ such that for $j$ big enough:
\begin{align}
\label{equation:loop_a}
\sigma_a^j \circ \lambda(x) & \;\; = \;\; 
\begin{cases}
 \Oo(1) & \text{if $x$ is non-stable}\\
 \max\{ \ j \cdot c_a^x + \sigma^{q_a}(x) + \Oo(1), \ j \cdot d_a^x + \Oo(1) \ \} & \text{otherwise,}
 \end{cases}
\end{align}
for any $\lambda \in \subs(\cA)$ whose size does not depend on $j$.
Analogously, by the definition of $\cA$ we can find a state $q_b$ in $\cA$, a word $v_b = b^{n_b}$ for some $n_b \geq 0$ and a reset substitution $\tau_b$ such that $\delta^*(q_b, v_b) = (q_b, \tau_b)$.
Then similar to the sequence $w_a(j)$, we can define the sequence of words:
$$w_b(j) = w^{q_b} \cdot v_b^{j+1} w^{q_b}$$
such that $\delta^*(q_b, w_b(j)) = (q_b, \sigma_b^j)$ for a substitution $\sigma_b^j$ that satisfies:
\begin{align}
\label{equation:loop_b}
\sigma_b^j \circ \lambda(x) & \;\; = \;\; 
\begin{cases}
 \Oo(1) & \text{if $x$ is non-stable}\\
 \max\{ \ j \cdot c_b^x + \sigma^{q_b}(x) + \Oo(1), \ j \cdot d_b^x + \Oo(1) \ \} & \text{otherwise,}
 \end{cases}
\end{align}
for any $\lambda \in \subs(\cA)$ whose size does not depend on $j$.

The next step is to understand the growth of stable registers when we repeat the loop $w_a(j)$ several times. 
For any $j \geq 1$ and $s \geq 1$ we define the sequence of words:
\[
w_a(s,j) \;\; = \;\; (w^{q_a} \cdot v_a^{j+1})^s \cdot w^{q_a}.
\]
Let $\delta^*(q_a,w_a(s,j)) = (q_a, \sigma_a^{s,j})$, then by definition $\sigma_a^{s,j} = (\sigma^{q_a} \circ \tau^{j+1} )^s \circ \sigma^{q_a}$.
For any natural number $s \geq 1$ if $x$ is a stable register we prove by induction that
\begin{align}
 \label{equality1}
 \sigma_a^{s,j} \circ \lambda(x) &  \;\; = \;\; 
 \max\{ \ s \cdot j \cdot c_a^x + \sigma^{q_a}(x) + \Oo(s), \   (s -1) \cdot j \cdot c_a^x + j \cdot d_a^x + \Oo(s) \ \},
\end{align}
for any $\lambda \in \subs(\cA)$ whose size does not depend on $j$, and the component hidden in $\Oo(s)$ depends on $s$ but does not depend on $j$.
For $s = 1$~(\ref{equality1}) follows from~(\ref{equation:loop_a}). For $s>1$ we show:
\setlength{\jot}{4pt}
\begin{align*}
& \sigma_a^{s+1,j} \circ \lambda (x)
= \;\; \sigma^{q_a}  \circ \tau^{j+1} \circ \sigma_a^{s,j} \circ \lambda(x) && (\text{by definition})\\
& = \;\; \sigma^{q_a}  \circ \tau^{j+1} \circ \big( \max\{ \ s \cdot j \cdot c_a^x + \sigma^{q_a}(x) + \Oo(s), \ (s-1) \cdot j \cdot c_a^x + j \cdot d_a^x + \Oo(s) \ \}  \big)  && (\text{by induction})\\
& = \;\; \max\{ \ s \cdot j \cdot c_a^x + \sigma_a^j(x) + \Oo(s), \ (s-1) \cdot j \cdot c_a^x + j \cdot d_a^x + \Oo(s) \ \}  \\
& = \;\; \max\{ \ s \cdot j \cdot c_a^x + \max\{ \ j \cdot c_a^x + \sigma^{q_a}(x) + \Oo(1), \ j \cdot d_a^x + \Oo(1) \ \}  + \Oo(s), \\
& \hspace{8cm} (s-1) \cdot j \cdot c_a^x + j \cdot d_a^x + \Oo(s) \ \}&& (\text{by (\ref{equation:loop_a}}))\\
& = \;\; \max\{ \ (s+1) \cdot j \cdot c_a^x + \sigma^{q_a}(x) + \Oo(s), \ s \cdot  j \cdot c_a^x + j \cdot d_a^x + \Oo(s), \\
& \hspace{8cm} (s-1) \cdot j \cdot c_a^x + j \cdot d_a^x + \Oo(s) \ \}  \\
& = \;\; \max\{ \ (s+1) \cdot j \cdot c_a^x + \sigma^{q_a}(x) + \Oo(s), \ s \cdot  j \cdot c_a^x + j \cdot d_a^x + \Oo(s) \ \}  && (\text{by dominance})\\
\end{align*}

We are ready to define the word for which we prove that $\cA$ gives the wrong output.
Recall that $w_0$ is a word such that $\delta^*(q_0, w_0) = (q_0', \sigma_0)$, where $q_0'$ is in the BSCC of $\cA$. Without loss of generality we assume that $q_0' = q_a$.
Let also $w^{q_a, q_b}$ and $w^{q_b, q_a}$ be the reset words between $q_a$ and $q_b$ from Proposition~\ref{prop:reset_register} (given that we are inside a BSCC of $\cA$, we can assume that $q_b$ is accessible from $q_a$ and viceversa).
For all $j \geq 0$ we define the sequence of words:
\begin{align}\label{eq:final_word}
w(s,j) \;\; = \;\; w_0 \cdot w_a(s,j) \cdot w^{q_a, q_b} \cdot w_b(j^2) \cdot w^{q_b, q_a} \cdot w_a(j)
\end{align}
To understand the construction of $w(s,j)$, consider the following diagram, which is a fragment of $\cA$.
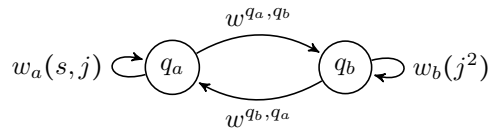
\begin{center}
\begin{tikzpicture}[-\string>,\string>=stealth',shorten \string>=1pt,auto,node distance=2.3cm,semithick,initial text={},scale=0.5,every node/.style={scale=1}]
\tikzstyle{every state}=[fill=white,draw=black,text=black,style={font=}]

\node[state,minimum size=2ex] 		(p) {$q_a$};
\node[state,minimum size=2ex] 		[right of=p](s) {$q_b$};

\path (p)   edge	[loop left]		node {$w_a(s,j)$}  (p)
(p)   edge	[bend left]	node {$w^{q_a,q_b}$}  (s)
(s)   edge	[bend left]	node {$w^{q_b,q_a}$}  (p)
(s)   edge	[loop right]	node [align=center]{$w_b(j^2)$}(s);
\end{tikzpicture}
\end{center}
The automaton $\cA$ starts from reading $w_0$ to reach the state $q_a$. Then it cycles in state $q_a$ reading $w_a(s,j)$. After that it moves to $q_b$ reading the reset word $w^{q_a,q_b}$. Then it cycles in state $q_b$ reading $w_b(j^2)$. Finally, it comes back to the state $q_a$, where it cycles again.
We show that $\cA$ outputs wrong value over $w(s,j)$ for some fixed $s$ (to be determined later) and for $j$ big enough. 

First we estimate the correct value $f_\cB^R(w(s,j))$ for any $s,j \geq 0$. By definition the words $w_a(s, j)$, $w_a(j)$ and $w_b(j^2)$ contain blocks of $a$'s and $b$'s separated by reset words.
By Proposition~\ref{prop:reset_register} each reset word contains  the letter $\#$, thus the cycles of $a$'s and $b$'s in the words $w_a(s, j)$, $w_a(j)$, and $ w_b(j)$ are all separated by $\#$. From the definition of these cycles, one can easily check that the number of $a$'s in $w_a(j)$ is $j \cdot n_a + \Oo(1)$ and the number of $b$'s in $w_b(j^2)$ is $j^2 \cdot n_b + \Oo(1)$. 
Furthermore, the number of $a$'s in $w_a(s,j)$ is equal to $s \cdot j \cdot n_a + \Oo(1)$. 
Notice that the only fragments in $w(s,j)$ whose size depends on $j$ are these cycles. The reset words between the cycles are of constant size and there is $\Oo(s)$ of them.
Recall that by~(\ref{eq:fbR-definition})
$f_\cB^R(w) = \max_{j} \left\{ \ \sum_{i=0}^{j-1}n_i^a + n_j^b \ \right\}$,
where $n_i^a$ and $n_i^b$ are the numbers of $a$'s and $b$'s, respectively, between the $\#$-letters. It is easy to see that for $j$ big enough:
\begin{align}
\label{eq:real-output-function}
f_\cB^R(w(s,j)) & \;\; = \;\; n_b \cdot j^2 + n_a \cdot s \cdot j  + \Oo(s)
\end{align}
where $\Oo(s)$ represents the fixed number of $a$'s (their number does not depend on $j$) that are presented in $w_0$, $w^{q_a}$, $w^{q_a,q_b}$ or $w^{q_b}$. 
Notice that the last suffix $w_a(j)$ of $a$'s is not contributing into the sum.
This is because the sequence $w_b(j^2)$ is overshadowing the last sequence $w_a(j)$, i.e., the max-operator considers the number of $b$'s in $w_b(j^2)$ instead of the number of $a$'s in $w_a(j)$.
The rest of the proof is to show that $\cA$ does not output the right value on $w(s,j)$. 
The intuition behind this misbehavior of $\cA$ with respect to $w(s,j)$ is that if $\cA$ is summing the sequence of $a$'s before $w_b(j^2)$ then it will also add the sequence of $a$'s after $w_b(j^2)$, which by the previous calculations should not happen.

We estimate now the values in the registers of $\cA$ after reading the word $w(s,j)$ for $j$ big enough.
By the construction of $w(s,j)$ we know that $\delta^*(q_0, w(s,j)) = (q_a,\sigma_{w(s,j)})$, and the final value in the registers of $\cA$ after reading $w(s,j)$ is given by composing substitutions corresponding to the composition of the words in~(\ref{eq:final_word}):
\[
\nu_0 \circ \sigma_{w(s,j)} = \nu_0 \circ \sigma_0 \circ \sigma_a^{s,j} \circ \sigma^{q_a,q_b} \circ \sigma_b^{j^2} \circ \sigma^{q_b,q_a} \circ \sigma_a^j.
\]
For all non-stable registers $y$ this expression is equivalent to $\sigma_a^j(y)$, which is estimated by $\Oo(1)$ by (\ref{equation:loop_a}) (where $\lambda$ is the identity substitution). 
For a stable register $x$, the story is much more complicated. We evaluate the expression $\sigma_{w(s,j)}(x)$ step-by-step to estimate its value. 
First, by (\ref{equation:loop_a}):
\[
\sigma_a^j(x) \;\; =\;\; \max\{ \ j \cdot c_a^x + \sigma^{q_a}(x) + \Oo(1), \ j \cdot d_a^x + \Oo(1) \ \}.
\]
Then by composing $\sigma_b^{j^2} \circ \sigma^{q_b, q_a}$ with  $\sigma_a^j(x)$  we get:
\begin{align*}
\sigma_b^{j^2} \circ \sigma^{q_b, q_a} \circ \sigma_a^j(x) & \;\; = \;\; \sigma_b^{j^2} \circ \sigma^{q_b, q_a} \circ \big( \max\{ \ j \cdot c_a^x + \sigma^{q_a}(x) + \Oo(1), \ j \cdot d_a^x + \Oo(1) \ \} \big) \\
& \;\; = \;\; \max\{ \ j \cdot c_a^x + \sigma_b^{j^2} \circ \sigma^{q_b, q_a} \circ \sigma^{q_a}(x) + \Oo(1), \ j \cdot d_a^x + \Oo(1) \ \}\\
& \;\; = \;\; \max\{ \ j \cdot c_a^x + \max\{ \ j^2 \cdot c_b^x + \sigma^{q_b}(x) + \Oo(1), \ j^2 \cdot d_b^x + \Oo(1) \ \} + \Oo(1), \ j \cdot d_a^x + \Oo(1) \ \}
\tag{$\heartsuit$}\label{eq:sigmab}\\
& \;\; = \;\; \max\big\{ \ j \cdot c_a^x + j^2 \cdot c_b^x + \sigma^{q_b}(x) + \Oo(1), \ j \cdot c_a^x + j^2 \cdot d_b^x + \Oo(1), \ j \cdot d_a^x + \Oo(1) \big\} 
\end{align*}
There is only one nontrivial equality~(\ref{eq:sigmab}). It holds by~(\ref{equation:loop_b}) and by considering $\lambda = \sigma^{q_b, q_a} \circ \sigma^{q_a}$.
The next step is to compose $\sigma_a^{s,j} \circ \sigma^{q_a,q_b}$ with $\sigma_b^{j^2} \circ \sigma^{q_b,q_a} \circ \sigma_a^j$.
We denote this composition by $\sigma_{w(s,j)}^{-w_0}$:
\setlength{\jot}{4pt}
\begin{align*}
\sigma_{w(s,j)}^{-w_0}(x) & \;\; = \;\; \sigma_a^{s,j} \circ \sigma^{q_a,q_b} \circ \big(\sigma_b^{j^2} \circ \sigma^{q_b,q_a} \circ \sigma_a^j(x) \big) \\¿
& \;\; = \;\; \sigma_a^{s,j} \circ \sigma^{q_a,q_b} \circ \max\big\{ \ j \cdot c_a^x + j^2 \cdot c_b^x + \sigma^{q_b}(x) + \Oo(1), \\
& \hspace{3.75cm} j \cdot c_a^x + j^2 \cdot d_b^x + \Oo(1), \ j \cdot d_a^x + \Oo(1) \big\} \\¿
& \;\; = \;\;  \max\big\{ \ j \cdot c_a^x + j^2 \cdot c_b^x + \sigma_a^{s,j} \circ \sigma^{q_a,q_b} \circ \sigma^{q_b}(x) + \Oo(1), \\
& \hspace{1.7cm} j \cdot c_a^x + j^2 \cdot d_b^x + \Oo(1), \ j \cdot d_a^x + \Oo(1) \big\} \\¿ 
& \;\; = \;\;  \max\big\{ \ j \cdot c_a^x + j^2 \cdot c_b^x + \max\big\{ \ s \cdot j \cdot c_a^x + \sigma^{q_a}(x) + \Oo(s), \   (s -1) \cdot j \cdot c_a^x + j \cdot d_a^x + \Oo(s) \ \big\} + \Oo(1), \tag{$\Diamond$}\label{eq:sigmaqaqb} \\
& \hspace{1.7cm} j \cdot c_a^x + j^2 \cdot d_b^x + \Oo(1), \ j \cdot d_a^x + \Oo(1) \big\} \\¿ 
& \;\; = \;\;  \max\big\{ \ j \cdot c_a^x + j^2 \cdot c_b^x + s \cdot j \cdot c_a^x + \sigma^{q_a}(x) + \Oo(s), \  j \cdot c_a^x + j^2 \cdot c_b^x +   (s -1) \cdot j \cdot c_a^x + j \cdot d_a^x + \Oo(s), \\
& \hspace{1.7cm} j \cdot c_a^x + j^2 \cdot d_b^x + \Oo(1), \ j \cdot d_a^x + \Oo(1) \big\} \\
& \;\; = \;\;  \max\big\{ \ (s+1) \cdot j \cdot c_a^x + j^2 \cdot c_b^x + \sigma^{q_a}(x) + \Oo(s), \ j \cdot (s \cdot c_a^x + d_a^x) + j^2 \cdot c_b^x  + \Oo(s), \\
& \hspace{1.7cm} j \cdot c_a^x + j^2 \cdot d_b^x + \Oo(1), \ j \cdot d_a^x + \Oo(1) \big\} 
\end{align*}
Like before all equalities are routine except for~(\ref{eq:sigmaqaqb}). It holds by~(\ref{equality1}) and by considering $\lambda = \sigma^{q_a, q_b} \circ \sigma^{q_b}$. We compose $\sigma_{w(s,j)}^{-w_0}$ with $\nu_0 \circ \sigma_0(x)$,
Given that $\nu_0 \circ \sigma_0(x)$ is a constant with respect to $j$ (i.e. $\nu_0 \circ \sigma_0(x) = \Oo(1)$), we have that $\nu_0 \circ \sigma_{w(s,j)}$ over a stable register $x$ in $\cA$ is equal to:
\begin{align}
\nu_0 \circ \sigma_{w(s,j)}(x)
& \; = \;  \max\big\{ \ (s+1) \cdot j \cdot c_a^x + j^2 \cdot c_b^x  + \Oo(s), \ j \cdot (s \cdot c_a^x + d_a^x) + j^2 \cdot c_b^x  + \Oo(s), \nonumber \\
& \hspace{1.7cm} j \cdot c_a^x + j^2 \cdot d_b^x + \Oo(1), \ j \cdot d_a^x + \Oo(1) \big\}  \label{equation:final}
\end{align}
For the rest of the proof, let us fix the value for $s$ that does not depend on $j$ (the exact value of $s$ will be defined at the very end of the proof). 
For a fixed $s$ we denote the quadratic functions on $j$ by:
\begin{align*}
g_1^x(j) & \;\; = \;\; c_b^x \cdot j^2  +  c_a^x \cdot (s+1) \cdot j + \Oo(1) \\ 
g_2^x(j) & \;\; = \;\; c_b^x \cdot j^2 + (c_a^x \cdot s + d_a^x) \cdot j + \Oo(1) \\
g_3^x(j) & \;\; = \;\; d_b^x \cdot j^2 + c_a^x \cdot j +  \Oo(1) \\
g_4^x(j) & \;\; = \;\; d_a^x \cdot j + \Oo(1)
\end{align*}
where $\Oo(1)$ denotes some constant that does not depend on $j$ (recall that $s$ is fixed).
Then for every stable register $x$ we have
$$
\nu_0 \circ \sigma_{w(s,j)}(x) = \max_{i=1}^4\{g_i^x(j)\}.
$$
Notice that if $d_a^x < c_a^x$ then for $j$ big enough it holds that $g_2^x(j) \leq g_1^x(j)$ and $g_2^x(j)$  is dominated by $g_1^x(j)$.
This motivates the definition of $e_a^x$, given by
$e_a^x = d_a^x - c_a^x$ when $d_a^x \geq c_a^x$ and $e_a^x = 0$ otherwise.
Furthermore, redefine $g_2^x(j)$ by:
\[
g_2^x(j) \;\; = \;\;  c_b^x \cdot j^2  + ( c_a^x \cdot (s+1) + e_a^x) \cdot j + \Oo(1) 
\]
It is easy to see that $\nu_0 \circ \sigma_{w(s,j)}(x) = \max_{i=1}^4\{g_i^x(j)\}$ holds for the new definition of $g_2^x(j)$. 

Towards the end of the proof, we study the output function $\asem{\cA}(w(s,j)) = \nu_0 \circ \sigma_{w(s,j)} \circ \mu(q_a)$. 
By Lemma~\ref{lemma:copyless-expressions}, we know that the copyless expression $\mu(q_a)$ can be presented as an expression of the form:
\begin{align*}
\mu(q_a) & \;\; = \;\; \max_{i = 1}^{k} \, \left\{ \ l_i + \sum_{x \in X_i} x \ \right\},
\end{align*}
where $X_1, \ldots, X_k$ is a sequence of different sets over $\cX$ and $l_1, \ldots, l_k$ is a sequence of values over $\nat$ for $k \ge 0$. Notice that in this representation the expression $\mu(q_a)$ might be not copyless.
By the previous analysis we know that for $j$ big enough $\nu_0 \circ \sigma_{w(s,j)}(x)$ is either estimated by $\Oo(1)$ or equal to $\max_{i=1}^4\{g_i^x(j)\}$.
Then by composing $\nu_0 \circ \sigma_{w(s,j)}$ with $\mu(q_a)$ we get:
\begin{align}
\asem{\cA}(w(s,j))& \;\; = \;\;  \nu_0 \circ \sigma_{w(s,j)} \circ \mu(q_a) \nonumber \\
& \;\; = \;\;  \nu_0 \circ \sigma_{w(s,j)} \circ \max_{i = 1}^{k} \, \big\{ l_i +  \sum_{x \in X_i} x \ \big\} \nonumber \\
& \;\; = \;\; \max_{i = 1}^{k} \, \big\{ \ l_i + \sum_{x \in X_i} \nu_0 \circ \sigma_{w(s,j)}(x) \ \big\} \nonumber  \\
& \;\; = \;\; \max_{i = 1}^{h} \, \big\{ \ n_i^b + \sum_{x \in Y_i} \max_{i=1}^4\{ \ g_i^x(j) \ \} \ \big\} \label{eq:output-in-terms-of-max-plus} 
\end{align}
where $Y_1, \ldots, Y_h$ is a new sequence of subsets of stable registers and $m_1, \ldots, m_h$ is a sequence of values over $\nat$ for $k \ge 0$.
The last equality holds since $\nu_0 \circ \sigma_{w(s,j)}(x)$ are constants for non-stable registers. Thus we can sum the constants with $l_i$ and get new constants $n_i^b$; and regroup the sets $X_i$ to form new sets of stable registers $Y_i$. Notice that we might need multiple copies of the same $Y_i$.

The next step is to further simplify the output $\asem{\cA}(w(s,j))$.
For this, note that $\asem{\cA}(w(s,j))$ is the sum and maximization of quadratic functions over $j$. 
Then for $j$ big enough there exists a constant $i^* \leq h$ and a partition of $Y_{i^*} = Z_1 \uplus Z_2 \uplus Z_3 \uplus Z_4$ such that:
\begin{align*}
\asem{\cA}(w(s,j)) & \;\; = m_{i^*} + \sum_{x \in Z_1} g_1^x(j) + \sum_{x \in Z_2} g_2^x(j) + \sum_{x \in Z_3} g_3^x(j) + \sum_{x \in Z_4} g_4^x(j).
\end{align*}
The index $i^*$ is the one where (\ref{eq:output-in-terms-of-max-plus}) is maximized for $j$ big enough and the partition $Y_{i^*} = Z_1 \uplus Z_2 \uplus Z_3 \uplus Z_4$ is the division of $Y_{i^*}$ where each function $g_i^x$ dominates for $j$ big enough.
Thus, $\asem{\cA}(w(s,j))$ is a sum of quadratic functions and by summing common $j$-terms we can reduce $\asem{\cA}(w(s,j))$ to a polynomial of the form $B \cdot j^2 + A \cdot j + C$. 
Intuitively, the value $B \cdot j^2$ should correspond to the number of $b$'s in $w(s,j)$ and $A \cdot j$ to the number of $a$'s in $w(s,j)$. 
In the last part of this proof, we analyze the $A$-coefficient and compare it with the corresponding coefficient in $f_\cB^R(w(s,j))$. 

Recall from (\ref{eq:real-output-function}) that the output of $f_\cB^R$ over $w(s,j)$ is equal to $n_b \cdot j^2 + n_a \cdot s \cdot j  + \Oo(s)$.
If the output $\asem{\cA}(w(s,j))$ is correct then we should have $A = n_a \cdot s$.
By adding the linear coefficients of $g_i^x(j)$ for $i \in\{1,2,3,4\}$ and $x \in Z_i$ we get that
\setlength{\jot}{7pt}
\begin{align*}
A & \;\; = \;\; \sum_{x \in Z_1} c_a^x \cdot (s+1) \  + \ \sum_{x \in Z_2} \big( c_a^x \cdot (s+1) \  + \ e_a^x \big) \  + \ \sum_{x \in Z_3} c_a^x \  + \ \sum_{x \in Z_4} d_a^x \\
& \;\; = \;\;  (s+1) \cdot \sum_{x \in Z_1 \cup Z_2} c_a^x  \  + \ \sum_{x \in Z_2} e_a^x  \  + \ \sum_{x \in Z_3} c_a^x \  + \ \sum_{x \in Z_4} d_a^x.
\end{align*}
Given that $A$ should be equal to $n_a \cdot s$ this implies that $\sum_{x \in Z_1 \cup Z_2} c_a^x  < n_a$.
Otherwise, $A > n_a \cdot s$. 
Therefore, $\sum_{x \in Z_1 \cup Z_2} c_a^x \leq n_a - 1$ and thus
\[
(s+1) \cdot \sum_{x \in Z_1 \cup Z_2} c_a^x \;\; \leq \;\; (s+1) \cdot (n_a - 1)
\]
By replacing this fact in the definition of $A$, we can over-approximate the $A$-coefficient as follows
\begin{align}
A & \;\; \leq \;\; (s+1) \cdot (n_a - 1) \  + \ \sum_{x \in Z_2} e_a^x  \  + \ \sum_{x \in Z_3} c_a^x \  + \ \sum_{x \in Z_4} d_a^x \nonumber \\
&  \;\; \leq \;\; s\cdot n_a \ - \  s \ + \ n_a \ - \ 1 \ + \ \sum_{x \in \cX} (e_a^x + c_a^x + d_a^x). \label{eq:A-bound} 
\end{align}
For the last bound, recall that $Z_i$ are pairwise disjoint subsets.
It is important to notice that the inequality~(\ref{eq:A-bound}) does not depend on how we choose $s$.

We are ready to show that, by suitably choosing a fix value for $s$, $A$ will not be equal to $n_a \cdot s$.
It suffices to choose $s > n_a - 1 + \sum_{x \in \cX} (e_a^x + c_a^x + d_a^x)$. This with (\ref{eq:A-bound}) implies that $A < n_a \cdot s$ which is a contradiction.
In other words, $\cA$ is not counting all the $a$'s in $w(s,j)$ for $j$ big enough.
\end{proof}

From Theorem~\ref{theorem:counterexample} we get that copyless CRA over max-plus semiring is strictly less expressive than WA: we know that copyless CRA are contained in WA~\cite{kickasspaper} and by Figure~\ref{fig:automaton_reverse} the function $f_{\cB}^R$ is definable by a linear CRA and thus by a WA. It is also well-known~\cite{Sakarovitch09,DrosteHWA09} that the class of functions defined by WA is closed under reverse. 

\begin{corollary} 
	\label{corollary:weighted-maxplus}
	The class of functions defined by copyless CRA over the max-plus semiring is strictly contained in the class of functions defined by WA over the max-plus semiring.
\end{corollary}

Moreover, from Theorem~\ref{theorem:counterexample} we immediately get the following corollary.

\begin{corollary}
	\label{corollary:reverse_closure}
	There exists a semiring $\SR$ such that the class of functions recognizable by copyless CRA over $\SR$ is not closed under reverse.
\end{corollary}


\section{Bounded alternation copyless CRA}
\label{sec:bounded_alternation}
Given that copyless CRA are not closed under reverse operation we look for a robust subclass of copyless CRA. 
The proof of Theorem~\ref{theorem:counterexample} suggests that the alternation between semiring operations is the reason why copyless CRA cannot replicate the behavior of the CRA $\cB$ in backward mode:
$\cB$ can sum the number of $a$- and $b$-symbols to maximize each time that a $\#$-symbol is read, but it cannot do the same alternation of operations when the word is read in the other direction.
This fact inspires the definition of \emph{bounded alternation} copyless CRA, a strict subclass of copyless CRA where the output is restricted to expressions with bounded alternation.
This class was proposed in~\cite{kickasspaper} and characterized in terms of the so-called \emph{Maximal Partition logic}. 
In this section, we show that bounded alternation copyless CRA has also good closure properties; this class is closed under unambiguous non-determinism, regular look-ahead and, moreover, under reverse.

The \emph{alternation} of $e \in \expr(\cX)$ is defined as the maximum number of switches between $\add$ and $\mult$ operations over all branches of the parse-tree of $e$. 
Formally, let $\otimes \in \{\add, \mult\}$ and  $\bar{\otimes}$ be the dual operation of $\otimes$ in $\SR$.
We define the set of expressions $\expr_0^\otimes(\cX)$ with $0$-alternation by $\expr_0^\otimes = \cX \cup \SR$.
For any $N \geq 1$, we define the set of expressions $\expr_N^\otimes(\cX)$ as the $\otimes$-closure of $\expr_{N-1}^{\bar \otimes}(\cX)$, namely, $\expr_N^\otimes(\cX)$ is the minimal set of expressions that contains $\expr_{N-1}^{\bar \otimes}(\cX)$ and satisfies that $e_1 \otimes e_2 \in \expr_N^\otimes(\cX)$ whenever $e_1, e_2 \in \expr_N^\otimes(\cX)$.
We define $\expr_N(\cX) = \expr_N^\add(\cX) \cup \expr_N^\mult(\cX)$.

We say that a copyless CRA $\cA$ has \emph{bounded alternation} if there exists~$N$ such that $\gsem{\cA}(w) \in \expr_N(\cX)$ for every $w \in \Sigma^*$.
A copyless CRA $\cA$ is called a bounded alternation copyless CRA (BAC) if $\cA$ has bounded alternation.
All the examples of copyless CRA presented in Section~\ref{sec:preliminaries} have bounded alternation. For example, one can easily check that the alternation of the copyless CRA in Example~\ref{ex:max-b-substrings} is bounded by 2.

The alternation of any expression can be easily derived just by counting what is the maximum number of alternation between $\add$ and $\mult$. 
However, it is not directly clear how to check if a copyless CRA has bounded alternation from its definition. 
We show that this semantical property can be verified in \NLogSpace{} in the size of a copyless CRA. 
\begin{proposition} \label{prop:verifying-bounded-alternation}
The problem of deciding whether a copyless CRA has bounded alternation can be computed in \NLogSpace. Furthermore, if a copyless CRA has bounded alternation, the alternation is bounded by $|Q|\cdot \max\{\, \operatorname{alt}(\sigma) \, \mid \, \exists. p, q \in Q. \ \delta(q, a) = (p, \sigma) \,\}$ where $\operatorname{alt}(\sigma)$ is the alternation of $\sigma$.
\end{proposition}
\begin{proof}
	\newcommand{\GcA}{\cG_{\cA}}
	\newcommand{\VcA}{V_{\cA}}
	\newcommand{\EcA}{E_{\cA}}
	
	Let $\cA \; = \; (Q, \AL, \cX, \Delta, q_0, \nu_0, \mu)$ be a copyless CRA. 
	We define the graph $\GcA = (\VcA, \EcA)$ in which we look for cycles that produce unbounded alternation in $\cA$.
	The set of vertices $\VcA$ of $\GcA$ are triples in $Q \times \cX \times \{\add, \mult\}$. 
	Each vertex $(q, x, \otimes) \in \VcA$ keeps track of the current state $q$, a register $x$, and the last operation~$\otimes$ seen in the last transition.
	We define $\EcA \subseteq (Q \times \cX \times \{\add, \mult\})^2$ such that $((q_1, x_1, \otimes_1), (q_2, x_2, \otimes_2)) \in \EcA$ if, and only if:
	(1) $\delta(q_1, a) = (q_2, \sigma)$ for some $a \in \Sigma$, (2) $x_1 \in \var(\sigma(x_2))$, and (3) $\otimes_2$ is equal to $\otimes$ whenever $\sigma(x_2) = e_1 \otimes e_2$ for some expressions $e_1$ and $e_2$, or equal to $\otimes_1$ otherwise. 
	Intuitively, $\otimes_2$ keeps track of the last operation seen when $x_1$ passes its value to $x_2$ in the expression $\sigma(x_2)$.
	
	Furthermore, we define $\EcA^* \subseteq \EcA$ that indicates edges that produce an alternation of operations. 
	Precisely, $((q_1, x_1, \otimes_1), (q_2, x_2, \otimes_2)) \in \EcA^*$ if, and only if, $((q_1, x_1, \otimes_1), (q_2, x_2, \otimes_2)) \in \EcA$ and there exists a subexpression $e_1 \otimes e_2$ in $\sigma(x_2)$ with $x_1 \in \var(e_1 \otimes e_2)$ and $\otimes \neq \otimes_1$.
	In other words, there exists an alternation with respect to the last operation $\otimes_1$ in the transition from $q_1$ to $q_2$. 
	
	Let $\VcA^*$ be the set of all $(q, x, \otimes) \in \VcA$ such that $x \in \var(\mu(q))$ and let $\GcA' = (\VcA', \EcA')$ be the subgraph of $\GcA$ induced by vertices that can reach $\VcA^*$ in $\GcA$. 
	It is straightforward to prove that if $\GcA'$ has a cycle with an $\EcA^*$-edge then $\cA$ has unbounded alternation.
	Conversely, note that if $\GcA'$ do not have a cycle with an $\EcA^*$-edge, then any path can cross at most $|Q|$ $\EcA^*$-edges, each with at most $\max\{\, \operatorname{alt}(\sigma) \, \mid \, \exists. p, q \in Q. \ \delta(q, a) = (p, \sigma) \,\}$ alternations.
	These properties can be checked in $\NLogSpace$ since $\GcA$ can be generated on the fly in logarithmic space.	
\end{proof}


\subsection{Closure under unambiguous non-determinism}
We first extend the model of bounded alternation copyless CRA  with non-determinism.
The class CRA was designed as a deterministic model in contrast to weighted automata, where non-determinism plays a crucial role. Thus, we restrict non-determinism to be unambiguous, namely, we allow many runs over a word but at most one accepting run, which defines the output.
Formally, a non-deterministic CRA is a tuple $\cA \; = \; (Q, \AL, \cX, \Delta, I_0, V_0, F, \mu)$ where $Q$, $\AL$, $\cX$, and $\mu$ are defined as before, $\Delta \subseteq Q \times \Sigma \times Q \times \subs(\cX)$ is a finite transition relation, $I_0 \subseteq Q$ is a set of initial states, $V_0: I_0 \rightarrow \val(\cX)$ assigns an initial valuation for each initial state, and $F$ is the set of final states. 
Additionally, we assume that for every $q, q' \in Q$ and $a \in \Sigma$ there exists at most one $\sigma \in \subs(\cX)$ such that $(q, a, q', \sigma) \in \Delta$.
Given a string $w = a_1 \ldots a_n \in \AL^*$, a run of $\cA$ over $w$ is a sequence of configurations:
$
(q_0, \nu_0) \:\trans{a_1} \: (q_1, \nu_1) \: \trans{a_2} \: \ldots \: \trans{a_n} \: (q_n, \nu_n) 
$
such that $q_0 \in I_0$, $\nu_0 = V_0(q_0)$, and for $1 \leq i \leq n$,  $(q_{i-1}, a_i, q_i, \sigma_i) \in \Delta$ and $\nu_i(x) = \asem{\nu_{i-1} \comp \sigma_i(x)}$ for each $x \in \cX$.
Furthermore, a run of $\cA$ over $w$ is an accepting run if $q_n \in F$. 
We say that $\cA$ is unambiguous if for every $w \in \Sigma^*$ there exists exactly one accepting run of $\cA$ over $w$. 
The output of $\cA$ over $w$ is defined as $\asem{\cA}(w) = \asem{\nu_n \circ \mu(q_n)}$ where $(q_n, \nu_n)$ is the final configuration of the only accepting run of $\cA$ over $w$.
The definitions of unambiguous copyless CRA and unambiguous BAC are straightforward restrictions of this definition.

\begin{example}\label{ex:unambiguous}
Fix $\Sigma = \{a,b\}$. For every $w \in \Sigma^*$ let $\#_a(w)$ denote the number of $a$'s in $w$ and $\#_b^{last}(w)$ denote the number of $b$'s in the last block (e.g. $\#_b^{last}(ab^2ab^4ab^3aa) = 3$). Consider the function $f : \Sigma^* \to \nat$ defined as $f(w) = \max(\#_a(w), \#_b^{last}(w))$. This function can be easily defined by a nondeterministic unambiguous BAC in Figure~\ref{fig:unambiguous}. The automaton keeps the number of $a$'s in register $x$ and the number of $b$'s in the last block in $y$. To do that the automaton guesses that the word reached the last block of $b$'s (i.e. state $q_2$) and then it continues only if the remaining letters are $a$ (i.e. state $q_3$). For simplicity, assignments keeping the value ($x:=x$) were omitted in the picture. All transitions are labeled with the alphabet letter and the register assignments. If both assignments keep the value of the register this is denoted by ``$-$''. 
\end{example}

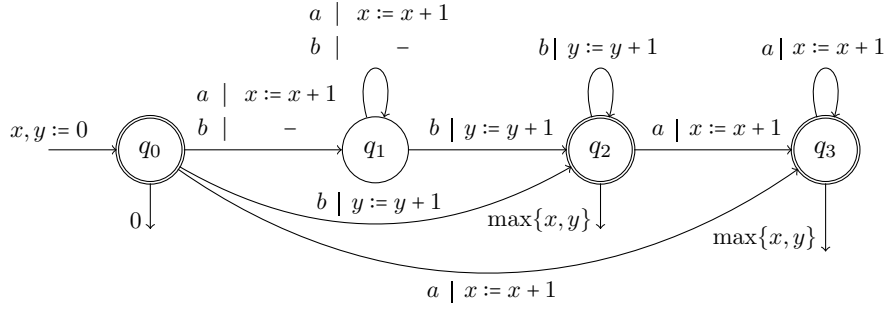
\begin{figure}
\centering
 \begin{tikzpicture}
\node[state, draw=white] 		(p0) at (0,0) {};
\node[state, right of=p0, node distance=1.8cm] 		(p) {$q_1$};
\node[state,  below of=p, node distance=1.5cm, draw=white] (pf) {};

\node[state,double, draw=white, right = 2.1cm of p0] 		(1p0){};
\node[state,double, right of=1p0, node distance=1.8cm] 		(1p) {$q_2$};
\node[state,double,  below of=1p, node distance=1.5cm, draw=white] (1pf) {};

\node[state,double, draw=white, right = 2.1cm of 1p0] 		(2p0){};
\node[state,double, right of=2p0, node distance=1.8cm] 		(2p) {$q_3$};
\node[state,double,  below of=2p, node distance=1.8cm, draw=white] (2pf) {};

\draw (1p) edge[pos=0.8,->] node[left,scale=0.9] {$\; \max\{x,y\}$} (1pf);
\draw (2p) edge[pos=0.8,->] node[left,scale=0.9] {$\; \max\{x,y\}$} (2pf);

\path (p)   edge	[loop above,looseness=10,->]	node[scale=0.9] {
	\renewcommand{\arraystretch}{0.7}
	$ 
	\begin{array}{ccc}
	a & \!\!\! \mid \!\!\! & x := x+1 \vspace{1.5mm} \\
	b & \!\!\! \mid \!\!\! & -
	\end{array}$
}  (p)
(p)   edge	[->]	node[above,scale=0.9] {
	\renewcommand{\arraystretch}{0.7}
	$
	\begin{array}{c|c}
	b & y := y+1
	\end{array}$}  (1p)

	(1p)   edge	[loop above,looseness=10,->]	node[scale=0.9] {
	\renewcommand{\arraystretch}{0.7}
	$  b \;  \begin{array}{|c}
	 y := y+1
	\end{array}$
}  (1p)

(1p)   edge	[->]	node[above,scale=0.9] {
	\renewcommand{\arraystretch}{0.7}
	$
	\begin{array}{c|c}
	a & x := x+1
	\end{array}$}  (2p)
	(2p)   edge	[loop above,looseness=10,->]	node[scale=0.9] {
	\renewcommand{\arraystretch}{0.7}
	$  a \;   \begin{array}{|c}
	 x := x+1
	\end{array}$
}  (2p);

\node[state,double, draw=white, left = 2.1cm of p0] 		(3p0){};
\node[state,double, right of=3p0, node distance=1.8cm] 		(3p) {$q_0$};
\node[state,double,  below of=3p, node distance=1.5cm, draw=white] (3pf) {};

\draw (3p0) edge[pos=0.01,->] node[above,scale=0.9] {$x,y:=0$} (3p);
\draw (3p) edge[pos=0.8,->] node[left,scale=0.9] {$\; 0$} (3pf);

\path
(3p)   edge	[->]	node[above,scale=0.9] {
	\renewcommand{\arraystretch}{0.7}
	$
	\begin{array}{ccc}
	a & \!\!\! \mid \!\!\! & x := x+1 \vspace{1.5mm} \\
	b & \!\!\! \mid \!\!\! & -
	\end{array}$}  (p)
	
(3p)   edge[bend right,->]	node[above,scale=0.9] {
	\renewcommand{\arraystretch}{0.7}
	$
	\begin{array}{c|c}
	b & y := y+1
	\end{array}$}  (1p)
(3p)
edge [->,bend right=35] node[below,scale=0.9] {
\renewcommand{\arraystretch}{0.7}
	$ 
	\begin{array}{c|c}
	a &  x := x+1
	\end{array}$
}
(2p);

\end{tikzpicture}
\caption{Unambiguous BAC recognizing $f$ from Example~\ref{ex:unambiguous}.}
\label{fig:unambiguous}
\end{figure}

We do not know whether for each unambiguous copyless CRA there is an equivalent deterministic copyless CRA. 
However, this is true when we assume bounded alternation. In particular, one can define $f$ from Example~\ref{ex:unambiguous} without unambiguous nondeterminism.

\begin{theorem} \label{theo:unambiguous-theorem}
	Let $\cA = (Q, \AL, \cX, \delta, I_0, V_0, F, \mu)$ be an unambiguous BAC whose alternation is bounded by $N$.
	There exists a deterministic BAC $\cA'$ that computes the same function as $\cA$, that is, $\asem{\cA}(w) = \asem{\cA'}(w)$ for every $w \in \Sigma^*$.
	Furthermore, the number of states of $\cA'$ is of size $2^{\Oo(|Q|^3 \cdot |\cX|^5 \cdot N^2)}$ and the number of registers in $\cA'$ is of size $\Oo(|Q| \cdot |\cX|^2 \cdot N)$.
\end{theorem}

\begin{proof}[Proof of Theorem~\ref{theo:unambiguous-theorem}]
To start we need to introduce some notation for trees, expressions, and substitutions that will be the main objects during the proof. 

\noindent \textbf{Trees.} Let $\Sigma$ be a set of labels. An (unordered) labeled $\Sigma$-tree $t$ is a finite function $t: \nodes(t) \rightarrow \Sigma$ such that $\nodes(t)$ is a finite prefix-closed subset of $\nat^*$ (i.e. $w \in \nodes(t)$ whenever $w\cdot i \in \nodes(t)$ for some $i\in \nat$). 
We say that $\epsilon$ is the root of $t$ and $w \cdot i \in \nodes(t)$ is a child of $w$.
For any $w \in \nodes(t)$, we denote by $t[w]$ the subtree rooted at $w$, i.e. $t[w](i) = t(w \cdot i)$ for every $i \in \nodes(t[w])$.
For every $a \in \Sigma$ we write $a\{t_1, \ldots, t_k\}$ to denote a tree whose root is labeled by $a$ and $t_1, \ldots, t_k$ are the subtrees hanging from the root.
We say that $w \in \nodes(t)$ is an internal node of $t$ if $w \cdot i \in \nodes(t)$ for some $i \in \nat$.
Otherwise, $w$ is called a leaf of $t$ and the set of all leaves of $t$ is denoted by $\tleaves(t)$.
We say that a tree is \emph{complete} if every internal nodes has at least two children. 
One can easily check that if $t$ is a complete tree, then $|\nodes(t)| \leq 2 \cdot |\tleaves(t)|$. 
Finally, we denote by $\trees(\Sigma)$  the set of all $\Sigma$-trees.

\noindent \textbf{Expressions and substitutions.} From now on, we see expressions as unordered labeled trees by exploiting the commutativity and associativity of the semiring. 
Formally, for $\otimes \in \{\add, \mult\}$ we define $\expr^{\otimes}(\cX)$ as the minimal set such that $S \cup \cX \subseteq \expr^{\otimes}(\cX)$ and $\otimes\{e_1, \ldots, e_k\} \in \expr^{\otimes}(\cX)$ for every  $e_1, \ldots, e_k \in \expr^{\bar{\otimes}}(\cX)$ with $k \geq 1$ (recall that $\bar{\otimes}$ is the dual operation of $\otimes$ in $\SR$).
Indeed, any expression can be represented by a unique unordered tree in $\expr^{\add}(\cX)$ or $\expr^{\add}(\cX)$. 
For example, the expression $((x \mult (y \mult 2)) \add 3) \add (z \mult 4)$ can be represented by:
\begin{equation*} \label{ex:tree-formula}
\add\big\{ \mult\{x,y,2\},\, 3,\, \mult\{z,4\}\big\}
\end{equation*}
Intuitively, the unordered tree encodes the expression by removing the parenthesis and the order of multiplication and addition. 
For the sake of simplification, in the sequel we represent every expression with its canonical representation in $\expr^{\add}(\cX) \cup \expr^{\mult}(\cX)$ and we associate $\expr(\cX)$ with $\expr^{\add}(\cX) \cup \expr^{\mult}(\cX)$.

For two disjoint set of variables $\cX_1$ and $\cX_2$, we define $\subs(\cX_1, \cX_2)$ to be the set of all copyless substitutions $\sigma: \cX_1 \rightarrow \expr(\cX_1 \cup \cX_2)$, that is, copyless substitutions where the domain contains only variables in $\cX_1$.
Here, composition between substitutions in $\subs(\cX_1, \cX_2)$ is defined in a straightforward way where $\cX_2$-variables are treated as constants.
Notice that the composition of two substitutions is copyless for the registers $\cX_1$ but not necessarily for the registers $\cX_2$. This is because the registers $\cX_2$ are not in the domain of these substitutions.

\noindent \textbf{Substitution trees.} We denote by $\trees(\cX_1, \cX_2)$ the set of all trees labeled by substitutions in $\subs(\cX_1, \cX_2)$.
Furthermore, we say that $t \in \trees(\cX_1, \cX_2)$ is \emph{copyless} if $t(u)$ is copyless for every $u \in \nodes(t)$ and:
\begin{eqnarray} \label{eq:substitution-trees}
\var(t(u)) \cap \var(t(v)) & \subseteq & \cX_1
\end{eqnarray}
for every $u, v \in \nodes(t)$.
In other words, each $\cX_2$-variable is used at most once in a substitution tree $t$ (note that there is no restriction in $\cX_1$). With the condition (\ref{eq:substitution-trees}), the composition of substitutions between different nodes is also copyless for the set $\cX_2$.
We say that $t \in \trees(\cX_1, \cX_2)$ is \emph{constant-free} if, for every $u \in \nodes(t)$ the substitution $t(u)$ does not use elements from $\cS$.
Finally, for any node $u \in \nodes(t)$ we define the \emph{collapse operation}  $t^\collapse(u)$ such that:
\[
t^\collapse(u) \; = \;  t(\epsilon) \circ t(u[\cdot, 1]) \circ \ldots \circ  t(u[\cdot, k])
\]
where $k = |u|$ and $u[\cdot, i]$ is the prefix of $u$ until position $i$. In other words, $t^\collapse(u)$ is the composition of all substitution along the branch from the root until $u$. Note that, by condition (\ref{eq:substitution-trees}), this composition always produces a copyless substitution.

Before going into the details of the proof we recall the following lemma which is a well-known property of unambiguous finite automata and, in particular, of unambiguous CRA. 
\begin{lemma}\cite{Sakarovitch09,weber1991degree} \label{lemma-unambiguous}
	For every different runs $\rho$ and $\rho'$ of an unambiguous finite automaton $\cA$ over $w \in \Sigma^*$, the last states of $\rho$ and $\rho'$ are different, that is, $\rho(|w|) \neq \rho'(|w|)$.
\end{lemma}

Let $\cA = (Q, \AL, \cX, \delta, I_0, V_0, F, \mu)$ be an unambiguous BAC whose alternation is bounded by $N$.
To construct a deterministic BAC $\cA'$ from $\cA$, the idea is to execute all runs of $\cA$ over a word in parallel, simulating some sort of subset construction~\cite{HopcroftU79}.
The problem here is that we cannot make arbitrary copies of registers (recall that $\cA'$ must be copyless) and then we cannot simulate directly with a deterministic automaton the arbitrary branching of runs, namely, non-deterministic transitions. 
To solve this, we keep in states a tree of runs that encodes how runs are branching when the word is read. 
Of course, if we keep the branching of all runs in memory, the tree would be unbounded. 
A characteristic property of unambiguous CRA (Lemma~\ref{lemma-unambiguous}) is that, for each prefix, there is at most $|Q|$ partial runs and, moreover, all leading states are different. 
These two facts suggest a tree structure of the runs where each branch is a run and where the number of branches is bounded by~$|Q|$. 
Although we can keep in finite memory the branching structure of the partial runs of $\cA$, we cannot do the same trick for registers and naively keep copies of registers for each run (recall again that $\cA'$ must be copyless).
To overcome this problem, $\cA'$ will postpone the evaluation of registers by keeping substitutions inside the internal tree structure of runs (i.e. substitutions trees).
Clearly, $\cA'$ cannot postpone these substitutions forever and store a long sequence of these objects with finite memory. 
The key idea here is to prune and reduce the tree structure by doing partial evaluation of the substitutions whenever is possible. 
We show that by exploiting the bounded alternation of the output and the copyless restriction of $\cA$, we need only a finite amount of memory to remember the tree structure and the substitutions of all runs.  
Finally, for the sake of simplification we present the main construction for the restricted case when $\cA$ has only one initial state, namely, $\cA = (Q, \AL, \cX, \delta, q_0, \nu_0, F, \mu)$ where $q_0$ is the only initial state and the initialization valuation is just $\nu_0: \cX \rightarrow \SR$.

Let $\cX$ be the set of registers in $\cA$ and $\hX = \{\hat{x} \mid x \in \cX\}$ a disjoint copy of $\cX$.
We construct a deterministic BAC $\cA' = (Q', \AL, \cY, \delta', q_0', \nu_0', \mu')$ as follows.
\begin{itemize} \itemsep2mm
	
	\item $Q'$ is the set of all pairs $(t, B)$ where $t \in \trees(\cX, \cY)$ is a complete, copyless, constant-free substitution tree and $B: \tleaves(t) \rightarrow Q$ is an injective function.
	
	\item $\cY$ is a set of registers of size $|\cY| = 2 \cdot |Q| \cdot |\cX|\cdot(|\cX| \cdot N +1)$ satisfying $\hX \subseteq \cY$ and $\cX \cap \cY = \emptyset$. 
	
	\item $q_0' = (t_0, B_0)$ is the initial state of $\cA'$ where $t_0$ is a single-node tree labeled with $\sigma_0 \in \subs(\hX, \cY)$ and $B_0(\epsilon) = q_0$ such that $\sigma_0(x) = \hat{x}$ for every $x \in \cX$. 
	
	\item $\nu_0': \cY \rightarrow \SR$ is the initial substitution such that $\nu_0'(\hat{x}) = \nu_0(x)$ for all $x \in \cX$ and $\nu_0'(y) = \zero$ for all $y \in \cY \ \backslash \ \hX$.
	
	\item $\mu'$ is the final substitution such that $\mu'((t,B)) = t^\collapse(u) \circ \mu(q)$ for every $(t,B) \in Q'$ whenever $u \in \tleaves(t)$ is the only leaf satisfying $B(u) \in F$ and $q = B(u)$. 
	
	\item $\delta'$ is the transition function defined below.
\end{itemize}

	We start explaining the relation between the set of registers $\hX$, $\cX$ and~$\cY$. 
	The $\cX$-variables in the construction are used inside the internal structure of $\cA'$ (i.e. states). 
	In fact, $\cX$-variables will never be used as real variables by $\cA'$.
	They will just be used to keep track of temporary substitutions.
	Regarding $\cY$, the decision of taking the size of $\cY$ equal to $2 \cdot |Q| \cdot |\cX|\cdot(|\cX| \cdot N +1)$ is technical and will be clear later in the proof.
	Regarding $\hX$, the property $\hX \subseteq \cY$ is needed for the definition of the initial substitution. Besides that we will not use the fact that $\hX \subseteq \cY$.
	
	Recall that each state in $Q'$ is composed by a complete, copyless, and constants free substitution tree $t$ and an injective function $B: \tleaves(t) \rightarrow Q$. 
	As was suggested before, $t$ keeps track of the branching history of the partial runs of $\cA$ and $B$ labels leaves of $t$ with states in $Q$.
	The plan here is that each partial run of $\cA$ is represented by a different branch of $t$ and the states assigned by $B$ represent the current state of each~run. 
	
	Next, we show that $Q'$ is a finite set. 
	Indeed, we can bound the size of any state $(t, B) \in Q'$ in terms of $|Q|$, $|\cX|$, and $|\cY|$.
	First, the number of nodes in $t$ is bounded by $2\cdot |Q|$, that is, $|\nodes(t)| \leq 2\cdot|\tleaves(t)| \leq 2 \cdot |Q|$ since $t$ is a complete tree and $B$ is an injective function over $\tleaves(t)$.
	Second, each copyless and constants-free expression that one can write with $|\cX \cup \cY|$ variables is of size at most $2\cdot (|\cX| + |\cY|)$ (i.e. the expression has no constant, and each variable can be used at most once). Third, the size of any copyless and constants-free substitution $\sigma \in \subs(\cX, \cY)$ is bounded by $|\sigma| = \Oo(2 \cdot |\cX| \cdot (|\cX| + |\cY|))$ and thus the number of possible labels for $t$ is finite, which implies that $|t|$ is of bounded size.
	Fourth, $B$ is an injective function from $\tleaves(t)$ into~$Q$ and, given that $t$ is of bounded size, then $B$ is also bounded.
	Finally, given that the size of $t$ and $B$ is bounded by $|Q|$, $|\cX|$, and $|\cY|$, then we can conclude that $Q'$ is a finite set. 
	Furthermore, by composing and counting the previous arguments, one can show that $|Q'| = 2^{\Oo(|Q|\cdot |\cX| \cdot |Y|^2)} = 2^{\Oo(|Q|^3 \cdot |\cX|^5 \cdot N^2)}$.
	
%
	
	For a given word $w$, consider the set of partial runs in $\cA$. By Lemma~\ref{lemma-unambiguous}  every partial run on $w$ is determined by its last state. Then, we define $Q(w) \subseteq Q$ as the set of states in the partial runs of $\cA$ over $w$. To understand the main purpose of $(t, B)$, we state the following Lemma. Its proof is postponed to the end of this section.
	
	\begin{lemma} \label{lemma: branch-collapsing}
		Let $\rho$ be a run of $\cA$ over $w$ and $(q, \nu)$ be the last configuration of $\rho$.
		Then there exists a run of $\cA'$ over $w$ reaching configuration $((t, B), \xi)$, where $\xi$ is a valuation over $\cY$.
		Moreover $Q(w) = range(B)$ and for every $q \in Q(w)$
		there exists $u \in \tleaves(t)$ such that $B(u) = q$ and
		$
		\nu \; = \;  \xi \circ t^\collapse(u).
		$
	\end{lemma}
	In other words, from a state $(t, B)$ we can recover the configuration $(q, \nu)$ from a leaf $u$, by applying $B$ over $u$ and composing the substitutions from the root to $u$.
	
	A key notion for defining the transition function $\delta'$ is the notion of an $\cX$-reduction.
	An $\cX$-reduction of an expression $e \in \expr(\cX \cup \cY)$ is a tuple $(r, \sigma)$ where $r \in \expr(\cX \cup \cY)$ is a copyless expression without constants and $\sigma: \cY \rightharpoonup \expr(\cY)$ is a partial substitution (i.e. partial function) such that $e = \hat{\sigma}(r)$. Notice that $\dom(\sigma) \cap \cX = \emptyset$ since we assumed $\cX \cap \cY = \emptyset$. 
	The goal of an $\cX$-reduction is to factorize constants into $\cY$-variables.
	Of course, an $\cX$-reduction could increment the number of $\cY$-variables by trying to remove constants.
	The trick is to define for every expression an $\cX$-reduction $(r, \sigma)$ such that the size of $\dom(\sigma)$ depends only on $N$ and $|\cX|$.
	
	We define an $\cX$-reduction by induction over expressions by using a function $\xred(\cdot)$. 
	For this definition we assume that $\cY$ is possibly infinite. Later, in Lemma \ref{lemma:reduction-variables}, we show that we only need a finite number of variables to define $\xred(\cdot)$. 
	For the base case, if $e = x \in \cX$, then we define $\xred(e) = (x, \sigma_\emptyset)$ where $\sigma_\emptyset$ is the empty function. Otherwise, if $e = s \in \cS$ or $e = y \in \cY$  then we choose a fresh variable $y' \in \cY$ and define $\xred(e) = (y', \sigma)$, where $\sigma(y') = e$ and $\dom(\sigma) = \{y'\}$.
	For the inductive step, suppose that $e$ is of the form:
	\begin{equation} \label{eq:xred-form}
	e \; = \;  \otimes\{e_1, \ldots, e_n, f_1, \ldots, f_m\}
	\end{equation}
	where each $e_i \in \expr(\cX \cup \cY)$ contains at least one variable in $\cX$ and each $f_i \in \expr(\cY)$ contains no variables in $\cX$.  
	Furthermore, suppose that $\xred(e_1) = (r_1, \sigma_1), \ldots, \xred(e_n) = (r_n, \sigma_n)$ are already defined and $\dom(\sigma_i) \cap \dom(\sigma_j) = \emptyset$ (without lost of generality, we can relabel the variables). 
	Then we define $\xred(e) = (r, \sigma)$ recursively as follows:
	\begin{eqnarray}
	r & = & \otimes\{r_1, \ldots, r_n, y\} \label{eq:def-x0-reduction} \\
	\sigma & = & \big(\sigma_1 \cup \ldots \cup \sigma_n\big)[y \rightarrow \otimes\{f_1, \ldots, f_m\}] \nonumber
	\end{eqnarray}
	where $y \in \cY$ is a fresh variable not used in $\sigma_1, \ldots, \sigma_n$ and $\big(\sigma_1 \cup \ldots \cup \sigma_n\big)$.
	If $m=0$, then we do not add the additional variable~$y$.
	
	In the recursive definition of $\xred(e)$, the subexpression $f_1, \ldots, f_m$ that do not use $\cX$ registers are replaced by a new fresh variable $y$ and its content $\otimes\{f_1, \ldots, f_m\}$ is assigned into $y$.
	It is clear from the definition that $(r, \sigma)$ is an $\cX$-reduction for $e$.
	Since $\sigma_1, \ldots, \sigma_n$ have disjoint domains and $y$ was chosen as a fresh variable then $r$ is copyless and the domain of $\sigma$ is equal to $\bigcup_i \dom(\sigma_i) \cup \{y\} = \var(r) \cap \cY$.
	For example, recall the expression $e = \add\{ \mult\{x,y,2\}, 3, \mult\{z,4\}\}$ in (\ref{ex:tree-formula}). Suppose that $x \in \cX$ and $y,z \in \cY$. Then the $\cX$-reduction $\xred(e) = (r, \sigma)$ is equal to:
	\[
	\begin{array}[t]{ll}
		r \; := \; \add\big\{ \mult\{x,u\},\, v \, \big\} \;\;\;\; & 
		\begin{array}[t]{rrcl}
			\sigma: \;\;  & u & := & y \mult 2 \\
			& v & := & 3 \add (z \mult 4)
		\end{array}
	\end{array}
	\]	
	where $u$ and $v$ are fresh variables in $\cY$. One can easily check that $e = \hat{\sigma}(r)$ and, thus, the expression defined by $e$ is preserved in $(r, \sigma)$. In the next lemma, we show that we need at most $|\cX| \cdot N + 1$ fresh variables from $\cY$ for the above definition of $\xred(\cdot)$.
	
	\begin{lemma} \label{lemma:reduction-variables}
		Let $e \in \expr(\cX \cup \cY)$ be  copyless and with alternation bounded by $N$. If $\xred(e) = (r, \sigma)$, then $r$ is copyless with alternation bounded by $N$, $\sigma$ is copyless with respect to $\dom(\sigma)$, and $|\var(r) \cap \cY| \leq |\cX| \cdot N + 1$.
	\end{lemma}
	
	\begin{proof}
		Suppose that $\xred(e) = (r, \sigma)$.
		By definition it is straightforward to check that $r$ is copyless and has alternation bounded by $N$.
		In fact, each time that a subexpression is replaced, we use a new fresh variable and since $e$ is copyless this proves that $r$ and $\sigma$ are also copyless.
		The most interesting part is to bound the number of $\cY$-variables in $r$. For this, we prove a slightly stronger bound: $|\var(r) \cap \cY| \leq |\var(r) \cap \cX| \cdot N +1$. Since $|\var(r) \cap \cX| \leq |\cX|$ this will prove the lemma.
		
		The proof goes by induction over the alternation of $e$. 
		For the base case we have $|\var(r)| \leq 1$ thus the bound is trivially true.
		For the inductive step suppose that the statement holds for expression with alternation at most $N$ and consider an expression $e$ like (\ref{eq:xred-form}) with alternation $N+1$.
		By definition of $\expr^{\add}(\cX)$ and $\expr^{\add}(\cX)$ each $e_i$ has alternation at most $N$ and the inductive hypothesis applies: $|\var(r_i) \cap \cY| \leq |\var(r_i) \cap \cX| \cdot N +1$ where $(r_i, \sigma_i) = \xred(e_i)$.
		Given that $\var(r_i) \subseteq \var(r)$
		\[
		\sum_{i = 1}^n |\var(r_i) \cap \cY| \;\;\; \leq \;\;\; \sum_{i=1}^n (|\var(r_i) \cap \cX| \cdot N +1) \;\;\; \leq \;\;\; |\var(r) \cap \cX| \cdot N + n.
		\]
		Also, the set $\var(r) \cap \cY$ is partitioned by the sets $\var(r_i) \cap \cY$ and the fresh variable $y$. Therefore, we derive the following bound:
		$$
		|\var(r) \cap \cY| \;\;\; \leq \;\;\; \sum_{i = 1}^n |\var(r_i) \cap \cY| + 1 \;\;\; \leq \;\;\; |\var(r) \cap \cX| \cdot N + n + 1
		$$ 
		Finally, given that $r$ is copyless, we cannot have more subexpressions $e_i$ than variables in $\cX$ and, thus, $n \leq |\var(r) \cap \cX|$ which leads to the desire conclusion:
		$|\var(r) \cap \cX| \cdot N + n + 1 \; \leq \; |\var(r) \cap \cX| \cdot (N+1) +1$.
	\end{proof}

	We can naturally extend the function $\xred(\cdot)$ from $\expr(\cX \cup \cY)$ to $\subs(\cX, \cY)$.
	More precisely, for any $\alpha \in \subs(\cX, \cY)$ define $\xred(\alpha) = (\beta, \sigma_\beta)$  such that $\xred(\alpha(x)) = (\beta(x), \sigma_{x})$ for every $x \in \cX$ and $\sigma_\beta = \bigcup_{x \in \cX} \sigma_x$ (similar to $\cX$-reduction of expressions we can assume that $\dom(\sigma_x) \cap \dom(\sigma_y) = \emptyset$ by relabeling $\sigma_x$ and $\sigma_y$ if necessary). 
	Further, we extend $\xred(\cdot)$ from substitutions to substitution trees such that $\xred(t) = (t', \sigma_{t'})$ where $\nodes(t) = \nodes(t')$, $\xred(t(u)) = (t'(u), \sigma_u)$ for each $u \in \nodes(t)$, and $\sigma_{t'}$ is the disjoint union of all substitutions $\sigma_u$ (again, we assume that domains of $\sigma_u$ are disjoint). 
	The following lemma, similar to Lemma~\ref{lemma:reduction-variables} but for substitution trees, will be useful later in the correctness proof of $\delta'$. We omit the proof since it is straightforward from Lemma~\ref{lemma:reduction-variables} and the copyless restriction.
	
	\begin{lemma} \label{lemma:reduction-correctness}
		For any copyless substitution tree $t \in \trees(\cX, \cY)$, if $\xred(t) = (t', \sigma_{t'})$, then $t'$ is a copyless and constants-free substitution tree and $\sigma_{t'}$ is a copyless partial substitution.
		Furthermore, the number of $\cY$-variables used in $t'$ is bounded by $|\nodes(t)| \cdot |\cX|\cdot(|\cX| \cdot N +1)$.
	\end{lemma}
	Recall that states of $\cA'$ are of the form $(t, B)$ where $t$ is a complete, copyless, and constants-free substitution tree and $B: \tleaves(t) \rightarrow Q$ is an injective function.
	To define the transition $\delta'(t, B) = ((t', B'), \sigma)$ we show how to convert $t$ into $t'$ and how to update $B$ into $B'$ through the composition of four different processes: $\textend$, $\tprune$, $\tshrink$, and $\treduce$.
	In the sequel, we explain each procedure in detail.
	
	\noindent {\bf Extend.} The first step is to extend branches in $(t, B)$ to the next states when reading $a \in \Sigma^*$. We define this process formally by the function $\textend((t, B), a) = (t_1, B_1)$ that receives a state $(t, B) \in Q'$ and a letter $a \in \Sigma^*$ and outputs the pair $(t_1, B_1)$, where $t_1$ is a substitution tree and $B_1: \tleaves(t_1) \rightharpoonup Q$ is a partial injective function. 
	The substitution tree $t_1$ is defined as an extension of $t$ (i.e., $\nodes(t) \subseteq \nodes(t_1)$) such that $t_1(u) = t(u)$ whenever $u \in \nodes(t)$ and for every $v \in \tleaves(t)$, if there exists a transition $(B(v), a, q, \sigma) \in \delta$, then there exists $i \in \nat$ such that  $t_1(v \cdot i) = \sigma$. The function $B_1$ is defined only on the new leaves such that $B_1(v \cdot i) = q$.	
	Intuitively, $\textend((t, B), a) = (t_1, B_1)$ extends $t$ whenever the state on a leaf of $t$ can evolve to a new state by reading~$a$.
	Notice that $\delta$ is not deterministic and a leaf $v \in \tleaves(t)$ could be extended with more than one nodes. Since trees are unordered, $\textend$ is a deterministic procedure.
	
	\noindent {\bf Prune.} The problem with $(t_1, B_1)$ is that there could exist leaves in $t_1$ that are not marked by the function $B_1$ and therefore $(t_1, B_1) \notin Q'$. This happens when for a leaf $v$ there is no transition $(B(v), a, q, \sigma)$ and this branch of the tree becomes a ``dead run''.
	The purpose of the function $\tprune(t_1, B_1) = (t_2, B_2)$ is to prune branches that are dead and to update $B_1$ into a total function~$B_2$.
	Formally, $t_2$ is a subset of $t_1$ (i.e., $\nodes(t_2) \subseteq \nodes(t_1)$ and $t_2(u) = t_1(u)$ for every $u \in \nodes(t_2)$) such that $u \in \nodes(t_2)$ iff $v \cdot u \in \dom(B_1)$ for some $u \in \nat^*$. 
	In other words, we keep only nodes that are ancestors of leaves that are marked by $B_1$.
	Finally, we define $B_2 = B_1$. Note that $\dom(B_2) = \tleaves(t_2)$ since we did not remove any node from the domain of $B_1$. Moreover, paths from the root to leaves were not modified and, in particular, it holds that $t_2^\collapse(u) = t_1^\collapse(u)$ for every $u \in \tleaves(t_2)$.
	
	\noindent {\bf Shrink.} By adding and removing branches with the procedures $\textend$ and $\tprune$ it could happen that $t_2$ is not a complete tree and $(t_2, B_2) \notin Q'$ (i.e., $t_2$ could contain internal nodes with just one child). 
	These nodes are redundant and they can be easily removed by shrinking the tree.
	For this purpose, we define the procedure $\tshrink(t_2, B_2) = (t_3, B_3)$ recursively.
	We define $\tshrink$ by induction on the depth of $t_2$, maintaining the following properties: $t_3$ is a complete tree; $B_3$ is an injective function from $\tleaves(t_3)$ into $Q$; $range(B_3) = range(B_2)$ and that for every $u \in \tleaves(t_2)$ there is $u' \in \tleaves(t_3)$ such that $B_2(u) = B_3(u')$ and $t_2^\collapse(u) = t_3^\collapse(u')$.
	In particular this means that the size of $t_3$ is bounded by $2 \cdot |Q|$.
	
	For trees that have only one node we define $\tshrink$ as the identity function and the properties are kept trivially. Suppose $t_2 = \sigma\{r_1, \dots, r_n\}$ for some $\sigma \in \subs(\cX, \cY)$. For every $j \in \{1,\dots,n\}$ let $i_j \in \nat$ be the node in $t_2$ corresponding to the root of $r_i$. Then for every leaf $u \in \tleaves(r_j)$ the node $i_j\cdot u$ is a leaf in $t_2$. Moreover $\tleaves(t_2) = \bigcup_j\{i_j\cdot u \mid u \in \tleaves(r_j)\}$. We define the functions $C_j : \tleaves(r_j) \to Q$ by $C_j(u) = B_2(i_j\cdot u)$. Notice that $range(B_2) = \bigcup_j range(C_j)$.
	
	If $n > 1$ then $\tshrink (t_2, B_2) = (\sigma\{r_1', \dots, r_n'\}, B_3)$, where $\tshrink(r_j, C_j) = (r_j', C_j')$ and $B_3(i_j \cdot u) = C_j'(u)$ for every $j$ and $u \in C_j$. By induction the properties are kept in $\tshrink(r_j, C_j)$ for all $j$. Then it is easy to see that they are also kept for $\tshrink (t_2, B_2)$.
	The remaining case is for $n = 1$, for simplicity we skip the indexes and write $t_2 = \sigma\{r\}$ and $C : \tleaves(r) \to Q$. Let $r'$ be a tree such that $\nodes(r') = \nodes(r)$, $r'(u) = r(u)$ for every node $u \neq \epsilon$ and $r'(\epsilon) = \sigma \circ r(\epsilon)$.
	Then we define $\tshrink(t_2, B_2) = \tshrink(r', C)$. That is, the edge between the root of $\sigma\{r\}$ and its unique child $r$ is removed.
	By induction the properties are kept in the step from $(r', C)$ to $(t_3, B_3) = \tshrink((r', C))$. Thus we only have to prove that $range(B_3) = range(B_2)$ and that for every $u \in \tleaves(t_2)$ there is $u' \in \tleaves(t_3)$ such that $t_2^\collapse(u) = t_3^\collapse(u')$. The first property follows from the fact that $range(B_2) = range(C) = range(B_3)$. To prove the second property let $u \in \tleaves(t_2)$. By definition there exists an $i$ such that $u = i \cdot v$ and $B_2(i \cdot v) = C(v)$. Let $|u| = k$ then:
	\setlength{\jot}{9pt}
	\begin{align*}
	t_2^\collapse(u) & \;\; = \;\; t_2(\epsilon) \ \circ \ t_2(u[\cdot, 1])  \ \circ  \ t_2(u[\cdot,2])  \ \circ \  \ldots  \ \circ  \  t_2(u[\cdot, k]) \\
	& \;\; = \;\; \underbrace{\sigma  \  \circ  \  t_2(u[\cdot, 1])}_{r'(\epsilon)}  \  \circ  \  \underbrace{t_2(u[\cdot, 2])}_{r'(v[\cdot,1])}  \  \circ  \  \ldots \ \circ  \  \underbrace{t_2(u[\cdot, k])}_{r'(v[\cdot, k-1])} \\
	& \;\; = \;\; r'^\collapse(v).
	\end{align*}
	By the induction assumption there is $u' \in \tleaves(t_3)$ such that $r'^\collapse(v) = t_3^\collapse(u')$.
	
	\noindent {\bf Reduce.} The pair $(t_3, B_3)$ is almost ready after shrinking single-child internal nodes, except that substitutions inside $t_3$ could have constants (e.g. $\textend(\cdot)$ could have introduced constants). 
	To solve this issue, we apply the $\cX$-reduction procedure $\xred(\cdot)$ to reduce all substitutions in $t_3$.
	Formally, we define the procedure $\treduce(t_3, B_3) = ((t_4, B_4), \sigma)$ where $\xred(t_3) = (t_4, \sigma)$ and $B_3 = B_4$. 
	The function $\xred$ changes only the labels (i.e. substitutions of the nodes) and, in particular, it does not change its tree structure, namely, $\nodes(t_3) = \nodes(t_4)$. 
	Thus $|\nodes(t_4)| = |\nodes(t_3)| \leq 2 \cdot |Q|$.
	By Lemma~\ref{lemma:reduction-correctness}, this implies that the number of fresh $\cY$-variables needed to apply the procedure $\xred(\cdot)$ is at most $|\nodes(t_3)| \cdot |\cX|\cdot(|\cX| \cdot N +1) =  2 \cdot |Q| \cdot |\cX|\cdot(|\cX| \cdot N +1)$ which is exactly the size of $\cY$.
	The remaining issue is that the function $\xred$ is not deterministic since it has to choose ``fresh variables'' to apply the $\cX$-reduction. To make it deterministic we can fix a deterministic choice of every fresh variable inside $\xred(\cdot)$ (e.g. by following an arbitrary total order over $\cY$). With this change the procedure $\treduce$ is deterministic.
	
	With the definitions of the procedures extend, prune, shrink and reduce, we are ready to define formally the transition function $\delta'$ of $\cA'$. Specifically, for every state $(t, B) \in Q'$ and every $a \in \Sigma^*$ we define:
	\[
	\delta'((t, B), a) \; = \; \treduce(\ \tshrink(\ \tprune(\ \textend((t,B),a) \ )\ ) \ ) \; = \; ((t_4, B_4), \sigma)
	\]
	Note that all procedures ($\treduce, \tshrink, \tprune, \textend$) are deterministic so $\delta'$ is also deterministic.
	In the next lemma, we show that the definition of $\delta'$ is correct. 
	
	\begin{lemma}
		\label{lemma:transition}
		For any $(t, B) \in Q'$ and $a \in \Sigma$ let:
		\[
		\treduce(\ \tshrink(\ \tprune(\ \textend((t,B),a) \ )\ ) \ ) \; = \; ((t_4, B_4), \sigma).
		\]
		Then $(t_4, B_4) \in Q'$ and $\sigma$ is a copyless substitution over $\cY$.
	\end{lemma}
	\begin{proof}
		Assume that we have $(t_1, B_1) = \textend((t, B), a)$, $(t_2, B_2) = \tprune(t_1, B_1)$, $(t_3, B_3) = \tshrink(t_2, B_2)$ and $(t_4, B_4) = \treduce(t_3, B_3)$.
		We first check that $t_4$ is a complete, copyless, and constants-free substitution tree in $\trees(\cX, \cY)$. 
		We start from showing that $t_4$ is complete. The tree $t_3$ is a result of $\tprune(t_2)$, which by definition is a complete tree. This proves completeness because $t_4$ has the same structure as $t_3$.
		
		We show that $t_4$ is copyless and constant-free.
		The procedure $\textend(\cdot)$ introduces new variables in $t$ only from $\cX$ in new nodes. We label the new nodes with substitutions $\sigma$ from $\cA$, which by definition are substitutions from $\subs(\cX, \cY)$. Moreover it does not introduce new variables from $\cY$ thus the tree $t_1$ is copyless.
		The procedures $\tprune(\cdot)$ and $\tshrink(\cdot)$ do not introduce new variables in $t_1$, and $t_2$. This shows that the tree $t_3$ is copyless.
		By Lemma~\ref{lemma:reduction-correctness} we conclude that $t_4$ is copyless, constants-free, and $\sigma$ is a copyless substitution.
	\end{proof}

	Now we have all the ingredients to prove Lemma \ref{lemma: branch-collapsing}.
	
	\begin{proof}[Proof of Lemma \ref{lemma: branch-collapsing}]
		The lemma is proved by induction over the size of $w \in \Sigma^*$. 
		For the base case $w = \epsilon$ by definition $q_0' = (t_0, B_0)$ where $t_0$ is a single-node tree labeled with $\sigma_0 \in \subs(\cX, \cY)$ and $B_0(\epsilon) = q_0$, where $\sigma_0(x) = \hat{x}$ for every $x \in \cX$, where $\hat{x}$ is the copy of $x$ in $\hX$.
		The initial function $\nu_0'$ is defined as $\nu_0'(\hat{x}) = \nu_0(x)$ for every  $\hat{x} \in \hX$ and $\nu_0'(y) = \zero$ for every $y \in \cY \ \backslash \ \hX$.
		In this setting
		$$
		\xi \circ t^\collapse(u)(x) = \nu_0' \circ \sigma_0(x) = \nu_0'(\hat{x}) = \nu_0(x).
		$$
		In the initial configuration there is only one run which ends in $q_0$ and $range(B_0) = \{q_0\}$, which finishes the proof of the base case.
		
		For the induction step, assume that the lemma holds for $w \in \Sigma^*$ and we show that it also holds for $w \cdot a$ for any $a \in \Sigma$. Let $(q', \nu')$ be a configuration reached by a run of $\cA$ over $w\cdot a$. By Lemma \ref{lemma-unambiguous}, there is a unique run of $\cA$ on $w$ that ends in a configuration $(q, \nu)$ such that $(q,a,q',\sigma) \in \delta$ and $\nu' = \nu \circ \sigma$. By the induction assumption we have that there is a run of $\cA'$ on $w$ reaching configuration $((t, B), \xi)$ such that
		$Q(w) = range(B)$ and for every $x \in \cX$
		that there exists $u \in \nodes(t)$ such that $B(u) = q$ and
		\begin{eqnarray}
		\nu & = &  \xi \circ t^\collapse(u). \label{eq:collapsing-1}
		\end{eqnarray}
		Fix $(t_1, B_1) = \textend((t, B), a)$, $(t_2, B_2) = \tprune(t_1, B_1)$, $(t_3, B_3) = \tshrink(t_2, B_2)$ and $(t_4, B_4, \tau) = \treduce(t_3, B_3)$.
		By definition of $\textend$ the set $range(B_1)$ is the set of all $p$ such that $(B(v), a, p, \sigma_p) \in \delta$ for some $\sigma_p$ and $v \in \tleaves(t)$. Since $Q(w) = range(B)$ then $Q(w \cdot a) = range(B_1)$. Clearly by definition of the procedures $\tprune, \tshrink, \treduce$ we have $range(B_1) = range(B_2) = range(B_3) = range(B_4)$. By the unambiguity of $\cA$ we get that $|range(B_4) \cap F| = |Q(w \cdot a) \cap F| \leq 1$. Similarly to prove that $B_4$ is injective it suffices to show that $B_1$ is injective. Suppose contrary that $B_1$ maps two different leaves to the same state $p$. Then by definition of $B_1$ there are at least two runs of $\cA$ on $w \cdot a$ that end in $q$, which contradicts Lemma~\ref{lemma-unambiguous}.
		We conclude by~Lemma \ref{lemma:transition} that $\delta'((t,B), a) = (t_4, B_4)$ and $\cA'$ is in configuration $((t_4, B_4), \xi \circ \tau)$ after reading $w \cdot a$. Since $range(B_4) = Q(w \cdot a)$ then there exists $i \in \nat$ such that $ui \in \tleaves(t_1)$, $B_1(u \cdot i) = q'$, and $t_1(ui) = \sigma$. 
		If we compose both sides of~\eqref{eq:collapsing-1} with $\sigma$ then we get:
		\[
		\underbrace{\nu \circ \sigma}_{\nu'} \; = \; \xi \circ \underbrace{t^\collapse(u) \circ \sigma}_{t_1^\collapse(u \cdot i)} \;\;\;\; \text{and} \;\;\;\; \nu' = \xi \circ t_1^\collapse(u \cdot i).
		\]
		We know that the procedure $\tprune(\cdot)$ and $\tshrink(\cdot)$ preserve the outputs of the collapse operation. Thus there exists $v \in \nodes(t_3)$ such that $B_3(v) = B'(v) = q'$ and
		$
		\nu' \; = \; \xi \circ t_3^\collapse(v).
		$
		Given that $(t_4, B_4, \tau) = \treduce(t_3, B_3)$ is an $\cX$-reduction, we know that $\tau \circ t_4^\collapse(v) = t_3^\collapse(v)$.
		By replacing the last equation in the above formula, we get
		$
		\nu' \; = \; \xi \circ \tau \circ t_4^\collapse(v)
		$,
		which concludes the induction step.
	\end{proof}
	
	Lemma \ref{lemma: branch-collapsing} in particular proves that $\delta'$ is a total function when $Q'$ is restricted to the reachable states $Q_r'$.
	For this reason we restrict the set of states to $Q_r'$. With $Q_r'$ as the set of states the automaton $\cA'$ is a deterministic CRA.
	
	To conclude the proof, we show that $\cA$ and $\cA'$ have the same output on every word $w$. Let $(q, \nu)$ be the configuration of the unique accepting run of $\cA$ on $w$ and let $((t, B), \xi)$ be the configuration of $\cA'$ of the run on $w$. By Lemma \ref{lemma: branch-collapsing} there exists a node $u$ such that
	$u \in \tleaves(t)$ such that $B(u) = q$ and
	$
	\nu = \xi \circ t^\collapse(u).
	$
	The output of $\cA'$ on $w$ is defined as $\xi \circ \mu'(t,B)$. By definition $q \in F$ and $|range(B) \cap F| = 1$ thus $q$ is the unique accepting state in $\tleaves(t)$. By definition $\mu'(t,B)  = t^\collapse(u) \circ \mu(q)$. Thus we get the following equalities for the output of $\cA'$ on $w$:
	\begin{align*}
	\xi \circ \mu'(t,B) = \xi \circ t^\collapse(u) \circ \mu(q) = \nu  \circ \mu(q).
	\end{align*}
	This finishes the proof given that $\nu \circ \mu(q)$ is exactly the output of $\cA$ on $w$.
\end{proof}

\subsection{Closure under regular look-ahead} Our next extension of the CRA model is based on regular look-ahead, namely, transitions that can check regular properties over the input. 
Regular look-ahead has been extensively studied for finite automata~\cite{engelfriet1976top,engelfriet1999macro} and has been stated as a key property of a model for computing non-boolean functions~\cite{alur2012streaming,alur2013regular}.
Let $\REG$ be the set of all regular languages over $\Sigma$.
A CRA with regular look-ahead (CRA-RLA) is a tuple $\cA \; = \; (Q, \AL, \cX, \Delta, q_0, \nu_0, \mu)$ where $Q$, $\AL$, $\cX$, $q_0$, $\nu_0$, and $\mu$ are defined as before and $\Delta: Q \times \REG \pmap Q \times \subs(\cX)$ is a partial transition function.
Given a string $w = a_1 \ldots a_n \in \AL^*$, the run of $\cA$ over $w$ is a sequence of configurations:
$
(q_0, \nu_0) \:\trans{L_1} \: (q_1, \nu_1) \: \trans{L_2} \: \ldots \: \trans{L_n} \: (q_n, \nu_n) 
$
such that for $1 \leq i \leq n$,  $\Delta(q_{i-1}, L_i) = (q_i, \sigma_i)$, $a_i \ldots a_n \in L_i$ and $\nu_i(x) = \asem{\nu_{i-1} \comp \sigma_i(x)}$ for each $x \in \cX$.
The output of $\cA$ over $w$ is defined as usual, i.e. $\asem{\cA}(w) = \asem{\nu_n \circ \mu(q_n)}$.
To keep determinism, we also restrict $\Delta$ as follows: for a fixed state $q$ let $\Delta(q, L_1) = (q_1, \sigma_1), \Delta(q, L_2) = (q_2, \sigma_2), \dots, \Delta(q, L_k) = (q_k, \sigma_k)$ be all transitions with $q$ in the first coordinate and $L_1,\ldots, L_k \in \REG$. Then the languages $L_1,\ldots, L_k$ are pairwise disjoint (i.e. $L_i \cap L_j = \emptyset$).
Note that $\cA$ is always deterministic, i.e. after reading the remaining suffix the automaton is forced to take at most one available transition.
The restrictions to copyless and bounded alternation CRA-RLA are defined as expected.

\begin{example}\label{ex:reglook}
Consider the function $f : \Sigma^* \to \nat$ from Example~\ref{ex:unambiguous}. This function can be defined by a BAC extended with regular look-ahead in Figure~\ref{fig:reglook}. 
The automaton has two registers $x$ and $y$ keeping the number of all $a$'s and the number of $b$'s in the last block, respectively. When updating register $y$, to verify if the automaton is in the last block of $b$'s it checks whether the suffix (denoted $v$) belongs to $b^*a^*$.
For simplicity, the trivial assignment when $v \not \in b^*a^*$ is omitted.
\end{example}

\begin{figure}
 \centering
 \begin{tikzpicture}
\coordinate (cA) at (0,0);
		
\node[state, draw=white] 		(p0) at (cA) {};
\node[state, right of=p0, node distance=1.8cm] 		(p) {};
\node[state,  below of=p, node distance=1.5cm, draw=white] (pf) {};

\draw (p0) edge[pos=0.01,->] node[above] {$x,y:=0$} (p);
\draw (p) edge[pos=0.8,->] node[left] {$\; \max\{x,y\}$} (pf);

\path (p)   edge	[loop above]	node {
	\renewcommand{\arraystretch}{0.7}
	$  a \;   \begin{array}{|c}
	 x := x+1
	\end{array}$
}  (p)
(p)   edge	[loop right]	node {
	\renewcommand{\arraystretch}{0.7}
	$\; b, \text{ if } v \in b^*a^* \;
	\begin{array}{|c}
	\;\; y := y+1
	\end{array}$}  (p); 
\end{tikzpicture}
 \caption{BAC extended with regular look-ahead recognizing $f$ from Example~\ref{ex:reglook}.}
 \label{fig:reglook}
\end{figure}
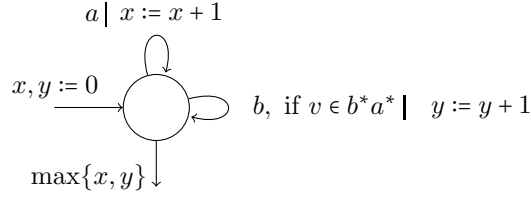

Like for unambiguous copyless CRA, we do not know if extending copyless CRA with regular look-ahead results in a more expressive model. Assuming bounded alternation we prove that the resulting class of functions is the same.

\begin{theorem} \label{theo:regular-lookahead}
	For every BAC $\cA$ with regular look-ahead there exists a BAC $\cA'$ without regular look-ahead that computes the same function, that is, $\asem{\cA}(w) = \asem{\cA'}(w)$ for every~$w \in~\Sigma^*$.
	Furthermore, the number of states and the number of registers of $\cA'$ is double-exponential and polynomial, respectively, in the size of $\cA$.
\end{theorem}

\begin{proof}
	Let $\cA = (Q, \AL, \cX, \Delta, q_0, \nu_0, \mu)$ be a BAC with regular look-ahead where $\Delta: Q \times \REG \pmap Q \times \subs(\cX)$ is the partial transition function (recall that $\Delta$ has finite domain). 
	To represent $\Delta$ finitely, let $L_1, \ldots, L_N$ be all the regular languages in the finite domain of $\Delta$ and, for each $i \leq N$, let $\cA_i = (P_i, \Sigma, \delta_i, p_i^0, F_i)$ be a finite state automaton such that $L_i = \cL(\cA_i)$.
	To unify the structure of each $\cA_i$, define the set of states $P = \biguplus_{i=1}^N P_i$, the transition function $\delta = \biguplus_{i=1}^N \delta_i$, and the set of final states $F = \biguplus_{i=1}^N, F_i$, that is, $P$, $\delta$ and $F$ are the disjoint union of states, transitions and final states, respectively, of the automata $\cA_1, \ldots, \cA_N$.
	Then for every $i \leq N$, define the automaton $\cR_i = (P, \Sigma, \delta, p_i^0, F)$. By definition of $P$, $\delta$ and $F$, one can easily see that $L_i = \cL(\cR_i)$ for every $i\leq N$.
	
	We are ready to show an unambiguous BAC equivalent to $\cA$. By Theorem~\ref{theo:unambiguous-theorem} this will show that $\cA$ can be defined by a deterministic BAC. 
	Let $\cA' = (Q', \AL, \cX, \Delta', I_0', \nu_0, F', \mu')$ be a BAC such that:
	\begin{itemize}
		\item $Q' = Q \times 2^P$ is the set of states,
		\item $\Delta' \subseteq Q' \times \Sigma \times Q' \times \subs(\cX)$ where $((q, S), a, (q', S'), \sigma) \in \Delta'$ iff there exist $(q, L_i, q', \sigma) \in \Delta$ for some $i \leq N$, 
		and a surjective function $f: S \cup \{p^0_i\} \rightarrow S'$ such that $\delta(s, a) = f(s)$ for every $s \in S \cup \{p^0_i\}$,
		\item $I_0' = \{(q_0, \emptyset)\}$, 
		\item $F' = \{(q, S) \mid S \subseteq F\}$, and
		\item $\mu':Q' \rightarrow \expr(\cX)$ where $\mu'(q, S) = \mu(q)$ for every $(q, S) \in Q'$.
	\end{itemize}
	The idea behind $\cA'$ is to guess, at each letter, all transitions in $\cA$ that are satisfied by the remaining suffix.
	In each state $(q, S)$ of $\cA'$ we keep the current state $q$ of a run and a subset $S$ of $P$ that includes all regular transitions that have been taken so far by the run on $q$. 
	Then we have a transition $((q, S), a, (q', S'), \sigma) \in \Delta'$ if there exist a transition $(q, L_i, q', \sigma) \in \Delta$ (i.e. a transition from $q$ to $q'$) such that we can extend each state in $S \cup \{p_i^0\}$ to a state in $S'$. Note that, in order to start simulating the finite automaton $\cR_i$ over the suffix, $p^0_i$ is also included on the set of states that are updated. 
	Finally, if the last state $(q, S)$ of a run satisfies $S \subseteq F$, then we know that all suffixes during the run satisfies the regular look-ahead of the transitions and, then, the state $(q, S)$ is final.
	
	We show first that $\cA'$ is unambiguous. By contradiction, suppose that $\cA'$ is not unambiguous, that is, there exist $w =  a_1 \ldots a_n  \in \Sigma^*$ and two different accepting runs $\rho$ and $\rho'$ of $\cA'$ over $w$. 
	Let $i \leq n$ be the least position such that $\rho(i) = \rho'(i) = (q, S)$ but $\rho(i+1) = (q_1, S_1) \neq (q_2, S_2) = \rho'(i+1)$.
	We know that this position exists since, by construction, it holds that $\rho(0) = \rho'(0)$.
	Let $(q, L, q_1, \sigma_1) \in \Delta$ and $(q, L', q_2, \sigma_2) \in \Delta$ be the transitions that witness the transitions  $((q,S), a_{i+1}, (q_1, S_1), \sigma_1) \in \Delta'$ and $((q,S), a_{i+1}, (q_2, S_2), \sigma_2) \in \Delta'$.
	Since both runs are accepting, then it is straightforward to show by induction that $w[i, \cdot] \in L$ and $w[i, \cdot] \in L'$ for the suffix $w[i, \cdot]$ of $w$ at position $i$.
	Then we have a contradiction since, by definition of CRA-RLA, we know that $L \cap L'= \emptyset$.
	We conclude that $\cA'$ must be unambiguous. 
	Note that during the construction we did not change the assignments $\sigma$ in the transitions.
	For this reason, we can also conclude that $\cA'$ is a copyless CRA with bounded alternation.
	
	For the last part of the proof, we have to show that $\asem{\cA}(w) = \asem{\cA'}(w)$ for every $w \in \Sigma^*$.
	It is easy to verify that for every run 
	$(q_0, \nu_0) \:\trans{L_1} \: \ldots \: \trans{L_n} \: (q_n, \nu_n)$ of $\cA$ over $w \in \Sigma^*$, there exist sets $S_i$ such that the sequence
	$((q_0, S_0), \nu_0) \:\trans{a_1} \: \ldots \: \trans{a_n} \: ((q_n, S_n), \nu_n)$ is an accepting run of $\cA'$ over $w$. For a given word $w$ the sets $S_i$ are uniquely determined by the transitions $\Delta$ and $\delta$.
	Thus, $\asem{\cA}(w) = \asem{\hat{\nu}_n(\mu(q_n))} = \asem{\cA'}(w)$ for every $w \in \Sigma^*$.
\end{proof}

\subsection{Closure under reverse} We finish this section proving that, in contrast to copyless CRA (see~Section~\ref{sec:nonexpressibility}), BAC are closed under reverse.
Recall that a subclass $\mathcal{C}$ of CRA is closed under reverse if for every $\cA \in \mathcal{C}$ there exists $\cA' \in \mathcal{C}$ such that $\asem{\cA}(w) = \asem{\cA'}(w^r)$ for every~$w \in \Sigma^*$. First, we show an example that sometimes it is more convenient to define the reverse of the function.

\begin{example}\label{ex:reverse}
Consider the function $f : \Sigma^* \to \nat$ from Example~\ref{ex:unambiguous}. We define the function $f^r$ by a BAC in Figure~\ref{fig:reverse}. 
The automaton has two registers $x$ and $y$ keeping the number of all $a$'s and the number of $b$'s in the last block, respectively. By adding extra two states the automaton updates the register $y$ only in the first block of $b$'s, and then it keeps its value.
\end{example}

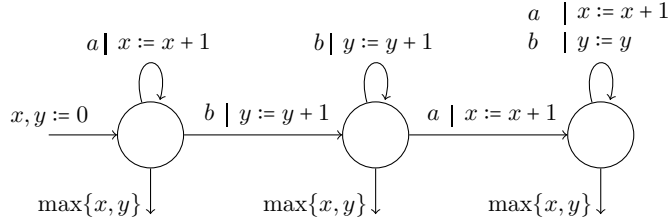
\begin{figure}
 \centering
 \begin{tikzpicture}
 \coordinate (cA) at (0,0);
		
\node[state, draw=white] 		(p0) at (cA) {};
\node[state, right of=p0, node distance=1.8cm] 		(p) {};
\node[state,  below of=p, node distance=1.5cm, draw=white] (pf) {};

\node[state, draw=white, right = 2.1cm of p0] 		(1p0){};
\node[state, right of=1p0, node distance=1.8cm] 		(1p) {};
\node[state,  below of=1p, node distance=1.5cm, draw=white] (1pf) {};

\node[state, draw=white, right = 2.1cm of 1p0] 		(2p0){};
\node[state, right of=2p0, node distance=1.8cm] 		(2p) {};
\node[state,  below of=2p, node distance=1.5cm, draw=white] (2pf) {};

\draw (p0) edge[pos=0.01,->] node[above,scale=0.9] {$x,y:=0$} (p);
\draw (p) edge[pos=0.8,->] node[left,scale=0.9] {$\; \max\{x,y\}$} (pf);
\draw (1p) edge[pos=0.8,->] node[left,scale=0.9] {$\; \max\{x,y\}$} (1pf);
\draw (2p) edge[pos=0.8,->] node[left,scale=0.9] {$\; \max\{x,y\}$} (2pf);

\path (p)   edge	[loop above]	node[scale=0.9] {
	\renewcommand{\arraystretch}{0.7}
	$  a \;   \begin{array}{|c}
	 x := x+1
	\end{array}$
}  (p)
(p)   edge	[->]	node[above,scale=0.9] {
	\renewcommand{\arraystretch}{0.7}
	$
	\begin{array}{c|c}
	b& y := y+1
	\end{array}$}  (1p)
	
	(1p)   edge	[loop above]	node[scale=0.9] {
	\renewcommand{\arraystretch}{0.7}
	$  b \;  \begin{array}{|c}
	 y := y+1
	\end{array}$
}  (1p)
(1p)   edge	[->]	node[above,scale=0.9] {
	\renewcommand{\arraystretch}{0.7}
	$
	\begin{array}{c|c}
	a & x := x+1
	\end{array}$}  (2p)
	(2p)   edge	[loop above]	node[scale=0.9] {
	\renewcommand{\arraystretch}{0.7}
	$ 
	\begin{array}{ll}
		 a \;  & \begin{array}{|c}
		 x := x+1
		 \end{array} \vspace{1mm} \\
		 b \; & \begin{array}{|c}
		 y := y
		 \end{array}
	\end{array}
	$
}  (2p)
	;
\end{tikzpicture}
 \caption{BAC recognizing $f^r$ from Example~\ref{ex:reverse}.}
 \label{fig:reverse}
\end{figure}


\begin{theorem} \label{theo:closed-under-reverse}
	For every BAC $\cA$ there exists a BAC $\cA'$ that computes the reverse function of $\cA$, that is, $\asem{\cA}(w) = \asem{\cA'}(w^r)$ for every $w \in \Sigma^*$.
	Furthermore, the number of states is double-exponential and the number of registers is polynomial in the size~of~$\cA$.
\end{theorem} 



\begin{proof}
	Let $\cA \; = \; (Q, \AL, \cX, \Delta, q_0, \nu_0, \mu)$ be a (deterministic) BAC with alternation bounded by~$N$. 
	The main idea of this construction is to run $\cA$ backwardly over an input. 
	For this purpose, we start from any state in $Q$ where $\cA$ could have potentially finished and follow the transitions in the other direction. 
	Of course, by starting from any state and going back there are many non-deterministic choices that the reverse automaton must take to find the run of $\cA$ over the input. 
	Fortunately, $\cA$ is a deterministic automaton which implies that there is at most one run that starts in $q_0$ and ends in some state in $Q$. 
	That is, the reverse automaton constructed from $\cA$ will be unambiguous and, by Theorem~\ref{theo:unambiguous-theorem}, we know that there is an equivalent deterministic BAC cost-register automaton. 
	The only remaining issue is to construct the output expression of $\cA$ over the input from the back. 
	For this, we exploit the bounded alternation of $\cA$ and $\cX$-reductions (see below) to extend and reduce expressions starting from the final output. 
	
	Let $\cY$ be a non-empty set of variables such that $\cX \cap \cY = \emptyset$. We recall here the idea of an $\cX$-reduction introduced in the proof of Theorem~\ref{theo:unambiguous-theorem}.
	An $\cX$-reduction for $e \in \expr(\cX \cup \cY)$ is a pair $\xred(e) = (r, \sigma)$ where $r \in \expr(\cX \cup \cY)$ is an expression without constants and $\sigma: \cY \to \expr(\cY)$ is a substitution such that $e = \hat{\sigma}(r)$. Moreover, $r$ is copyless whenever $e$ is copyless. 
	By Lemma~\ref{lemma:reduction-variables}, we know that if $e$ has alternation bounded by $N$, then $r$ has alternation bounded by $N$, $\sigma$ is a copyless substitution, and the number of $\cY$-variables in $r$ is bounded by $|\cX| \cdot N + 1$.
	
	We have now all the ingredients to define a reverse BAC cost register automaton from $\cA$. Formally, we construct an unambiguous BAC $\cA' = (Q', \AL, \cY, \Delta', I_0', V_0', F', \mu')$ from $\cA$ as follows.
	\begin{itemize} \itemsep2mm
		
		\item $Q'$ is the set of all pairs $(q, r)$ where $q \in Q$ and $r \in \expr(\cX \cup \cY)$ is a copyless and constants-free expression, in particular, the size of $r$ is bounded by $|\cX| + |\cY|$.
		
		\item $\cY$ is a set of registers such that $\cX \cap \cY = \emptyset$ and $|\cY| = |\cX| \cdot N + 1$.
		
		\item $I_0'$ is the set of all pairs $(q, r)$ such that $\xred(\mu(q)) = (r, \sigma)$ for some substitution~$\sigma$.
		
		\item $V_0': I_0' \rightarrow \val(\cY)$ is the initial substitution function such that $V_0'((q, r)) = \sigma$ for every $(q, r) \in I_0'$ and $\xred(\mu(q)) = (r, \sigma)$. 
		
		\item $F'$ is the set of all pairs $(q_0, r) \in Q'$ such that $q_0$ is the initial state of $\cA$. 
		
		\item $\mu'$ is the final substitution such that 
		$\mu'((q,r)) = \nu_0 \circ r$ for every $(q,r) \in Q'$. 
		
		\item $\Delta'$ is the transition relation where $((q,r), a, (q',r'), \sigma') \in \Delta'$ if, and only if, $\delta(q', a) = (q, \sigma)$ and $\xred(\sigma \circ r) = (r', \sigma')$.
	\end{itemize}
	We start by explaining the definition of the initial substitution function $V_0'$ and the transition relation $\Delta'$. 
	First, note that $V_0'$ is well defined, namely, $V_0'((q, r)) = \sigma$ is a valuation over $\cY$ for every $(q, r) \in I_0$ satisfying $\xred(\mu(q)) = (r, \sigma)$.
	Indeed, the expressions $\mu(q)$ contains $\cX$-variables and constants which implies
	that $\sigma$ is a substitution over $\cY$, where $\sigma(y)$ is a ground expression for every $y \in \dom(\sigma)$. 
	We consider $\sigma(y)$ as a constant because as an expression without variables it can be evaluated into a constant. 
	Regarding the transition relation $\Delta'$, notice that for $((q,r), a, (q',r'), \sigma') \in \Delta'$ we are following the transition $\delta(q', a) = (q, \sigma)$ backwardly and ``extending and reducing'' $r$ by applying $\sigma$.
	This is similar to the $\textend$ function used in the transition function of Theorem~\ref{theo:unambiguous-theorem}.
	The result of this ``extension'' is reduced by the $\xred$ procedure into an expression that has bounded size by Lemma~\ref{lemma:reduction-variables}.
	In other words, we are constructing the output of $\cA$ over the input backwardly by extending the final output $\mu(q)$, reducing the tree whenever it is possible and storing this in the states.  
	Finally, one can easily show that $\cA'$ is unambiguous given that for each transition $((q,r), a, (q',r'), \sigma') \in \Delta'$ we are following the transition $\delta(q', a) = (q, \sigma)$ backwardly and given $(q,t)$, $a$, and $q'$ one can easily check by definition that $r'$ and $\sigma'$ are uniquely determined.
	
	It is left to show that for every word $w = a_1 a_2 \ldots a_n \in \Sigma^*$ we have $\asem{\cA}(a_1 a_2 \ldots a_n) = \asem{\cA'}(a_n a_{n-1} \ldots a_1)$. 
	Let 
	$\rho: \, (q_0, \nu_0) \:\trans{a_1} \: \ldots \: \trans{a_n} \: (q_n, \nu_n)$
	be a run of $\cA$ over $w$ such that, for every $1 \leq i \leq n$,  $\delta(q_{i-1}, a_i) = (q_i, \sigma_i)$ and $\nu_i(x) = \asem{\nu_{i-1} \comp \sigma_i(x)}$ for each $x \in \cX$.
	Recall that the output of $\cA$ over $w$ is defined by $\asem{\cA}(w) = \asem{\nu_n \circ \mu(q_n)}$.
	By the construction of $\cA'$ and $\Delta'$, let
	$$
	\rho': \ ((q_n, r_n), \chi_n) \:\trans{a_n} \: ((q_{n-1}, r_{n-1}), \chi_{n-1}) \: \trans{a_{n-1}} \: \ldots \: \trans{a_1} \: ((q_0, r_0), \chi_0)
	$$
	be a run of $\cA'$ over $a_n a_{n-1} \ldots a_1$ such that $\xred(\mu(q_n)) = (r_n, \chi_n)$ and, for every $1 \leq i \leq n$, it holds that $((q_{i},r_{i}), a_{i}, (q_{i-1},r_{i-1}), \tau_{i}) \in \Delta'$ where $\xred(\sigma_{i} \circ r_{i}) = (r_{i-1}, \tau_i)$ and $\chi_{i-1}(y) =  \asem{\chi_{i} \comp \tau_i(y)}$ for each $y \in \cY$.
	One can easily check that $\rho'$ is the only run of $\cA'$ over $a_n a_{n-1} \ldots a_1$.
	Indeed, $r_n$ is determined by $q_n$ and each pair $(q_{i-1}, r_{i-1})$ is determined by $q_i$, $r_i$ and $\rho$.
	Therefore, it is left to show that the output of $\rho'$ is equal to the output of $\rho$. 
	For this, we prove by backward induction (i.e. starting from $n$ and ending in $0$) that:
	\begin{eqnarray} \label{eq:reverse-induction}
	\asem{\nu_{i} \circ \chi_{i} \circ r_i} = \asem{\cA}(w[1,i]).
	\end{eqnarray}
	This concludes the proof, since $\nu_i$ and $\chi_i$ are valuations over disjoint sets of variables
	$$
	\asem{\cA}(w) = \asem{\nu_{0} \circ \chi_{0} \circ r_0} = \asem{\chi_{0} \circ \nu_{0} \circ r_0} = \asem{\chi_{0} \circ \mu'((q_0, r_0))} = \asem{\cA'}(a_n a_{n-1} \ldots a_1).$$
	To prove the base of induction in~(\ref{eq:reverse-induction}) notice that $\asem{\cA}(w) = \asem{\nu_n \circ \mu(q_n)} = \asem{\nu_n \circ \chi_n \circ r_n}$, since by definition $\chi_n \circ r_n = \mu(q_n)$.
	For the inductive step suppose that~(\ref{eq:reverse-induction}) holds for $i+1$ and we prove it for $i$.
	By definition of $\Delta'$ we have $\xred(\sigma_{i+1} \circ r_{i+1}) = (r_i, \tau_{i+1})$, furthermore
	$$
	\begin{array}{rcll}
	\tau_{i+1} \circ r_i & = & \sigma_{i+1} \circ r_{i+1} & (\text{by definition of $\cX$-reduction}) \\
	\nu_i \circ \tau_{i+1} \circ r_i & = & \nu_i \circ \sigma_{i+1} \circ r_{i+1}  & \\
	\nu_i \circ \tau_{i+1} \circ r_i & = & \nu_{i+1} \circ r_{i+1}  & (\text{by definition of $\nu_{i}$}) \\
	\chi_{i+1} \circ \nu_i \circ \tau_{i+1} \circ r_i & = & \chi_{i+1} \circ \nu_{i+1} \circ r_{i+1}  &  \\
	\nu_{i} \circ \chi_{i} \circ r_i & = & \nu_{i+1} \circ \chi_{i+1} \circ r_{i+1}  &  (\text{by definition of $\chi_{i}$}) \\
	\nu_{i} \circ \chi_{i} \circ r_i & = & \asem{\cA}(w)  & (\text{by inductive hyphotesis}) \\
	\end{array}
	$$
\end{proof}
Notice that, as proved in Theorem~\ref{theorem:counterexample}, this construction fails for the function $f_\cB^R$, where $f_\cB$ is given by the copyless CRA $\cB$. This is because while $\cB$ is copyless its alternation is not bounded. In the above construction when the output is read backwards we encode in the states the tree expression obtained by composing the substitutions from backwards. In the constructions we ensure that one can store such a tree in a succinct way by using the $\cX$-reductions. The $\cX$-reductions deal with constants and registers that are on the same `level' of alternation but it does not lower its alternation level. Since we do not have any reductions to drop the alternation of the tree, keeping such trees would require unbounded memory.

We conclude this section by stressing the robustness of bounded alternation copyless CRA: they are closed under unambiguous non-determinism, regular look-ahead, and reverse operation. Notice that all the structural properties studied in Section~\ref{sec:structure_copyless} also apply to this class, in particular, the results regarding normal form and stable registers.

\section{Conclusions and future work}
\label{sec:conclusions}
In this paper, we studied structural properties, expressiveness and closure properties of copyless CRA. 
In particular, we showed that the class of functions recognized by copyless CRA are not closed under reverse.
Due to this result we proved that copyless CRA are strictly contained in the class of weighted automata.
In~\cite{alur2013regular} Alur et al. introduced the class of \emph{regular cost functions} defined in terms of streaming string-to-tree transducers.
This class contained copyless CRA but it was left open whether this inclusion is strict or not. 
Since the class of regular cost functions is closed under reverse operation, we can conclude that copyless CRA are strictly contained in the class of regular cost functions.

To recover the closure properties of CRA, we proposed the subclass of bounded alternation copyless CRA (BAC).
We prove that BAC are closed under unambiguous non-determinism, regular look-ahead, and reverse.
An open problem here is to show whether these constructions are optimal or not. In particular, the construction to show that BAC are closed under regular look-ahead is double exponential: first constructing an unambiguous BAC; and second by determinizing it. Each step takes exponential time, and one can envision a construction in one step that requires only a single exponential. 
For general copyless CRA, we do not know whether this class is closed under unambiguous non-determinism or regular look-ahead. A positive answer would be surprising, since unambiguous automata are often trivially closed under the reverse operation (e.g. finite automata or weighted automata).


To study whether a subclass of CRA is closed under reverse, one could approach the problem in a more general way. A natural extension of BAC is to enhance the model with the ability of moving in both directions. Our results do not give a straightforward argument that BAC is closed under two-way extension, but we believe that this could be shown by exploiting the copyless restriction and bounded alternation similar as in the proof of Theorem~\ref{theo:regular-lookahead}.

The most important task for future work is to study the decidability properties of copyless CRA or BAC. The classical questions for computational models, like boundedness or equivalence of two automata, remain unanswered. We hope that the machinery developed in Section~\ref{sec:structure_copyless} is a first step towards answering these questions.

\paragraph*{Acknowledgments}
We thank the anonymous referees for their helpful comments.
The first author was partially supported by Poland's National Science Center grant 2013/09/B/ST6/01575.
The last author was supported by CONICYT + PAI / Concurso Nacional Apoyo al Retorno de Investigadores/as desde el extranjero -- Convocatoria 2013 + 821320001 and by the Millenium
Nucleus Center for Semantic Web Research under grant NC120004.

\bibliographystyle{abbrv}
\bibliography{biblio}


\end{document}